\documentclass[format=acmsmall,nonacm]{acmart}

\def\BibTeX{{\rm B\kern-.05em{\sc i\kern-.025em b}\kern-.08emT\kern-.1667em\lower.7ex\hbox{E}\kern-.125emX}}

\setcopyright{rightsretained}
\acmJournal{TALG}
\acmYear{0} \acmVolume{0} \acmNumber{0} \acmArticle{0}
\acmMonth{0} \acmPrice{}\acmDOI{0}

\usepackage[utf8]{inputenc}
\usepackage[T1]{fontenc}
\usepackage[english]{babel}

\usepackage{amsmath,amsfonts,amsthm}
\usepackage{csquotes}
\usepackage{xspace}
\usepackage{hyperref}
\usepackage{color}
\usepackage{graphicx}
\usepackage[labelfont=bf]{caption}
\usepackage{subcaption}
\usepackage[inline]{enumitem}
\usepackage{wrapfig}
\usepackage{mathtools}
\usepackage{varioref}

\usepackage{thmtools}
\declaretheoremstyle[bodyfont = \normalfont]{defstyle} 
\declaretheorem[numberwithin=section,name=Theorem]{thm}
\declaretheorem[sibling=thm,name=Lemma]{lem}
\declaretheorem[sibling=thm,name=Observation]{obs}
\declaretheorem[sibling=thm,name=Proposition]{prop}
\declaretheorem[sibling=thm,name=Corollary]{cor}
\declaretheorem[sibling=thm,style=defstyle,name=Definition]{defn}
\declaretheorem[sibling=thm,style=defstyle,name=Problem]{problem}
\declaretheorem[sibling=thm,style=defstyle,name=Claim]{claim}
\declaretheorem[sibling=thm,style=defstyle,name=Hypothesis]{hypo}
\declaretheorem[sibling=thm,style=defstyle,name=Summary]{summary}

\newcommand{\secref}[1]{Section~\ref{sec:#1}}
\newcommand{\thmref}[1]{Theorem~\ref{thm:#1}}
\newcommand{\lemref}[1]{Lemma~\ref{lem:#1}}
\newcommand{\propref}[1]{Proposition~\ref{prop:#1}}
\newcommand{\obsref}[1]{Observation~\ref{obs:#1}}
\newcommand{\corref}[1]{Corollary~\ref{cor:#1}}
\newcommand{\correfs}[2]{Corollaries~\ref{cor:#1} and~\ref{cor:#2}}
\newcommand{\norm}[1]{\left\lVert#1\right\rVert}

\newcommand{\Oh}{{\mathcal{O}}}
\newcommand{\tOh}{\tilde{\mathcal{O}}}

\newcommand{\Fr}{Fréchet\xspace}
\newcommand{\eps}{\varepsilon}
\newcommand{\RR}{\mathbb{R}}
\newcommand{\ind}{\mathrm{ind}}

\newcommand{\zero}{\mathbf{0}}

\newcommand{\ie}[0]{i.e.\xspace}

\usepackage{algorithm}
\usepackage[noend]{algpseudocode}
\newcommand{\Algblankline}{\item[]}

\usepackage{footnote}
\makesavenoteenv{algorithmic}
\makesavenoteenv{algorithm}

\usepackage{mdframed}
\newenvironment{corbox}{\begin{mdframed}[skipabove=\topsep,skipbelow=\topsep]\begin{cor}}{\end{cor}\end{mdframed}
}
\newenvironment{defnbox}{\begin{mdframed}[skipabove=\topsep,skipbelow=\topsep]\begin{defn}}{\end{defn}\end{mdframed}
}

\newenvironment{sumbox}{\begin{mdframed}[skipabove=\topsep,skipbelow=\topsep]\begin{summary}}{\end{summary}\end{mdframed}
}

\begin{document}
\title[Discrete Fréchet Distance Under Translation: Conditional Hardness and an Improved Algorithm]{Fr{\'e}chet Distance Under Translation: Conditional Hardness and an Algorithm via Offline Dynamic Grid Reachability}
\titlenote{An extended abstract of this paper appeared at \cite{BringmannKN19}.}

\author{Karl Bringmann}
\email{kbringma@mpi-inf.mpg.de}
\authornote{This work is part of the project TIPEA that has received funding from the European Research Council (ERC) under the European Unions Horizon 2020 research and innovation programme (grant agreement No. 850979).}
\affiliation{\institution{Saarland University and Max Planck Institute for Informatics, Saarland Informatics Campus}
  \streetaddress{Campus E1 3}
  \city{Saarbrücken}
  \postcode{66123}
  \country{Germany}
}

\author{Marvin Künnemann}
\email{marvin@mpi-inf.mpg.de}
\affiliation{\institution{Max Planck Institute for Informatics, Saarland Informatics Campus}
  \streetaddress{Campus E1 4}
  \city{Saarbrücken}
  \postcode{66123}
  \country{Germany}
}

\author{André Nusser}
\email{anusser@mpi-inf.mpg.de}
\orcid{0000-0002-6349-869X}
\affiliation{\institution{Max Planck Institute for Informatics and Graduate School of Computer Science, Saarland Informatics Campus}
  \streetaddress{Campus E1 4}
  \city{Saarbrücken}
  \postcode{66123}
  \country{Germany}
}

\renewcommand{\shortauthors}{Bringmann, Künnemann, Nusser}

\begin{abstract}
The discrete Fr\'echet distance is a popular measure for comparing polygonal curves. An important variant is the discrete Fr\'echet distance under translation, which enables detection of similar movement patterns in different spatial domains. For polygonal curves of length $n$ in the plane, the fastest known algorithm runs in time $\tilde{\mathcal{O}}(n^{5})$ \cite{avraham2015faster}. This is achieved by constructing an arrangement of disks of size ${\mathcal{O}}(n^{4})$, and then traversing its faces while updating reachability in a directed grid graph of size $N := {\mathcal{O}}(n^2)$, which can be done in time $\tilde{\mathcal{O}}(\sqrt{N})$ per update \cite{diks2007dynamic}. The contribution of this paper is two-fold.

First, although it is an open problem to solve dynamic reachability in directed grid graphs faster than $\tilde{\mathcal{O}}(\sqrt{N})$, we improve this part of the algorithm: We observe that an offline variant of dynamic $s$-$t$-reachability in directed grid graphs suffices, and we solve this variant in amortized time $\tilde{\mathcal{O}}(N^{1/3})$ per update, resulting in an improved running time of $\tilde{\mathcal{O}}(n^{4.66...})$ for the discrete Fr\'echet distance under translation. Second, we provide evidence that constructing the arrangement of size ${\mathcal{O}}(n^{4})$ is necessary in the worst case, by proving a conditional lower bound of $n^{4 - o(1)}$ on the running time for the discrete Fr\'echet distance under translation, assuming the Strong Exponential Time Hypothesis.

\end{abstract}



\maketitle

\clearpage

\section{Introduction}

\paragraph{\Fr distance.}
Modern tracking devices yield an abundance of movement data, e.g., in the form of GPS trajectories. This data is usually given as a sequence of points in $\RR^d$ for some small dimension $d$ like 2 or 3. By interpolating linearly between consecutive points, we obtain a corresponding polygonal curve. 
One of the most fundamental tasks on such objects is to measure similarity between two curves $\pi,\sigma$. A popular approach is to measure their distance using the \Fr distance, which has two important variants: The classic \emph{continuous \Fr distance} is the minimal length of a leash connecting a dog and its owner as they continuously walk along the interpolated curves $\pi$ and $\sigma$, respectively, from the startpoints to the endpoints without backtracking. In the \emph{discrete \Fr distance}, at any time step the dog and its owner must be at vertices of their curves and may jump to the next vertex. 
This discrete version is well motivated when we think of the inputs as sequences of points rather than polygonal curves, i.e., if the interpolated line segments between input points have no meaning in the underlying application.
In comparison to other similarity measures such as the Hausdorff distance, the \Fr distance considers the ordering of the vertices along the curves, thus reflecting an intuitive property of curve similarity.

The time complexity of the \Fr distance is well understood.
For the continuous \Fr distance, Alt and Godau designed an $\Oh(n^2 \log n)$-time algorithm for polygonal curves $\pi,\sigma$ consisting of $n$~vertices~\cite{AltG95}. Buchin et al.~\cite{BuchinBMM14} improved on this result by giving an algorithm that runs in time $\Oh(n^2 \sqrt{\log n} (\log \log n)^{3/2})$ on the Real RAM and $\Oh(n^2 (\log \log n)^2)$ on the Word RAM.
The first algorithm for the discrete \Fr distance ran in time $\Oh(n^2)$~\cite{EiterM94}, which was later improved to $\Oh\big(n^2 \tfrac{\log \log n}{\log n}\big)$~\cite{AgarwalBAKS13}. 
On the hardness side, conditional on the Strong Exponential Time Hypothesis, Bringmann~\cite{Bringmann14} ruled out $\Oh(n^{2-\eps})$-time algorithms for any $\eps > 0$, for both variants of the \Fr distance. Recently, Abboud and Bringmann~\cite{abboud2018tighter} showed that any $\Oh(n^2 / \log^{17+\eps} n)$-time algorithm for the discrete \Fr distance would prove novel circuit lower bounds. On the practical side, several fast implementations for computing the continuous Fréchet distance resulted from the SIGSPATIAL GIS Cup 2017 \cite{sigspatial1, sigspatial2, sigspatial3} with a follow-up work significantly improving on these results \cite{BKN19}.

Many extensions and variants of the \Fr distance have been studied, e.g., generalizing from curves to other types of objects, replacing the ground space $\RR^d$ by more complex spaces, and many more (see, e.g., \cite{Indyk02,BuchinBW09,AltB10,ChambersETAL10,
Wenk2010geodesic,MaheshwariSSZ11,DriemelHP13,
avraham2015discrete}).
Applications of the \Fr distance range from moving objects analysis (see, e.g.,~\cite{BuchinBGLL11}) through map-matching tracking data (see, e.g.,~\cite{BrakatsoulasPSW05}) to signature verification (see, e.g.,~\cite{MunichP99}).

\paragraph{\Fr distance under translation.}
For some applications, it is useful to change the definition of the \Fr distance slightly. In particular, several applications on curves evolve around the theme of \emph{detecting movement patterns}.
Consider the task of signature verification.
Whether two signatures are similar is a translation-invariant property -- intuitively, by translating a signature in space, we cannot make it more or less similar to another signature. Consider another example: given GPS trajectories of an animal, we might want to detect different movement patterns (just considering their shape) by chopping the trajectories into smaller pieces and clustering these pieces according to some distance measure. For the two applications mentioned above, it is inconvenient that the \Fr distance is not invariant under translation.\footnote{In this context one could even ask for a version of the \Fr distance that is translation- and rotation-invariant, but we focus on the former in this paper.}
In order to overcome this issue, the \emph{\Fr distance under translation} between curves $\pi,\sigma$ is defined as the minimal \Fr distance between $\pi$ and any translation of $\sigma$, i.e., we minimize over all possible translations of $\sigma$. Clearly, this yields a translation-invariant distance measure, and thus enables the above application. 

The \emph{continuous} \Fr distance under translation was independently introduced by Efrat et al.~\cite{efrat2001pattern} and Alt et al.~\cite{alt2001matching}, who designed algorithms in the plane with running time $\tOh(n^{10})$ and $\tOh(n^8)$, respectively\footnote{By $\tOh(\cdot)$ we hide polylogarithmic factors in $n$.}. Both groups of researchers also presented approximation algorithms, e.g., a $(1+\eps)$-approximation running in time $\Oh(n^2 / \eps^2)$ in the plane~\cite{alt2001matching}. This line of work was extended to three dimensions with a running time of $\tOh(n^{11})$~\cite{wenk2002phd}.

The \emph{discrete} \Fr distance under translation was first studied by Jiang et al.~\cite{jiang2008protein} who designed an $\tOh(n^6)$-time algorithm in the plane. Mosig et al.~\cite{mosig2005approximately} presented an approximation algorithm that computes the discrete \Fr distance under translation, rotation, and scaling in the plane, up to a factor close to 2, and runs in time $\Oh(n^4)$.
The best known exact algorithm for the discrete \Fr distance under translation in the plane is due to Ben Avraham et al.~\cite{avraham2015faster}. It is an improvement of the algorithm by Jiang et al.~\cite{jiang2008protein} and runs in time $\tOh(n^5)$.

\paragraph{Our contribution.}
In this paper, we further study the time complexity of the discrete \Fr distance under translation in the plane. First, we improve the running time from $\tOh(n^{5})$ to $\tOh(n^{4.66...})$. This is achieved by designing an improved algorithm for a subroutine of the previously best algorithm, namely offline dynamic $s$-$t$-reachability in directed grid graphs, see Section~\ref{sec:techoverview} below for a more detailed overview.

\begin{thm} \label{thm:mainupper}
  The discrete \Fr distance under translation on curves of length $n$ in the plane can be computed in time $\tOh(n^{14/3}) = \tOh(n^{4.66..})$.
\end{thm}

Our second main result is a lower bound of $n^{4-o(1)}$, conditional on the standard Strong Exponential Time Hypothesis. The Strong Exponential Time Hypothesis essentially asserts that Satisfiability requires time $2^{n-o(n)}$; see Section~\ref{sec:preliminaries} for a definition. This (conditionally) separates the discrete \Fr distance under translation from the classic \Fr distance, which can be computed in time $\tOh(n^2)$. Moreover, the first step of all known algorithms for the discrete \Fr distance under translation is to construct an arrangement of disks of size $\Oh(n^{4})$. Our conditional lower bound shows that this is essentially unavoidable.

\begin{restatable}{thm}{mainlower} \label{thm:mainlower}
  The discrete \Fr distance under translation of curves of length $n$ in the plane requires time $n^{4-o(1)}$, unless the Strong Exponential Time Hypothesis fails.
\end{restatable}

\noindent
We leave closing the gap between $\tOh(n^{4.66..})$ and $n^{4-o(1)}$ as an open problem.

\subsection{Technical Overview} \label{sec:techoverview}

\paragraph{Previous algorithms for the discrete \Fr distance under translation.}
Let us sketch the algorithms by Jiang et al.~\cite{jiang2008protein} and Ben Avraham et al.~\cite{avraham2015faster}. Given sequences $\pi = (\pi_1,\ldots,\pi_n)$ and $\sigma = (\sigma_1,\ldots,\sigma_n)$ in $\RR^2$ and a number $\delta \ge 0$, we want to decide whether the discrete \Fr distance under translation of $\pi$ and $\sigma$ is at most $\delta$. From this decision procedure one can obtain an algorithm to compute the actual distance via standard techniques (i.e., parametric search).

\begin{wrapfigure}{r}{0.2\textwidth}
\centering
\includegraphics[width=0.2\textwidth]{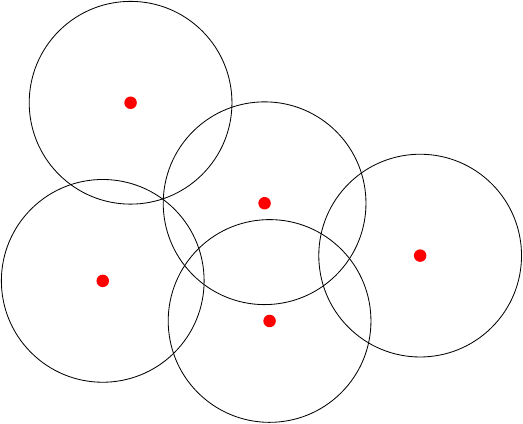}
\end{wrapfigure}
The translations $\tau$ for which the distance of $\pi_i$ and $\sigma_j + \tau$ is at most $\delta$ form a disk in $\RR^2$. Over all pairs $(\pi_i,\sigma_j)$ this yields $\Oh(n^2)$ disks, all of them having radius $\delta$. Construct their arrangement $\mathcal{A}$ (see an illustration to the right), which is guaranteed to have $\Oh(n^4)$ faces. Within each face of $\mathcal{A}$, any two translations are equivalent, in the sense that they leave the same pairs $(\pi_i,\sigma_j)$ in distance at most $\delta$. Thus, whether the discrete \Fr distance is at most $\delta$ is constant in each face. Hence, it suffices to compute the discrete \Fr distance between $\pi$ and $\sigma$ translated by $\tau$ over $\Oh(n^4)$ choices for $\tau$, one for each face of $\mathcal{A}$. Since the discrete \Fr distance can be computed in time $\Oh(n^2)$, this yields an $\Oh(n^6)$-time algorithm, which is essentially the algorithm by Jiang et al.~\cite{jiang2008protein}.

Ben Avraham et al.~\cite{avraham2015faster} improve this algorithm as follows. Denote by $M$ the $n \times n$ matrix with $M_{i,j} = 1$ if the points $\pi_i,\sigma_j$ are in distance at most $\delta$, and $M_{i,j} = 0$ otherwise ($M$ is called the ``free-space diagram''). It is well-known that the discrete \Fr distance of $\pi,\sigma$ is at most $\delta$ if and only if there exists a monotone path from the lower left to the upper right corner of $M$ using only 1-entries.
Equivalently, consider a \emph{directed grid graph} $G_M$ on $n \times n$ vertices, where each node $(i,j)$ has directed edges to $(i+1,j), (i,j+1)$, and $(i+1,j+1)$, and the nodes $(i,j)$ of $G_M$ with $M_{i,j} = 0$ are ``deactivated'' (i.e., removed). Then the discrete \Fr distance of $\pi,\sigma$ is at most $\delta$ if and only if node $(n,n)$ is reachable from node $(1,1)$ in $G_M$. See Figure \ref{fig:freespace_diagram} for an example of a pair of curves, its corresponding free-space diagram $M$, and directed grid graph $G_M$.
\begin{figure}
	\centering
	\includegraphics[width=.9\textwidth]{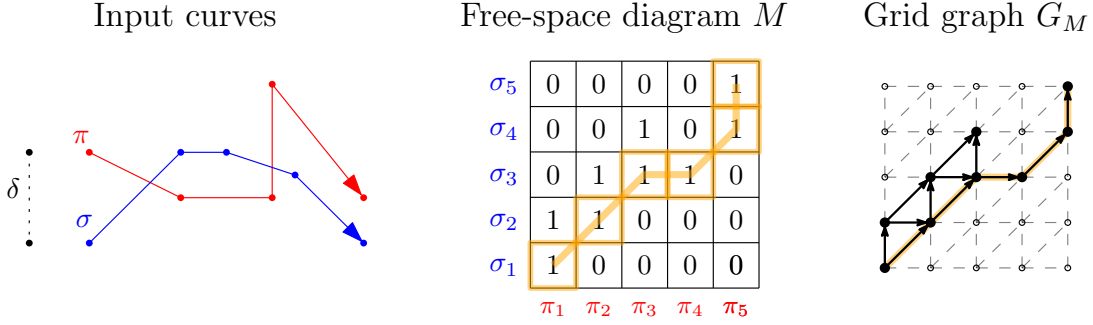}
	\caption{Two input curves $\pi, \sigma$ and a distance $\delta$, the corresponding free-space diagram $M$, and the grid graph $G_M$ corresponding to $M$. A monotone traversal of $M$ and $G_M$ is marked in orange.}
	\label{fig:freespace_diagram}
\end{figure}

Ben Avraham et al. observe that it is easy to construct a sequence of $\Oh(n^4)$ faces $f_1,\ldots,f_L$ of the arrangement~$\mathcal{A}$ such that (1)~each face of $\mathcal{A}$ is visited at least once and (2) $f_\ell$ and $f_{\ell+1}$ are neighboring in $\mathcal{A}$ for all~$\ell$.  Since consecutive faces in this sequence are neighbors, only one pair $(\pi_i,\sigma_j)$ changes its distance, i.e., either $\pi_i,\sigma_j$ are in distance at most $\delta$ in $f_\ell$ and in distance larger than~$\delta$ in $f_{\ell+1}$, or vice versa. This corresponds to one activation or deactivation of a node in $G_M$. After this update, we want to again check whether node $(n,n)$ is reachable from node $(1,1)$ in $G_M$. That is, using a dynamic algorithm for $s$-$t$-reachability in directed grid graphs, we can maintain whether the \Fr distance is at most $\delta$. The best known solution to dynamic reachability in directed $n \times n$ grids runs in time $\tOh(n)$~\cite{diks2007dynamic}.\footnote{This algorithm even works more generally for dynamic reachability in directed planar graphs.} Over all $\Oh(n^4)$ faces, this yields time $\tOh(n^5)$ for the discrete \Fr distance under translation in the plane \cite{avraham2015faster}. 

\paragraph{Intuition.}
There are two parts to the above algorithm: (1) Constructing the arrangement $\mathcal{A}$ and iterating over its faces, and (2) maintaining reachability in the grid graph $G_M$. Both parts could potentially be improved. 

The natural first attempt is to attack the arrangement enumeration, i.e., part (1). The size of the arrangement is $\Oh(n^4)$, and for no other computational problem it is known -- to the best of our knowledge -- that any optimal algorithm must construct such a large arrangement, so this part seems intuitively wasteful. Surprisingly, our conditional lower bound of Theorem~\ref{thm:mainlower} shows that constructing the arrangement is essentially unavoidable.

The remaining part (2) at first sight seems much less likely to be improvable, since it is a well-known open problem to find a faster dynamic algorithm for reachability in directed grid graphs.
Nevertheless, we show how to improve the running time of this part of the algorithm.

\paragraph{Our algorithm.}
We observe that we do not need the full power of dynamic reachability, since we can precompute all $\Oh(n^4)$ updates. This leaves us with the following problem. 

\medskip
\noindent 
\emph{Offline Dynamic Grid Reachability}: 
We start from the directed $n\times n$-grid graph $G$ in which all nodes are deactivated.\footnote{In fact, our algorithm also works in the general case in which nodes can be arbitrarily activated/deactivated in the beginning.}
We are given a sequence of updates $u_1,\ldots,u_U$, where each $u_\ell$ is of the form ``activate node $(i,j)$'' or ``deactivate node $(i,j)$''. The goal is to compute for each $1 \le \ell \le U$ whether node $(1,1)$ can reach node $(n,n)$ in $G$ after performing the updates $u_1,\ldots,u_\ell$.
\medskip

Our main algorithmic contribution is an algorithm for Offline Dynamic Grid Reachability in amortized time $\tOh(n^{2/3})$ per update. This is faster than the update time $\tOh(n)$ obtained by using a dynamic algorithm for reachability in directed planar graphs~\cite{diks2007dynamic}.

\begin{thm} \label{thm:algoreach}
  Offline Dynamic Grid Reachability can be solved in time $\tOh(n^2 + U \cdot n^{2/3})$.
\end{thm}

The high-level approach of this algorithm is to consider all $U$ updates in batches of size at most~$k$, which we call \emph{chunks}. Roughly speaking, we design a grid reachability data structure that given a chunk of $k$ updates $u_1, \dots, u_k$, enables us to (1) for any $1\le j\le k$, answer a grid reachability query in the matrix updated by $u_1,\dots, u_j$ in time $\tOh(k)$ and (2) obtain the data structure for the matrix updated by the complete chunk $u_1,\dots, u_k$ in time $\tOh(n \sqrt{k} + k)$. This way, for each of the $\Oh(U/k)$ chunks, we only need time $\tOh(k^2)$ to answer all $k$ reachability queries for this chunk and time $\tOh(n\sqrt{k} + k)$ to update the data structure for the next chunk, leading to a total time of $\tOh( (U/k) (k^2 + n \sqrt{k})) = \tOh( U (k + n / \sqrt{k}))$. By setting $k \approx n^{2/3}$, we obtain the desired algorithm running in time $\tOh(U n^{2/3})$ after $\tOh(n^2)$ preprocessing.

To obtain our data structure, we build on the reachability data structure of Ben Avraham et al.~\cite{avraham2015faster}, augmented by two crucial insights: How to incorporate a chunk of $k$ updates faster than $k$ single updates, and how to succinctly store reachability information for $k$ distinguished nodes in the grid (coined \emph{terminals}, which correspond to the updates of the next chunk) in the data structure. The latter is given by a surprisingly succinct characterization of reachability of terminals in a grid graph (see Corollary~\ref{cor:reachlabels}), which is the key technical contribution for the algorithm. 

\smallskip
\begin{wrapfigure}{r}{0.25\textwidth}
\centering
\includegraphics[width=0.25\textwidth]{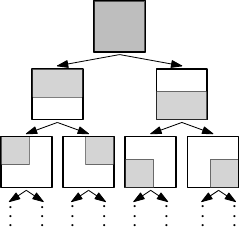}
\end{wrapfigure}
Let us give a more detailed overview of our algorithm and its main ingredients. 
Start with a \emph{block} $[n] \times [n]$ corresponding to the matrix $M$. Repeatedly split every block horizontally in the middle, and then split every block vertically in the middle, until we end up with constant-size blocks.
We call all the blocks considered during this process (not just the constant-size blocks!) the ``canonical'' blocks, see the figure to the right. Ben Avraham et al.~\cite{avraham2015faster} showed that one can store for each canonical block of sidelength $s$ reachability information for each pair of boundary nodes, succinctly represented using only $\tOh(s)$ bits of space, and efficiently computable in time $\tOh(s)$ from the information of the two canonical child-blocks. In particular, over all blocks this information can be maintained in time $\tOh(n)$ per update $u_i$.

\emph{Ingredient 1: Batched updates.} 
The first insight is that we can compute the reachability after a given chunk of $k$ updates $u_1,\ldots, u_k$ faster than $\tOh(nk)$: Intuitively, each update \enquote{touches} roughly $2 \log n$ blocks -- all those that contain the node which is activated or deactivated. Our approach now uses that among the canonical blocks containing an update, the large blocks must be shared by many updates. Specifically, instead of recomputing the reachability information of the large blocks at the top of the hierarchy $k$ times, we perform those updates jointly and thus avoid the runtime of $k$ explicit updates of large blocks. A careful tradeoff yields an update time of $\tOh(n\sqrt{k} + k)$.

\emph{Ingredient 2: Reachability among terminals.} 
Now fix a chunk $C = u_{\ell + 1},\ldots,u_{\ell+k}$ and let $M$ denote the matrix at the beginning of $C$, i.e., after incorporating all updates prior to $u_{\ell+1}$. Denote by $\mathcal{T}$ (``terminals'') the entries that get activated or deactivated during this chunk $C$, and also add $(1,1)$ and $(n,n)$ to the set of terminals. We first deactivate all terminals, obtaining a matrix $M^\zero$ and a corresponding grid graph $G_{M^\zero}$. The basic idea now is to determine for each pair of terminals $t,t' \in \mathcal{T}$ whether $t'$ is reachable from $t$ in $G_{M^\zero}$. 

Let us sketch a simplified algorithm that assumes we have built a graph~$H$ with vertex set $\mathcal{T}$, containing a directed edge $(t,t')$ if and only if $t'$ is reachable from $t$ in $G_{M^\zero}$. To answer the reachability query whether $(n,n)$ is reachable from $(1,1)$ after updating $M$ by $u_{\ell+1}, \dots, u_{\ell+j}$, we proceed as follows:
 For each terminal $t$, activate $t$ in $H$ if and only if $t$ is activated in $M$ updated by $u_{\ell+1}, ..., u_{\ell+j}$. Check whether $(n,n)$ is reachable from $(1,1)$ in $H$. Since $H$ has $O(k)$ nodes and $O(k^2)$ edges, this reachability check can be performed in time $\Oh(k^2)$. (By choosing a chunk size of $k \approx n^{2/5}$ this would result in an $\tOh(U n^{4/5})$ algorithm for Offline Dynamic Grid Reachability, ignoring the preprocessing time.) We will later show how to improve the reachability query time from $\Oh(k^2)$ to $\tOh(k)$ by working directly on the graph $G_{M^\zero}$ instead of constructing the graph $H$. These details are given in the subsequent sections.

\begin{wrapfigure}{r}{0.3\textwidth}
\centering
\includegraphics[width=0.3\textwidth]{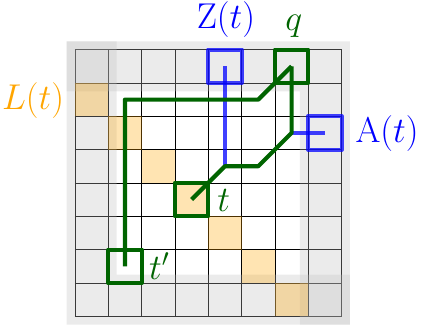}
\end{wrapfigure}
It remains to describe how to determine reachability information among terminals. To this end, we design a surprisingly succinct representation of reachability from terminals to block boundaries.
Consider a canonical block $B$ and let $\mathcal{T}_B$ be the terminals in $B$. For each terminal $t \in \mathcal{T}_B$ let $A(t)$ be the lowest/rightmost point on the right/upper boundary of $B$ that is reachable from $t$, and similarly let $Z(t)$ be the highest/leftmost reachable point, see the illustration to the right. We label any terminal $t = (x,y)$ by $L(t) := x+y$, i.e., the anti-diagonal that $t$ is contained in. For any right/upper boundary point $q$ of $B$, let $\ell(q)$ be the minimal label of any terminal in $\mathcal{T}_B$ from which $q$ is reachable. We prove the following succinct representation of reachability (see Corollary~\ref{cor:reachlabels}) that significantly generalizes a previous characterization for reachability among the \emph{boundaries} of blocks~\cite{AltG95, avraham2015faster}. 

\medskip
\noindent 
\begin{center}
\emph{For any right/upper boundary point $q$ of $B$ and any terminal $t \in \mathcal{T}_B$, \\ $q$ is reachable from $t$ if and only if $q \in [A(t),Z(t)]$ and $\ell(q) \le L(t)$.}
\end{center}
\medskip

Here, $q \in [A(t),Z(t)]$ is to be understood as ``$q$ lies between $A(t)$ and $Z(t)$ along the boundary of $B$'', which can be expressed using a constant number of inequalities. 
The ``only if'' part is immediate, since $t$ can only reach boundary vertices in $[A(t),Z(t)]$, and $\ell(q)$ is the minimal label of any terminal reaching $q$; the ``if'' part is surprising.

Assume we can maintain the information $A(t),Z(t),\ell(q)$. Then using this characterization we can determine all terminals reaching a boundary point $q$ by a single call to \emph{orthogonal range searching}, since we can express the characterization using a constant number of inequalities. A complex extension of this trick allows us to determine reachability among terminals (indeed, this technical overview is missing many details of Section~\ref{sec:reach-ds}). This yields our algorithm, see Sections~\ref{sec:algo} and \ref{sec:reach-ds} for details.

\paragraph{Conditional lower bound.}
Our reduction starts from the $k$-OV problem, which asks for $k$ vectors from $k$ given sets such that in no dimension all vectors are 1. More formally:

\medskip
\noindent 
\emph{$k$-Orthogonal Vectors ($k$-OV)}: 
Given sets $V_1,\ldots,V_k$ of $N$ vectors in $\{0,1\}^D$, are there $v_1 \in V_1,\ldots, v_k \in V_k$ such that for any $j \in [D]$ there exists an $i \in [k]$ with $v_i[j] = 0$?
\medskip

\noindent 
A naive algorithm solves $k$-OV in time $\Oh(N^k D)$. 
It is well-known that the Strong Exponential Time Hypothesis implies that $k$-OV has no $\Oh(N^{k-\eps} \textup{poly}(D))$-time algorithm for all $\eps > 0$ and $k \ge 2$~\cite{Wil05}. 

In our reduction we set $k = 4$. An overview of our construction can be found in Figure \vref{fig:whole_reduction}. We consider \emph{canonical translations} of the form $\tau = (\eps \cdot h_1,\eps \cdot h_2) \in \RR^2$ with $h_1,h_2 \in \{0,\ldots,N^2-1\}$. 
By a simple gadget, we ensure that any translation resulting in a \Fr distance of at most 1 must be close to a canonical translation. For simplicity, here we restrict our attention to exactly the canonical translations.
Note that there are $N^4$ canonical translations, and thus they are in one-to-one correspondence to choices of vectors $(v_1,\ldots,v_4) \in V_1 \times \ldots \times V_4$. In other words, the outermost existential quantifier in the definition of 4-OV corresponds to the existential quantifier over the translation~$\tau$ in the Fréchet distance under translation. 

The next part in the definition of 4-OV is the universal quantifier over all dimensions $j \in [D]$. For this, our constructed curves $\pi,\sigma$ are split into $\pi = \pi^{(1)} \ldots \pi^{(D)}, \sigma = \sigma^{(1)} \ldots \sigma^{(D)}$ such that $\pi^{(i)}, \sigma^{(j)}$ are very far for $i \ne j$. This ensures that the \Fr distance of $\pi,\sigma$ is the maximum over all \Fr distances of $\pi^{(i)},\sigma^{(i)}$, and thus simulates a universal quantifier.

The next part is an existential quantifier over $i \in [k]$. Here we need an OR-gadget for the \Fr distance. Such a construction in principle exists in previous work~\cite{Bringmann14,abboud2018tighter}, however, no previous construction would work with translations, in the sense that a translation in $y$-direction could only decrease the \Fr distance. By constructing a more complex OR-gadget, we avoid this monotonicity.

Finally, we need to implement a check whether the translation $\tau$ corresponds to a particular choice of vectors. We exemplify this with the first dimension of the translation, which we call $\tau_1$, explaining how it corresponds to choosing $(v_1,v_2)$. Let $\ind(v_1),\ind(v_2) \in \{0, \dots, N-1\}$ be the indices of these vectors in their sets $V_1,V_2$, respectively. We want to test whether $\tau_1 = \eps \cdot (\ind(v_1) + \ind(v_2) \cdot N)$. We split this equality into two inequalities. For the inequality $\tau_1 \ge \eps \cdot (\ind(v_1) + \ind(v_2) \cdot N)$, in one curve we place a point at $\pi_1 = (1 + \eps \cdot \ind(v_1), -1-\eta)$, and in the other we place a point at $\sigma_1 = (-1 - \eps \cdot \ind(v_2) \cdot N, -1-\eta)$, for some $\eta > 0$ which we specify later in this work. Then the distance of $\pi_1$ to the translated $\sigma_1$ is essentially their difference in $x$-coordinates, which is $(1 + \eps \cdot \ind(v_1)) - (-1 - \eps \cdot \ind(v_2) \cdot N + \tau_1) = 2 + \eps \cdot (\ind(v_1) + \ind(v_2) \cdot N) - \tau_1$. This is at most~2 if and only if the inequality for $\tau_1$ holds. 
We handle the opposite inequality similarly, and we concatenate the constructed points for both inequalities in order to test equality.

In total, our construction yields curves $\pi,\sigma$ such that their discrete \Fr distance under translation is at most 1 if and only if $V_1,\ldots,V_4$ contain orthogonal vectors. The curves $\pi,\sigma$ consist of $n = \Oh(D \cdot N)$ vertices. Hence, an algorithm for the discrete \Fr distance under translation in time $\Oh(n^{4-\eps})$ would yield an algorithm for 4-OV in time $\Oh(N^{4-\eps} \textup{poly}(D))$, and thus violate the Strong Exponential Time Hypothesis. See Section~\ref{sec:lowerbound} for details.

\subsection{Further related work}
\paragraph{On directed planar/grid graphs.}
In this paper we improve offline dynamic $s$-$t$-reachability in directed grid graphs. The previously best algorithm for this problem came from a more general solution to dynamic reachability in directed planar graphs. For this problem, a solution with $\tOh(N^{2/3})$ update time was given by Subramanian~\cite{subramanian1993fully}, which was later improved to update time $\tOh(\sqrt{N})$ by Diks and Sankowski~\cite{diks2007dynamic}. In particular, our work yields additional motivation to study offline variants of classic dynamic graph problems.

Related work on dynamic directed planar or grid graphs includes, e.g., shortest path computation~\cite{klein1998fully,abraham2012fully,italiano2011improved}, reachability in the decremental setting~\cite{italiano2017decremental}, or computing the transitive closure~\cite{diks2007dynamic}. Recently, the first conditional lower bounds for dynamic problems on planar graphs were shown by Abboud and Dahlgaard~\cite{abboud2016popular}, however, they did not cover dynamic reachability in directed planar graphs.

Other work on directed planar and grid graphs studies, e.g., the minimum amount of space necessary to determine reachability between two nodes in polynomial time~\cite{asano2014,AshidaN18}. For grid graphs this was recently improved from $\tOh(\sqrt{N})$ to $\tOh(N^{1/3})$~\cite{AshidaN18}, but with very different techniques compared to ours.

\paragraph{On related distance measures.}
A related distance measure to the Fréchet distance, which does not take traversals into account, is the Hausdorff distance. The Hausdorff distance can be computed in almost linear time in the plane. As for the Fréchet distance, there are also related variants of the Hausdorff distance that allow rigid motions of one of the input curves. For the Hausdorff distance under translation there exists a cubic-time algorithm \cite{huttenlocher1993upper} for all $L_p$ norms in the plane. For the $L_1$ and $L_\infty$ norm, there even is a quadratic-time algorithm for the Hausdorff distance under translation in the plane \cite{chew_improvements_1992}. Surprisingly, despite this quadratic-time algorithm, the is a cubic lower bound for the size of the arrangement of the Hausdorff distance under translation that holds for all $L_p$ norms \cite{rucklidge_lower_1996}.

The cubic-time algorithm for the Hausdorff distance under translation constructs the upper envelope of Voronoi surfaces and then tests for minima on the vertices and edges of this construction. Thus, this algorithm is also an arrangement-based approach. Similarly, the first algorithms \cite{efrat2001pattern, alt2001matching} for the continuous Fréchet distance under translation as well as the weak continuous Fréchet distance under translation relied on arrangement constructions.

A recently introduced variant of the \Fr distance is the \emph{\Fr gap}~\cite{filtser2015discrete,fan2017computing}. 
Some researchers have argued that this measure is similar to the \Fr distance under translation in certain aspects, in particular since the \Fr gap between a curve $\pi$ and a translation $\pi + \tau$ of the same curve is 0~\cite{fan2017computing}. Moreover, the \Fr gap can be computed significantly faster, with the currently fastest known algorithm running in time $\Oh(n^3)$ \cite{filtser2015discrete}. In some sense, our conditional lower bound in Theorem~\ref{thm:mainlower} explains why replacing the \Fr distance under translation by such a surrogate measure is necessary to obtain more efficient algorithms. Additionally, the discrete Fréchet distance \emph{with shortcuts} was also recently considered in a translation-invariant setting \cite{Filtser18}.

\paragraph{On related reachability data structures.}
In \cite{AltG95}, a reachability data structure on the free-space diagram is given to compute the Fréchet distance between closed curves and also to compute the best matching to any subcurve under the Fréchet distance. In \cite{maheshwari_improved_2014}, a data structure called the free-space map is presented, which improves the reachability data structure of \cite{AltG95} and as a consequence shaves off logarithmic factors in the running time for computing the above mentioned applications as well as improving the running time for other applications.

\subsection{Organization}
We start off with introducing basic definitions, notational conventions, and algorithmic tools in \secref{preliminaries}. Afterwards, in \secref{algo}, we give an overview of our algorithmic result and we reduce the problem to designing a certain data structure for Offline Dynamic Grid Reachability. This data structure, our main technical contribution, is developed in \secref{reach-ds}. Finally, we prove our conditional lower bound of $n^{4-o(1)}$ in \secref{lowerbound}.

 \newcommand{\OR}{{\mathcal{OR}}}
\newcommand{\univ}{{\mathcal{U}}}

\newcommand{\concat}{\circ}

\newcommand{\dF}{\delta_F}

\section{Preliminaries} \label{sec:preliminaries}

We let $[n]$ denote the set $\{1,\dots, n\}$. Furthermore, for convenience, we use as convention that $\min \emptyset = \infty$ and $\max \emptyset = - \infty$.

\subsection{Curves, Traversals, Fréchet distances, and more}

A polygonal curve $\pi$ of length $n$ over $\mathbb{R}^d$ is a sequence of points $\pi_1, \dots, \pi_n \in \mathbb{R}^d$. Throughout the paper, we only consider polygonal curves in the Euclidean plane, i.e., $d=2$. Given any translation vector $\tau \in \mathbb{R}^2$, we denote by $\pi + \tau$ the polygonal curve $\pi' = (\pi'_1, \dots, \pi'_n)$ given by $\pi'_i = \pi_i + \tau$. 

We call any pair $(i,j)\in [n]\times[n]$ a \emph{position}. A \emph{traversal} $T$ is a sequence $t_1, \dots, t_\ell$ of positions, where $t_k = (i,j)$ implies that $t_{k+1}$ is either $(i+1, j)$ (that is, we advance one step in $\pi$ while staying in $\sigma_j$), $(i, j+1)$ (we advance in $\sigma$ while staying in $\pi_i$), or $(i+1,j+1)$ (we advance in both curves simultaneously). We call $T= (t_1, \dots, t_\ell)$ a traversal of $\pi, \sigma$, if $t_1 = (1,1)$ and $t_\ell = (n,n)$. 

We now define two types of concatenations: a concatenation of curves and a concatenation of traversals. Let $\pi = (\pi_1, \dots, \pi_n), \sigma = (\sigma_1, \dots, \sigma_n)$ be polygonal curves of lengths $n$. We define the concatenation of $\pi$ and $\sigma$ as $\pi \concat \sigma \coloneqq (\pi_1, \dots, \pi_n, \sigma_1, \dots, \sigma_n)$. The resulting curve has length $2n$.
We now define the concatenation of traversals. Given two traversals $T = (t_1, \dots, t_\ell)$ and $T' = (t'_1, \dots, t'_{\ell'})$ with $t_\ell = t'_1$, we define the concatenated traversal as $T \concat T' \coloneqq (t_1, \dots, t_\ell = t'_1, t'_2, \dots, t'_{\ell'})$. Note that we obtain a traversals from $t_1$ to $t'_{\ell'}$.

The \emph{discrete Fréchet distance} is formally defined as
\[ \dF(\pi, \sigma) \coloneqq \min_{T = ((i_1, j_1), \dots, (i_\ell, j_\ell))} \max_{1\le k \le \ell} \|\pi_{i_k} - \sigma_{j_k} \|, \]
where $T$ ranges over all traversals of $\pi, \sigma$ and $\|\cdot \|$ denotes the Euclidean distance in $\mathbb{R}^2$.

We obtain a well-known equivalent definition as follows: Fix some distance $\delta \ge 0$. We call a position $(i,j)$ \emph{free} if $\| \pi_i - \sigma_j\| \le \delta$. We say that a traversal $T= (t_1, \dots, t_\ell)$ of $\pi, \sigma$ is a \emph{valid traversal} for $\delta$ if $t_1, \dots, t_\ell$ are all free positions. The discrete Fréchet distance of $\pi, \sigma$ is then the smallest $\delta$ such that there is a valid traversal of $\pi,\sigma$ for $\delta$.

Analogously, consider the $n \times n$ matrix $M$ with $M_{i,j} = 1$ if $(i,j)$ is free, and $M_{i,j} = 0$ otherwise. We call any traversal $T = (t_1, \dots, t_\ell)$ a \emph{monotone path from $t_1$ to $t_\ell$}. If all positions $(i,j)$ visited by $T$ satisfy $M_{i,j} = 1$, we call $T$ a \emph{monotone 1-path} from $t_1$ to $t_\ell$ in $M$. As yet another formulation, consider the $n \times n$ grid graph $G_M$ where vertex $(i,j)$ has directed edges to all of $(i,j+1), (i+1,j)$, and $(i+1,j+1)$ (in case they exist). Deactivate (i.e., remove) all non-free vertices $(i,j)$ from $G_M$. Then a monotone 1-path in $M$ corresponds to a (directed) path in $G_M$. 
Hence, $\dF(\pi, \sigma) \le \delta$ is equivalent to the existence of a valid traversal of $\pi, \sigma$ for $\delta$, which in turn is equivalent to the existence of a monotone 1-path from $(1,1)$ to $(n,n)$ in the matrix~$M$, and to vertex $(n,n)$ being reachable from $(1,1)$ in $G_M$.

Finally, we define the discrete Fréchet distance under translation as $\min_{\tau\in \mathbb{R}^2} \dF(\pi, \sigma + \tau)$, i.e., the smallest discrete Fréchet distance of $\pi$ to any translation of $\sigma$.

\subsection{Hardness Assumptions}

The \emph{Strong Exponential Time Hypothesis (SETH)} was introduced by Impagliazzo and Paturi~\cite{ImpagliazzoP01} and essentially postulates that there is no exponential-time improvement over exhaustive search for the Satisfiability Problem.
\begin{hypo}[Strong Exponential Time Hypothesis (SETH)]
For any $\epsilon > 0$ there exists $k \ge 3$ such that $k$-SAT has no $\Oh((2-\epsilon)^n)$-time algorithm.
\end{hypo}

In fact, our reductions even hold under a weaker assumption, specifically, the \emph{$k$-OV Hypothesis}.\footnote{In fact, we only need the corresponding hypothesis for $4$-OV.}  Recall the $k$-OV problem: Given sets $V_1,\ldots,V_k$ of $N$ vectors in $\{0,1\}^D$, the task is to determine whether there are $v_1 \in V_1,\ldots, v_k \in V_k$ such that for all $j \in [D]$ there exists an $i \in [k]$ with $v_i[j] = 0$.  
\begin{hypo}[$k$-OV Hypothesis]
For any $k \ge 2$ and $\epsilon > 0$, there is no $\Oh(N^{k-\epsilon} \mathrm{poly}(D))$-time algorithm for $k$-OV.
\end{hypo}
 The well-known split-and-list technique due to Williams~\cite{Wil05} shows that SETH  implies the $k$-OV Hypothesis. Thus, any conditional lower bound that holds under the $k$-OV hypothesis also holds under SETH.

\subsection{Orthogonal range data structures}
\label{sec:orthrange}

We will use a tool from geometric data structures, namely \emph{(dynamic) orthogonal range data structures}. Let $S$ be a set of key-value pairs $s = (k_s,v_s) \in \mathbb{Z}^d \times \mathbb{Z}$; in our applications, we will have $d=2$ or $d=3$. An orthogonal range data structure on $S$ enables us to query the maximal value of any pair in $S$ whose key lies in a given orthogonal range.
Formally, we say \emph{$\OR$ stores $v_s$ under the key $k_s$ for $s \in S$ for minimization queries}, if $\OR$ supports, for any $\ell_1, u_1, \ell_2, u_2, \dots, \ell_d, u_d \in \mathbb{Z} \cup \{-\infty,\infty\}$, queries of the form 
\[ \OR.\min([\ell_1, u_1]\times \cdots \times [\ell_d, u_d]): \text{ return } \min \{ v_s \mid s = (k_s,v_s) \in S, k_s \in [\ell_1, u_1]\times \cdots \times [\ell_d, u_d] \}. \]
We will also consider analogous maximization queries.

Classic results~\cite{GabowBT84, Chazelle88} show that for any set $S$ of size $n$ and $d=2$, we can construct such a data structure $\OR$ in time and space $\Oh(n\log n)$, supporting minimization (or maximization) queries in time $\Oh(\log n)$.

In Section~\ref{sec:singlestepreach}, we will also use an orthogonal range searching data structure that allows (1) to \emph{report} all values of pairs in $S$ whose keys lie in a given orthogonal range, and (2) to \emph{remove} a key-value pair from $S$. Formally, we say that \emph{$\OR$ stores $v_s$ under the key $k_s$ for $s \in S$ for decremental range reporting queries}, if $\OR$ supports, for any $\ell_1, u_1, \ell_2, u_2, \dots, \ell_d, u_d \in \mathbb{Z} \cup \{-\infty,\infty\}$, queries of the form 
\[ \OR.\mathrm{report}([\ell_1, u_1]\times \cdots \times [\ell_d, u_d]): \text{ return } \{ v_s \mid s = (k_s,v_s) \in S, k_s \in [\ell_1, u_1]\times \cdots \times [\ell_d, u_d] \}, \]
as well as deletions from the set $S$.

Mortensen~\cite{Mortensen06} and Chan and Tsakalidis~\cite{ChanT17} showed how to construct such a data structure $\OR$ for any set $S$ of size $n$ in time and space $\Oh(n \log^{d-1} n)$, deletion time $\Oh( \log^{d-1} n)$ and query time $\Oh(\log^{d-1} n + k)$, where $k$ denotes the output size of the query. (These works obtain even stronger results, however, we use simplified bounds for ease of presentation.)

\newcommand{\trav}[2]{{#1 \rightsquigarrow #2}}
\newcommand{\travseq}[3]{{#1 \rightsquigarrow #2 \rightsquigarrow \dots \rightsquigarrow #3}}
\newcommand{\travsub}[3]{{#1 \rightsquigarrow_{#3} #2}}

\newcommand{\mypath}{P}

\newcommand{\sA}{\mathrm{A}}
\newcommand{\sZ}{\mathrm{Z}}
\newcommand{\intvl}{{\mathcal{I}}}

\newcommand{\rev}{\mathrm{rev}}
\newcommand{\ellrev}{\ell^\rev}
\newcommand{\sArev}{\sA^\rev}
\newcommand{\sZrev}{\sZ^\rev}
\newcommand{\Lrev}{L^\rev}
\newcommand{\intvlrev}{\intvl^\rev}
\newcommand{\reachrev}{\reach^\rev}

\newcommand{\lt}{\mathbf{l}}
\newcommand{\rt}{\mathbf{r}}
\newcommand{\tp}{\mathrm{top}}

\newcommand{\Pcom}{P_\mathrm{com}}
\newcommand{\face}{{\ensuremath {\mathcal{F}}}\xspace}
\newcommand{\reachr}{{\mathcal{R}}}

\newcommand{\reach}{\mathrm{Reach}}
\newcommand{\singleStepReach}{\mathrm{SingleStepReach}}

\newcommand{\R}{\mathbb{R}}

\newcommand{\Bin}{B^-}
\newcommand{\Bout}{B^+}
\newcommand{\Bmid}{B^\mathrm{mid}}
\newcommand{\Bmidfree}{B^\mathrm{mid}_\mathrm{free}}
\newcommand{\Jmid}{{J^\mathrm{mid}}}
\newcommand{\sBbound}{{|\partial B|}}

\newcommand{\bell}{\bar{\ell}}

\newcommand{\qcenter}{q_\mathrm{c}}

\newcommand{\diagram}{{\mathcal{D}}}

\newcommand{\update}[1]{[\![#1]\!]}

\newcommand{\Parameter}{\State \textbf{parameter:} }

\newcommand{\term}{{\mathcal{T}}}

\newcommand{\blocks}{{\mathcal{B}}}

\newcommand{\ORB}{\OR_{B}}

\newcommand{\dec}{\mathrm{dec}}

\newcommand{\arr}{{\mathcal{A}}}

\newcommand{\disk}{D}

\newcommand{\optdelta}{{\delta^*}}

\newcommand{\crossingop}[3]{C^{#3}_{#1}(#2)}
\newcommand{\crossingA}[2]{\crossingop{#1}{#2}{1}}
\newcommand{\crossingB}[2]{\crossingop{#1}{#2}{2}}
\newcommand{\crossing}[2]{\crossingop{#1}{#2}{\cdot}}

\newcommand{\crossings}[1]{{\mathcal{C}}_{#1}}

\newcommand{\partJ}{{\mathcal{J}}}

\section{Algorithm: Reduction to Grid Reachability} \label{sec:algo}

In this section, we prove our algorithmic result by showing how a certain grid reachability data structure (that we give in \secref{reach-ds}) yields an $\tOh(n^{4+2/3})$-time algorithm for computing the discrete Fréchet distance under translation. 

We start with a formal overview of the algorithm.
First, we reduce the decision problem (i.e., is the discrete Fréchet  distance under translation of $\pi, \sigma$ at most $\delta$?) to the problem of determining reachability in a dynamic grid graph, as shown by Ben Avraham et al.~\cite{avraham2015faster}. However, noting that all updates and queries are known in advance, we observe that the following \emph{offline} version suffices.

\begin{problem}[Offline Dynamic Grid Reachability]
Let $M$ be an $n \times n$ matrix over $\{0,1\}$. We call $u = (p,b)$ with $p \in [n]\times [n]$ and $b\in \{0,1\}$ an \emph{update} and define $M\update{u}$ as the matrix obtained by setting the bit at position $p$ to $b$, i.e., 
    \[M\update{(p,b)}_{i,j}  = \begin{cases} b & \text{if } p = (i,j), \\ M_{i,j} & \text{otherwise. } \end{cases} \]    
For any sequence of updates $u_1, \dots, u_k \in ([n] \times [n]) \times \{0,1\}$ with $k \ge 2$, we define $M\update{u_1, \dots, u_k} \coloneqq (M\update{u_1})\update{u_2, \dots, u_k}$.

The \emph{Offline Dynamic Grid Reachability} problem asks to determine, given $M$ and any sequence of updates $u_1, \dots, u_U \in ([n]\times [n]) \times \{0,1\}$, whether there is a monotone 1-path from $(1,1)$ to $(n,n)$ in $M\update{u_1, \dots, u_k}$ for any $1 \le k \le U$. 
\end{problem}

We show the following reduction in \secref{redDynGridReach}.

\begin{lem}\label{lem:redDynGridReach}
Assume there is an algorithm solving Offline Dynamic Grid Reachability in time $T(n,U)$. Then there is an algorithm that, given $\delta > 0$ and polygonal curves $\pi,\sigma$ of length $n$ over $\mathbb{R}^2$, determines whether $\dF(\pi,\sigma+\tau) \le \delta$ for some $\tau \in \mathbb{R}^2$ in time $\Oh(T(n, n^{4}))$.
\end{lem}

Our speed-up is achieved by solving Offline Dynamic Grid Reachability in time $T(n,U) = \tOh(n^2 + U n^{2/3})$ (Ben Avraham et al.~\cite{avraham2015faster} achieved $T(n,U) = \Oh(n^2 + Un)$). To this end, we devise a \emph{grid reachability data structure}, which is our central technical contribution.

\begin{restatable}[Grid reachability data structure]{lem}{dslemma}\label{lem:ds}
Given an $n \times n$ matrix $M$ over $\{0,1\}$ and a set of \emph{terminals} $\term \subseteq [n] \times [n]$ of size~$k>0$, there is a data structure $\diagram_{M, \term}$ with the following properties. 
\begin{enumerate}[label=\roman*)]
\item\label{enum:ds-init} (Construction:) We can construct $\diagram_{M, \term}$ in time $\Oh(n^2 + k\log^2 n)$.
\item\label{enum:ds-reachqueries} (Reachability Query:) Given $F\subseteq \term$, we can determine in time $\Oh(k \log^3 n)$ whether there is a monotone path from $(1,1)$ to $(n,n)$ using only positions $(i,j)$ with $M_{i,j} = 1$ or $(i,j) \in F$. 
\item\label{enum:ds-updates} (Update:) Given $\term' \subseteq [n]\times [n]$ of size $k$ and an $n\times n$ matrix $M'$ over $\{0,1\}$ differing from $M$ in at most $k$ positions, we can update $\diagram_{M,\term}$ to $\diagram_{M',\term'}$ in time $\Oh(n\sqrt{k} \log n + k\log^2 n)$. Here, we assume $M'$ to be represented by the set $\Delta$ of positions in which $M$ and $M'$ differ.
\end{enumerate}
\end{restatable}

\secref{reach-ds} is dedicated to devising this data structure, i.e., proving \lemref{ds}.
Equipped with this data structure, we can efficiently \emph{batch} updates and queries to the data structure. Specifically, we obtain the following theorem.

\begin{thm}\label{thm:algo-batchfrechet}
Offline Dynamic Grid Reachability can be solved in time $\Oh(n^2 + Un^{2/3}\log^2 n)$.
\end{thm}

We prove this theorem in \secref{solveDynGridReach}. Finally, it remains to use standard techniques of parametric search to transform the decision algorithm to an algorithm computing the discrete Fréchet distance under translation. This has already been shown by Ben Avraham et al.~\cite{avraham2015faster}; we sketch the details in \secref{paramsearch}.

\begin{lem}\label{lem:paramsearch}
Let $T_\dec(n)$ be the running time to decide, given $\delta > 0$ and polygonal curves $\pi,\sigma$ of length $n$ over $\mathbb{R}^2$, whether $\dF(\pi,\sigma+\tau) \le \delta$ for some $\tau \in \mathbb{R}^2$. Then there is an algorithm computing the discrete Fréchet distance under translation for any curves $\pi,\sigma$ of length $n$ over $\mathbb{R}^2$ in time $\Oh( (n^{4} + T_\dec(n))\log n )$.
\end{lem}

Combining \lemref{paramsearch}, \lemref{redDynGridReach} and \thmref{algo-batchfrechet}, we obtain an algorithm computing the discrete Fréchet distance under translation in time
\[ \Oh( (n^{4} + T(n,n^{4}))\log n) =  \Oh( n^{4+2/3} \log^3 n),  \]
as desired.
In the remainder of this section, we provide the details of all steps mentioned above, except for \lemref{ds} (which we prove in \secref{reach-ds}).

\subsection{Reduction to Offline Dynamic Grid Reachability}
\label{sec:redDynGridReach}

In this section we prove \lemref{redDynGridReach}.
Given polygonal curves $\pi, \sigma$ of length $n$ over $\mathbb{R}^2$ and $\delta > 0$, we determine whether $\dF(\pi, \sigma + \tau) \le \delta$ for some $\tau \in \mathbb{R}^2$ as follows. 

For any radius $r$ and point $p \in \mathbb{R}^2$, we let $\disk_r(p)$ denote the disk of radius $r$ with center $p$.
\begin{obs}\label{obs:vertextest}
Let $\tau \in \mathbb{R}^2$ and define the $n\times n$ matrix $M^{\tau}$ over $\{0,1\}$ by 
\[ M_{i,j}^\tau = 1 \qquad \iff \qquad \tau \in \disk_\delta(\pi_i - \sigma_j).\]
We have $\dF(\pi, \sigma + \tau) \le \delta$ if and only if there is a monotone 1-path from $(1,1)$ to $(n,n)$ in $M^\tau$.
\end{obs}
By the above observation, it suffices to check for the existence of monotone 1-paths from $(1,1)$ to $(n,n)$ in a bounded number of matrices. To this end, let $Q \coloneqq \{ \pi_i - \sigma_j \mid i,j \in [n]\}$. We construct the arrangement $\arr_\delta$ of the disks $\disk_\delta(q)$ for $q\in Q$, in the sense that we construct the following plane graph $G_\delta$ (cf.\ Figure~\ref{fig:arrangement}). First, we include the vertices of $\arr_\delta$ in its node set (i.e., intersections of disks $\disk_\delta(q), \disk_\delta(q')$ with $q,q'\in Q$). Second, for each $q\in Q$ for which $\disk_\delta(q)$ intersects no $\disk_\delta(q')$ for $q'\in Q\setminus\{q\}$, we include an arbitrary $\tau_q$ on the boundary of $\disk_\delta(q)$. Finally, we add an arbitrary vertex $\tau_0 \in \mathbb{R}^2$ lying in the outer face of $\arr_\delta$ to the node set. Any nodes $\tau, \tau'$ of $G_\delta$ are connected by an edge if they are neighboring vertices on the boundary of some face of $\arr_\delta$; additionally, we connect $\tau_0$ to all nodes which lie on the boundary separating the outer face from some other face. Observe that $G_\delta$ is a connected plane graph, has $\Oh( |Q|^2 ) = \Oh(n^4)$ nodes and edges, and can be constructed in time $\Oh(n^4)$. 

\begin{figure}
\centering
\includegraphics[width=0.8\textwidth]{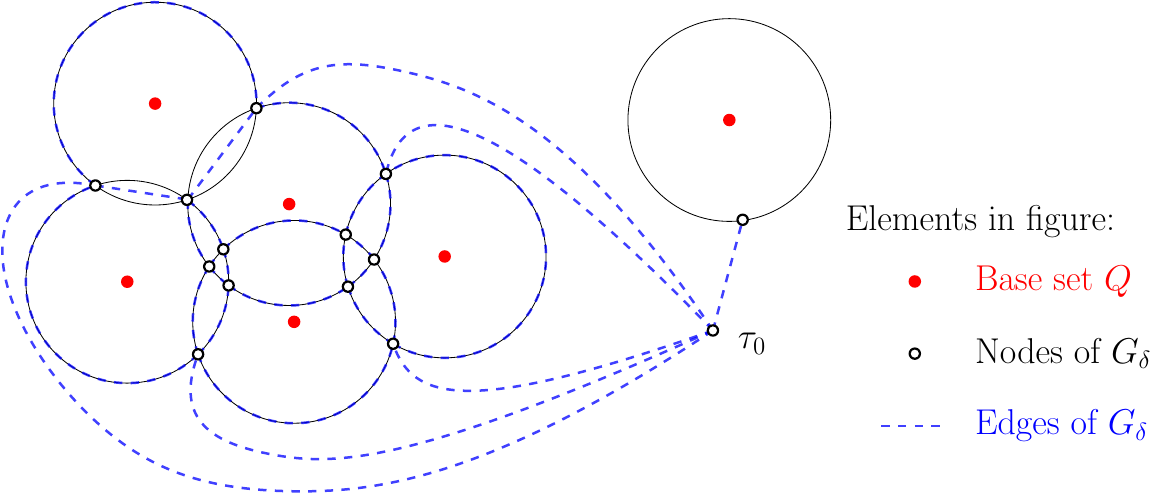}
\caption{Arrangement $\arr_\delta$ and construction of $G_\delta$.}
\label{fig:arrangement}
\end{figure}

Note that by \obsref{vertextest}, it suffices to check whether $\dF(\pi, \sigma + \tau_v) \le \delta$ for any node $\tau_v$ in\footnote{Note that our approach here deviates somewhat from the description in the introduction. This is due to the fact that for adversarial $\delta$, we might need to consider degenerate faces consisting of a single point only; due to the parametric search that we describe in \secref{paramsearch}, we may not assume $\delta$ to avoid such degenerate cases. Traversing vertices of the arrangement instead of the faces takes care of such border cases in a natural manner.} $G_\delta$: for any (bounded) face $f$ of $\arr_\delta$, there is at least one point $\tau_v$ in $G_\delta$ that lies on the boundary of $f$. The corresponding matrix $M^{\tau_v}$ has at least the same 1-positions as the matrix $M^\tau$ for any $\tau \in f$ (and might have more).

To obtain a walk visiting all nodes in $G_\delta$, we simply compute a spanning tree $T$ of $G_\delta$, double all edges of $T$, and find an Eulerian cycle starting and ending in $\tau_0$. Denote this cycle by $\tau_0, \dots, \tau_L$ and observe that $L = \Oh(n^4)$. Let $M_0 = M^{\tau_0}$ be the $n\times n$ all-zeroes matrix. For any $0 \le i < L$, we construct an update sequence $\bar{u}_i$ that first sets all positions $(i,j)$ with $M^{\tau_i} = 1$ and $M^{\tau_{i+1}} = 0$ to zero, and then sets all positions $(i,j)$ with $M^{\tau_i} = 0$ and $M^{\tau_{i+1}} = 1$ to 1. Thus, if we start with $M^{\tau_{i}}$ and perform the updates in $\bar{u}_i$, then at any point in time, the current matrix is dominated by either $M^{\tau_i}$ or $M^{\tau_{i+1}}$ (that is, the free positions of the current matrix are always a subset of $M^{\tau_i}$'s free positions or a subset of $M^{\tau_{i+1}}$'s free positions), and at the end we obtain $M^{\tau_{i+1}}$. Thus, by concatenating all updates to $\bar{u}_0, \dots, \bar{u}_{L-1}$, we obtain an instance of the Offline Dynamic Grid Reachability problem with initial matrix $M_0$ and update sequence $u_1, \dots, u_{L'}$ with the following property: There is some $i \in \{0, \dots, L\}$ with $\dF(\pi, \sigma+ \tau_i) \le \delta$ if and only if there is some $i' \in [L']$ such that $(n,n)$ is reachable from $(1,1)$ via a monotone 1-path in $M_0\update{u_1, \dots, u_{i'}}$. Since $\tau_0, \dots, \tau_L$ visits all nodes in $G_\delta$, this is equivalent to testing whether $\dF(\pi, \sigma + \tau) \le \delta$ for any $\tau \in \mathbb{R}^2$.

It remains to bound $L'$. We assume general position of the input points $P\cup S$ where  $P = \{\pi_i \mid i \in [n]\}$ and $S=\{\sigma_j \mid j \in [n]\}$. Observe that there is some universal constant~$C$ such that no $C$~points in $Q$ lie on a common circle.\footnote{To be more precise, we sketch how to argue that the general position assumption for $P \cup S$ ``transfers'' to~$Q$. Assume that there exist points $q_1, \dots, q_\ell \in Q$ lying on a common circle. For all $i$, we must have $q_i = p_i - s_i$ for some $p_i \in P, s_i \in S$. First assume that $\ell = 4$ and that there is some $s$ such that $s_i = s$ for all $i \in \{1,2,3,4\}$. Then already $p_1, \dots, p_4$ lie on a common circle (it has the same radius as the original circle, and its center is translated by $s$), which violates the general position assumption of points in $P$. Otherwise, let the points $q_1, \dots, q_\ell$ be arbitrary with $\ell \ge 36$. By the first case, any $s_i$ appears at most 3 times among $s_1, \dots, s_\ell$. After removing copies, we may assume without loss of generality that $q_1, \dots, q_{\ell'}$ with $\ell' \ge \ell/3$ have distinct $s_i$'s. Similarly, we may also assume that $q_1, \dots, q_{\ell''}$ with $\ell'' \ge \ell/9 \ge 4$ have distinct $p_i$'s as well. The fact that $q_4$ lies on the circle defined by $q_1, q_2, q_3$ can be expressed by a nonzero degree-2 polynomial $P_{q_1,q_2,q_3}(x, y)$ vanishing on $q_4$. Since $q_4 = p_4 - s_4$, we obtain a nonzero degree-2 polynomial $P_{q_1, q_2, q_3}'(p,s)$ vanishing on $(p_4, s_4)$. This contradicts general position of $P \cup S$.}
Thus, if we move from vertex $\tau_i$ to $\tau_{i+1}$ along an edge in $G_\delta$, e.g., from one vertex  of the boundary of some face to a neighboring vertex on that boundary, there are at most $2C$ entries that change from $M^{\tau_i}$ to $M^{\tau_{i+1}}$, since for both $\tau_i$ and $\tau_{i+1}$, there are at most $C$ disks intersecting this vertex and no other entries change when moving along this edge (by construction of $G_\delta$). Thus, $L'  \le 2CL = \Oh(n^4)$. Consequently, given an algorithm solving Offline Dynamic Grid Reachability in time $T(n,U)$, we can determine whether $\dF(\pi, \sigma + \tau) \le \delta$ for some $\tau \in \mathbb{R}^2$ in time $\Oh(T(n,L')) = \Oh(T(n, n^4))$.

\subsection{Solving Offline Dynamic Grid Reachability}
\label{sec:solveDynGridReach}

We prove \thmref{algo-batchfrechet} using the grid reachability data structure given in \lemref{ds}.  
Specifically, we claim that the following algorithm (formalized as Algorithm~\ref{alg:batchfrechet}) solves Offline Dynamic Grid Reachability in time $\Oh(n^2 + Un^{2/3}\log^2 n)$. We partition our updates $u_1,\dots, u_U$ into groups $\bar{u}_1, \dots, \bar{u}_{\Oh(U/k)}$ containing $k$ updates each. For any group $\bar{u}_i$, let $M_i$ be obtained from $M$ by performing all updates \emph{prior} to $\bar{u_i}$. Note that $\bar{u}_i$ will update a set of at most $k$ positions; denote this set by $\term_i$. We build the grid reachability data structure $\diagram_i=\diagram_{M^\zero_i, \term_i}$ with terminal set $\term_i$ and matrix $M^\zero_i$ obtained from $M_i$ by setting the positions of all terminals $\term_i$ to 0. Observe that the state after any update within $\bar{u}_i$ corresponds to $M_i^\zero$ with some additional positions in $\term_i$ set to $1$ (the \emph{free terminals}). Thus, for each update within $\bar{u}_i$, we can determine whether it creates a monotone 1-path from $(1,1)$ to $(n,n)$ by simply determining the set $F\subseteq \term_i$ of free terminals at the point of this update and performing the corresponding reachability query in~$\diagram_i$. 
It is straightforward to argue that the resulting algorithm correctly solves Offline Dynamic Grid Reachability.

\begin{algorithm}
\begin{algorithmic}[1]
\Function{OfflineDynamicGridReachability}{$M$, $u_1, \dots, u_U$}
\Parameter $k$
\State Divide $u_1, \dots, u_U$ into $s =  \left\lceil \frac{U}{k} \right\rceil$ subsequences $\bar{u}_1, \dots, \bar{u}_{s}$ of length $k$.\footnote{If necessary, repeat the last element of the last group to make all groups consist of exactly $k$ updates.}
\State Initialize $M_1 \gets M$ 
\State Set $\term_1$ to the set of positions updated in $\bar{u}_1$.
\State Let $M^\zero_1$ be obtained from $M_1$ by updating all positions in $\term_1$ to $0$
\State Build $\diagram_{M^\zero_1,\term_1}$ \label{line:build}
\For{$i\gets 1$ \textbf{to} $s$} 
	\For{$j\gets 1$ \textbf{to} $k$}
		\State Let $F\subseteq \term_i$ be the free terminals in $M_i\update{\bar{u}_i[1], \dots, \bar{u}_i[j]}$
		\If{reachability query in $\diagram_{M^\zero_i, \term_i}$ with free terminals $F$ is successful} \label{line:query}
			\State \Return \textbf{true}
	        \EndIf	
	\EndFor
\State Set $M_{i+1} \gets M_i\update{\bar{u_i}}$
	\State Set $\term_{i+1}$ to the set of positions updated in $\bar{u}_{i+1}$.\footnote{We let $\bar{u}_{s+1}$ consist of $k$ arbitrary updates, as we will never make use of these values.}
	\State Let $M^\zero_{i+1}$ be obtained from $M_{i+1}$ by updating all positions in $\term_{i+1}$ to $0$
	\State update $\diagram_{M^\zero_i, \term_i}$ to $\diagram_{M^\zero_{i+1}, \term_{i+1}}$ \label{line:update}
\EndFor
\State \Return \textbf{false}
\EndFunction
\end{algorithmic}
\caption{Solving Offline Dynamic Grid Reachability on matrix $M$ and update sequence $u_1, \dots, u_U$.} 
\label{alg:batchfrechet}
\end{algorithm}

To analyze the running time of Algorithm~\ref{alg:batchfrechet}, note that each data structure $\diagram_{M^\zero_i, \term_i}$ has a terminal set of size at most $k$ and each $M^\zero_{i+1}$ differs from $M^\zero_i$ in at most $2k$ entries. Thus by \lemref{ds}, we need time $\Oh(n^2 + k \log^2 n)$ to build $\diagram_1 = \diagram_{M^\zero_1, \term_1}$ in Line~\ref{line:build}. The time spent for handling a single group $\bar{u}_i$ is bounded by the time to perform $k$ queries in $\diagram_i = \diagram_{M^\zero_i, \term_i}$ plus the time to update $\diagram_i = \diagram_{M^\zero_i, \term_i}$ to $\diagram_{i+1}= \diagram_{M^\zero_{i+1}, \term_{i+1}}$, which amounts to $\Oh( k^2 \log^3 n + n\sqrt{k} \log n + k \log^2 n) = \Oh(k^2 \log^3 n + n\sqrt{k}\log n)$ by \lemref{ds}. Thus, in total, we obtain a running time of \[\Oh\left( n^2 + k\log^2 n + \frac{U}{k} \left( k^2 \log^3 n + n\sqrt{k} \log n\right)\right) = \Oh\left(n^2 + U \left(k\log^3 n + \frac{n}{\sqrt{k}}\log n\right) \right).\]
This expression is minimized by setting $k\coloneqq n^{2/3}/\log^{4/3} n$, resulting in a total running time of $\Oh(n^2 + Un^{2/3} \log^{1+2/3} n) = \Oh(n^2 + Un^{2/3}\log^2 n)$, as desired.

\subsection{Parametric Search}
\label{sec:paramsearch}

In this section, we sketch how to use parametric search techniques (due to Megiddo~\cite{Megiddo83} and Cole~\cite{Cole87}) to reduce the optimization problem to the decision problem with small overhead, i.e., we prove \lemref{paramsearch}. Specifically, for the readers' convenience, we describe the arguments made by Ben Avraham et al.~\cite{avraham2015faster} in slightly more detail.   

Our aim in this section is to compute the discrete Fréchet distance under translation of polygonal curves $\pi,\sigma$ of length $n$ over $\mathbb{R}^2$, i.e., to determine
\[ \optdelta \coloneqq \min_{\tau \in \mathbb{R}^2} \dF(\pi, \sigma+\tau). \]

Using the decision algorithm, we can determine, for any $\delta > 0$, whether $\optdelta \le \delta$ in time $T_\dec(n)$. As we shall see below, there is a range of $\Oh(n^6)$ possible values  (defined by the point set of $\pi, \sigma$) that $\optdelta$ might attain (called \emph{critical values}). Naively computing all critical values and performing a binary search would result in an $\Oh((n^6 + T_\dec(n)) \log n)$-time algorithm, which is too slow for our purposes. Instead, we use the parametric search technique to perform an implicit search over these critical values. 

Conceptually, we aim to determine the combinatorial structure of the arrangement $\arr_\optdelta$ defined in \secref{redDynGridReach} (captured by the graph $G_\optdelta$) without knowing $\optdelta$ in advance. To specify this combinatorial structure, define for every $q \in Q$ the set
\[
I_\delta(q) \coloneqq \{ q' \in Q\setminus\{q\} \mid \disk_\delta(q), \disk_\delta(q') \text{ intersect} \}.
\]
Note that for every $q' \in I_\delta(q)$, there are one or two intersection points of $\disk_{\delta}(q)$, $\disk_{\delta}(q')$, which we denote by $\crossingA{\delta}{q,q'}$ and $\crossingB{\delta}{q,q'}$ (note that we allow these points to coincide if $\disk_\delta(q), \disk_\delta(q')$ intersect in a single point only) -- we assume this notation to be chosen consistently in the sense that $\crossingA{\delta}{q,q'}$ and $\crossingA{\delta}{q',q}$ refer to the same point (likewise for $\crossingB{\delta}{q,q'}$ and $\crossingB{\delta}{q',q}$). We denote by $\crossings{\delta}(q)$ the set of all intersection points on the boundary of $\disk_\delta(q)$, i.e., $\crossingA{\delta}{q,q'}$, $\crossingB{\delta}{q,q'}$ for all $q' \in I_\delta(q)$. We obtain a list $L_\delta(q)$ by starting with the rightmost point on the boundary of $\disk_\delta(q)$, say $r_q$, and listing all intersection points $C \in \crossings{\delta}(q)$ in counter-clockwise order. Observe that the combinatorial structure of $\arr_\delta$ is completely specified by the lists $L_\delta(q)$ for $q\in Q$. 

We wish to construct $L_\optdelta(q)$ for all $q\in Q$ using calls to our decision algorithm, i.e., queries of the form ``Is $\optdelta \le \delta$?''. Along the way, we maintain a shrinking interval $(\alpha, \beta]$ such that $\optdelta \in (\alpha, \beta]$ -- our aim is that in the end $(\alpha, \beta]$ no longer contains critical values except for $\beta$, and thus $\optdelta = \beta$ can be derived. We proceed in two steps.

\paragraph{Step 1: Determining $I_\optdelta(q)$.}
The critical values for this step are the half-distances of all pairs $q,q'\in Q$ (cf. Figure~\ref{fig:criticalvalues1}). We list all these values and perform a binary search over them, using our decision algorithm. Since there are at most $\Oh(|Q|^2) = \Oh(n^4)$ such values, we obtain an algorithm running in time $\Oh((n^4 + T_\dec(n)) \log n)$ returning an interval $(\alpha_1, \beta_1]$ such that $\optdelta \in (\alpha_1, \beta_1]$ and no half-distance of a pair $q,q'\in Q$ is contained in $(\alpha_1, \beta_1)$. Thus, from this point on we know $I_\optdelta(q)$ for all $q\in Q$ (without knowing the exact value of $\optdelta$ yet).

\begin{figure}
	\begin{subfigure}[t]{0.32\textwidth}
		\includegraphics[width=\textwidth]{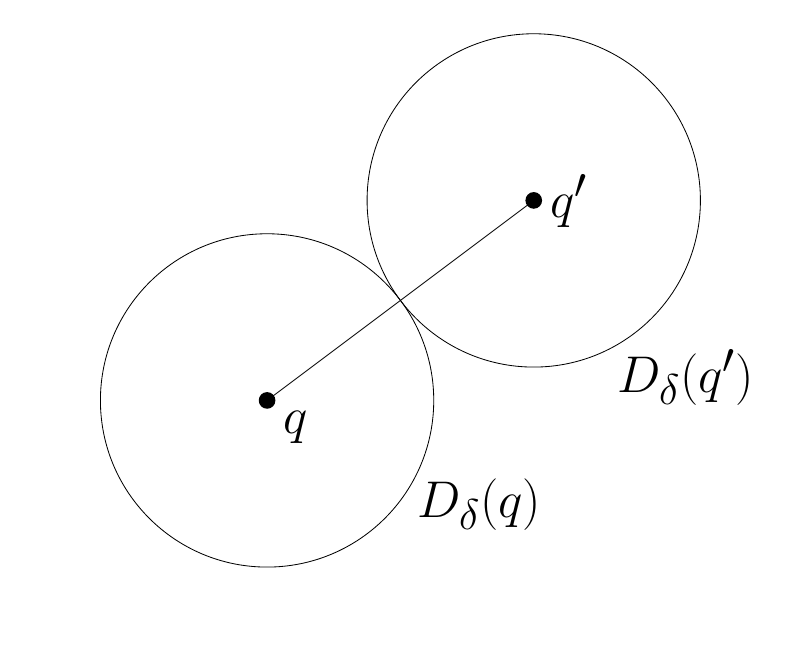}
		\caption{The elements in $L_\delta(q)$ change at radius $\delta$ that is a half-distance to  some other $q'\in Q$.}
		\label{fig:criticalvalues1}
	\end{subfigure}
	\hfill
	\begin{subfigure}[t]{0.32\textwidth}
		\includegraphics[width=\textwidth]{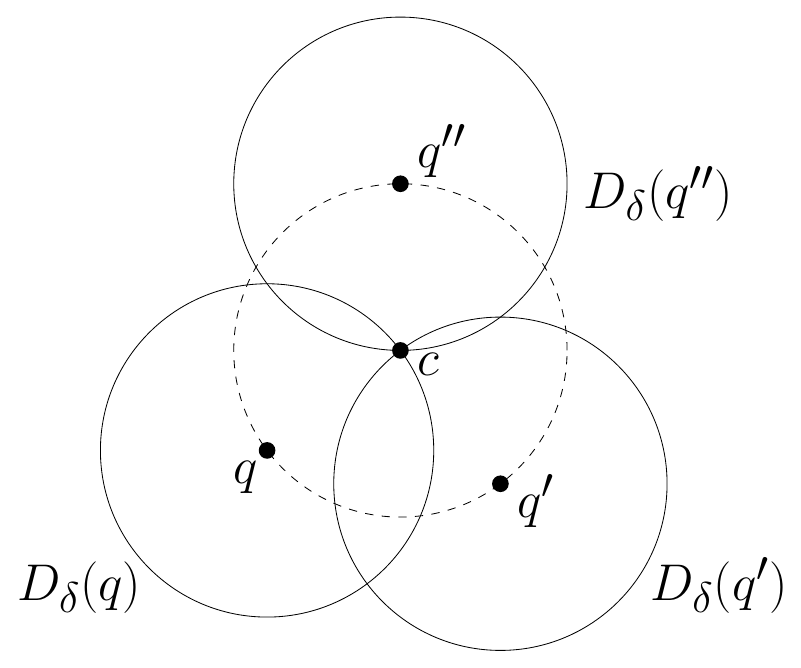}
		\caption{The ordering of $L_\delta(q)$ might change at radius $\delta$ such that $q,q',q'' \in Q$ lie on a common circle around some center $c$.}
		\label{fig:criticalvalues2}
	\end{subfigure}
	\hfill
	\begin{subfigure}[t]{0.32\textwidth}
		\includegraphics[width=\textwidth]{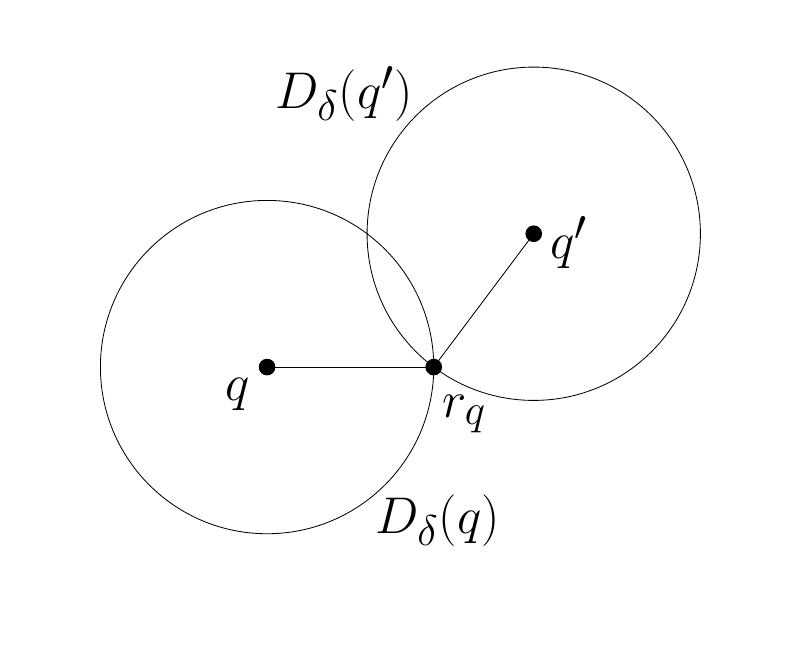}
		\caption{The ordering of $L_\delta(q)$ might change at radius $\delta$ such that $D_\delta(q')$ intersects $D_\delta(q)$ in $r_q$.}
		\label{fig:criticalvalues3}
	\end{subfigure}

	\caption{Critical values for the lists $L_\delta(q), q\in Q$.}
	\label{fig:criticalvalues}
\end{figure}

\paragraph{Step 2: Sorting $L_\optdelta(q)$.}
We use the following well-known variant of Meggido's parametric search that is due to Cole~\cite{Cole87}. 
\begin{lem}[implicit in \cite{Cole87}]\label{lem:colesorting}
Let parametric values $f_1(\delta), \dots, f_N(\delta)$ be given. Assume there is an unknown value $\optdelta > 0$ and a decision algorithm determining, given $\delta > 0$, whether $\optdelta \le \delta$ in time $T(N)$. If we can determine $f_i(\optdelta) \le f_j(\optdelta)$ for any $i,j\in [N]$ using only a constant number of queries to the decision algorithm, then in time $\Oh((N + T(N))\log n)$, we can sort $f_1(\optdelta), \dots, f_N(\optdelta)$ and obtain an interval $(\alpha, \beta]$ such that $\optdelta \in (\alpha, \beta]$ and no critical value for the sorted order of $f_1(\delta), \dots, f_N(\delta)$ is contained in $(\alpha, \beta)$.
\end{lem}
Consider first the problem of sorting $L_\optdelta(q)$ for some $q\in Q$. By the above technique, we only need to argue that we can determine whether some $\crossingop{\optdelta}{q,q'}{a}$ with $q' \in Q, a\in \{1,2\}$ precedes some $\crossingop{\optdelta}{q,q''}{b}$ with $q''\in Q, b \in \{1,2\}$ in $L_\optdelta(q)$. Note that $\crossingop{\delta}{q,q'}{a}, \crossingop{\delta}{q,q''}{b}$ move continuously on the boundary of $\disk_\delta(q)$ (while $\delta$ varies) and there are only constantly many choices of $\delta$ for which any of the points $\crossingop{\delta}{q,q'}{a}, \crossingop{\delta}{q,q''}{b}, r_q$ coincide (and thus the order might possibly change).\footnote{The important critical values for this step are the $\Oh(|Q|^3) = \Oh(n^6)$ radii of points with three (or more) points of $Q$ on their boundary. See Figure~\ref{fig:criticalvalues} for an illustration of all types of critical values.} By testing for these $\Oh(1)$ critical values of $\delta$, we can determine the order of $\crossingop{\optdelta}{q,q'}{a}, \crossingop{\optdelta}{q,q''}{b}, r_q$ on the boundary of $\disk_\delta(q)$, and thus resolve a comparison of $\crossingop{\optdelta}{q,q'}{a}$ and $\crossingop{\optdelta}{q,q''}{b}$ in the order of $L_\optdelta(q)$ using only a constant number of calls to the decision algorithm. 

Note that by arbitrarily choosing an order of $Q$, we may use Cole's sorting procedure (\lemref{colesorting}) to construct all lists $L_\optdelta(q), q\in Q$ simultaneously (we simply need to adapt the comparison function to compare $\crossingop{\optdelta}{q,q'}{a}, \crossingop{\optdelta}{\tilde{q},q''}{b}$ according to the order of $Q$ if $q\ne \tilde{q}$). Note that in this application of \lemref{colesorting}, we have $N = \sum_{q\in Q} |\crossings{\optdelta}(q)| = \Oh(|Q|^2) = \Oh(n^4)$. 

It follows that in time $\Oh((n^4 + T_\dec(n))\log n)$, we can obtain an interval $(\alpha_2, \beta_2]$ such that $\optdelta \in (\alpha_2, \beta_2]$, while $\beta_2$ is the only value for $\delta$ for which the combinatorial structure of $\arr_\delta$ changes in $(\alpha_2, \beta_2]$. Thus, $\optdelta = \beta_2$, as desired.

The overall running time of the above procedure amounts to $\Oh((n^4 + T_\dec(n)) \log n)$, which concludes the proof of \lemref{paramsearch}.

\section{Grid Reachability Data Structure}
\label{sec:reach-ds}

In this section, we prove \lemref{ds}, which we restate here for convenience. 
\dslemma*

The rough outline is as follows: We obtain the data structure by repeatedly splitting the free-space diagram into smaller blocks. This yields $\Oh(\log n)$ levels of blocks, where in each block we store reachability information from all ``inputs'' to the block (i.e., the lower-left boundary) to all ``outputs'' of the block (i.e., the upper-right boundary). Any change in the matrix $M$ is reflected only in $\Oh(\log n)$ blocks containing this position, thus, we can quickly update the information. This approach was pursued already by Ben Avraham et al.~\cite{avraham2015faster}. 

In addition, however, we need to maintain reachability of all terminals $\term$ to the inputs and from the outputs of each block. Surprisingly, we only need an additional storage of $\Oh(|\term|)$ per block. We show how to maintain this information also under updates and how it can be used by a divide and conquer approach to answer any reachability queries.

To this end, we start with some basic definitions (block structure, identifiers for each position, etc.) in \secref{ds-prepare}. We can then prove the succinct characterization of terminal reachability in \secref{ds-characterization}, which is the key aspect of our data structure. Given this information, we can define exactly what information we store for each block in \secref{ds-information}. We give algorithms computing the information for some block given the information for its children in \secref{ds-computeParentInformation}, which allows us to prove the initialization and update statements (i.e., \ref{enum:ds-init} and \ref{enum:ds-updates} of \lemref{ds}) in \secref{ds-updates}. Finally, \secref{ds-reachqueries} is devoted to the reachability queries, i.e., proving \ref{enum:ds-reachqueries} of \lemref{ds}.

\subsection{Basic Structures and Definitions}
\label{sec:ds-prepare}

Without loss of generality, we may assume that $n = 2^\kappa + 1$ for some integer $\kappa\in \mathbb{N}$. Otherwise, for any $n\times n$ matrix $M$ over $\{0,1\}$, we could define an $n' \times n'$ matrix $M'$ with (1) $n' = 2^\kappa + 1$ for some $\kappa\in \mathbb{N}$ with $n < n' \le 2n$ and (2) setting $M'_{i,j} = M_{i,j}$ for all $(i,j) \in [n]\times [n]$ and  setting $M'_{i,j} = 1$ if and only if $i=j$ for all $(i,j) \in [n']\times [n'] \setminus [n]\times [n]$. Clearly, existence of a monotone 1-path from $(1,1)$ to $(n,n)$ in $M$ is equivalent to existence of a monotone 1-path from $(1,1)$ to $(n',n')$ in $M'$. 

\paragraph{Canonical blocks.}
Let $I,J$ be intervals in $[n]$ with $n = 2^\kappa + 1$. We call $I \times J \subseteq [n] \times [n]$ a \emph{block}. In particular, we only consider blocks obtained by splitting the square $[n]\times [n]$ alternately horizontally and vertically until we are left with $2\times 2$ blocks. Formally, we define $\blocks_0 \coloneqq \{ ([n],[n]) \}$ and construct $\blocks_{\ell+1}$ inductively by splitting each block $B \in \blocks_{\ell}$ as follows: 
\begin{itemize}
\item For $\ell = 2i$ with $0 \le i < \kappa$, we have $B=(I,J)$ with $|I|= |J| = 2^{\kappa - i} + 1$. We split $J$ into intervals $J_1,J_2$, where $J_1$ contains the first $(2^{\kappa-i-1}+1)$ elements in $J$ and $J_2$ contains the last $(2^{\kappa-i-1}+1)$ elements in $J$ (thus $J_1$ and $J_2$ intersect in the middle element of $J$). Add $(I,J_1)$ and $(I, J_2)$ to $\blocks_{\ell+1}$. 
\item For $\ell = 2i + 1$ with $0 \le i < \kappa$, we have $B=(I,J)$ with $|I|= 2^{\kappa-i}+1$ and $|J| = 2^{\kappa-i-1} + 1$. Analogous to above, we split $I$ into two equal-sized intervals $I_1,I_2$, where $I_1$ contains the first $(2^{\kappa-i-1}+1)$ elements in $I$ and $I_2$ contains the last $(2^{\kappa-i-1}+1)$ elements in $I$. Add $(I_1,J)$ and $(I_2, J)$ to $\blocks_{\ell+1}$.
\end{itemize}

We let $\blocks \coloneqq \bigcup_{\ell = 0}^{2\kappa} \blocks_\ell$ be the set of \emph{canonical blocks}, and call each block $B \in \blocks_\ell$ a \emph{canonical block on level $\ell$}. The blocks $B_1 = (I_1,J), B_2 = (I_2, J) \in \blocks_{\ell+1}$ (or $B_1 = (I,J_1), B_2 = (I,J_2) \in \blocks_{\ell+1}$, respectively) obtained from $B= (I,J)\in \blocks_{\ell}$ are called the \emph{children} of $B$. See Figure \ref{fig:canonical_blocks}.

\begin{figure}
\centering
\includegraphics[width=\textwidth]{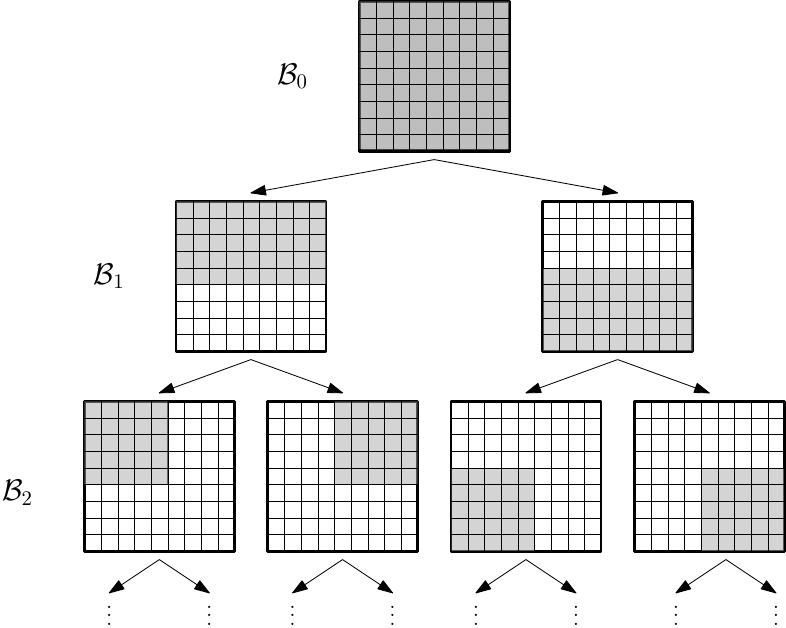}
\caption{The sets of canonical blocks $\mathcal{B}_0, \mathcal{B}_1, \mathcal{B}_2, \dots, \mathcal{B}_{2\kappa}$. We alternate between horizontal and vertical splits. Note that child blocks overlap at their boundary.}
\label{fig:canonical_blocks}
\end{figure}

\paragraph{Boundaries.}
For any $B = (I,J) \in \blocks$, we denote the lower left boundary of $B$ as $\Bin = \{ \min I \} \times J \cup  I \times \{ \min J \}$, and call each $p \in \Bin$ an \emph{input} of $B$. Analogously, we denote the upper right boundary of $B$ as $\Bout = \{ \max I \} \times J \cup  I \times \{ \max J \}$, and call each $q \in \Bout$ an \emph{output} of $B$. By slight abuse of notation, we define $\sBbound = |\Bin \cup \Bout|$ as the size of the boundary of $B$, i.e., the number of inputs and outputs of $B$.

If $B$ splits into children $B_1, B_2$, we call $\Bmid = \Bout_1 \cap \Bin_2$ the \emph{splitting boundary} of $B$. 

\begin{figure}
	\begin{subfigure}[b]{0.45\textwidth}
		\includegraphics[width=\textwidth]{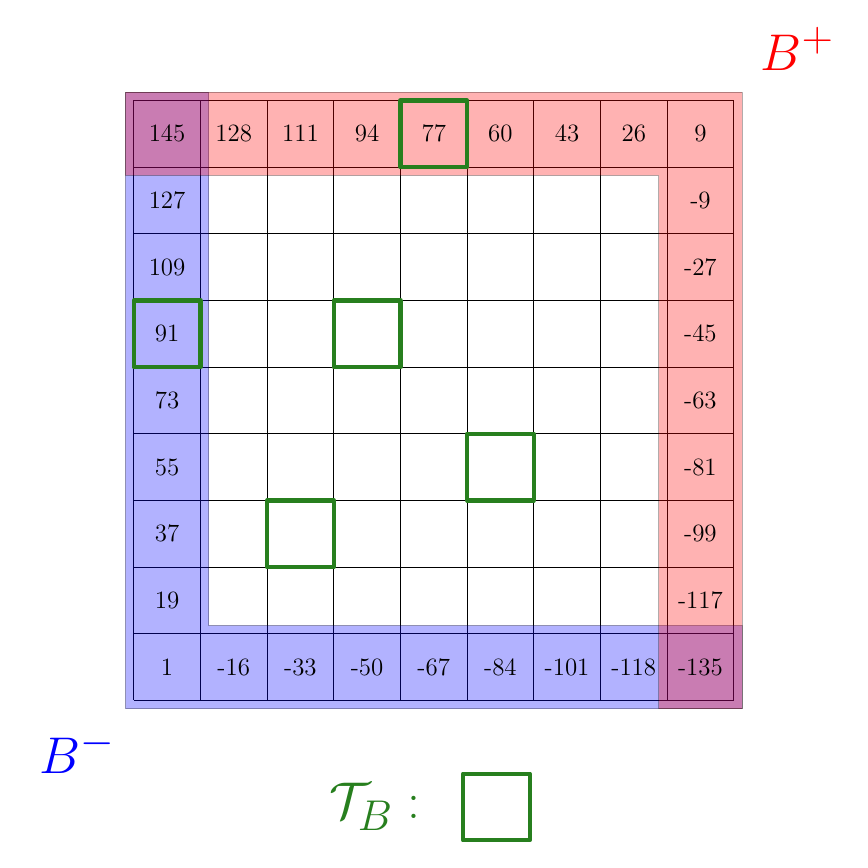}
\label{fig:blockdefs1}
	\end{subfigure}
	\hfill
	\begin{subfigure}[b]{0.45\textwidth}
		\includegraphics[width=\textwidth]{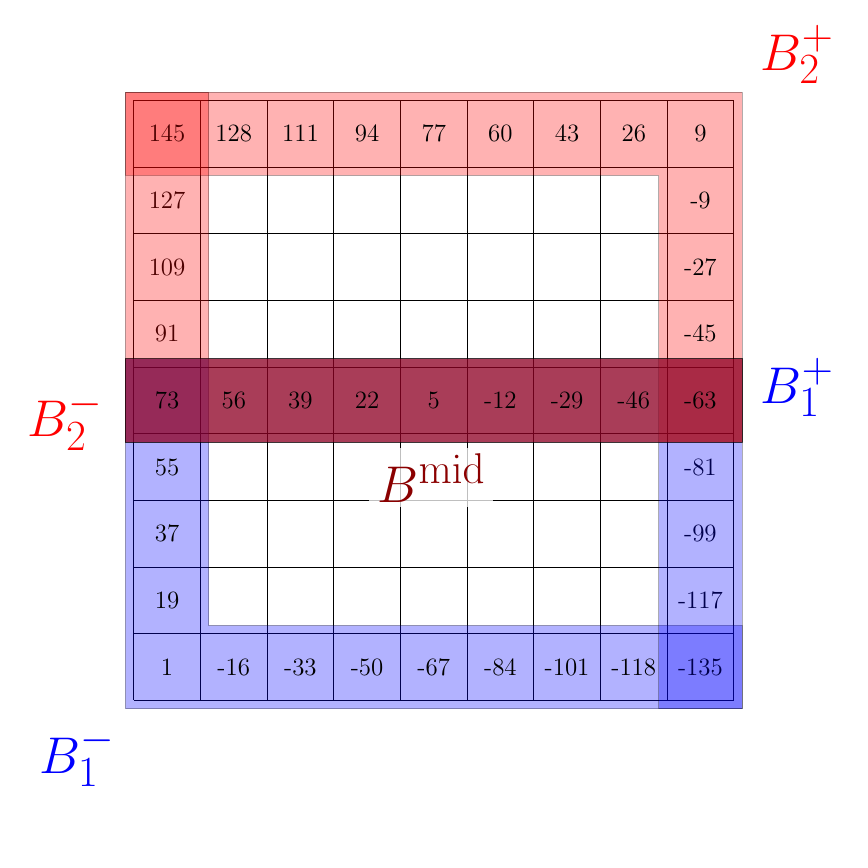}
\label{fig:blockdefs2}
	\end{subfigure}
	\caption{Structure of a block $B$}
	\label{fig:blockdefs}
\end{figure}

\paragraph{Indices.}
To prepare the description of this information, we first define, for technical reasons, \emph{indices} for all positions in $[n] \times [n]$. It allows us to give each position a unique identifier with the property that for any canonical block $B$, the indices yield a local ordering of the boundaries.

\begin{obs}
	Let $\ind : [n] \times [n] \to \mathbb{N}$, where for any point $p = (x,y) \in [n]\times[n]$, we set $\ind(p) \coloneqq (y-x) (2n) + x$. We call $\ind(p)$ the \emph{index of $p$}. This function satisfies the following properties:
\begin{enumerate}
	\item The function $\ind$ is injective, can be computed in constant time, and given $i = \ind(p)$, we can determine $\ind^{-1}(i) \coloneqq p$ in constant time.
\item For any $B \in \blocks$, $\ind$ induces an ordering of $\Bout$ in counter-clockwise order and  an ordering of $\Bin$ in clockwise order.
\end{enumerate}
\end{obs}

We refer to Figure~\ref{fig:blockdefs} for an illustration of a block $B$, its boundaries, and the indices of all positions.

\subsection{Reachability characterization}
\label{sec:ds-characterization}

Our aim is to construct a data structure $\diagram_{M, \term} = (\diagram_{M, \term}(B))_{B\in \blocks}$, where $\diagram_{M,\term}(B)$ succinctly describes reachability (via monotone 1-paths) between the boundaries $\Bin, \Bout$ and the terminals $\term_B \coloneqq \term \cap B$ inside $B$. In particular, we show that we only require space $\Oh(\sBbound + |\term_B|)$  to represent this information.  

To prepare this, we start with a few simple observations that yield a surprisingly simple characterization of reachability from any terminal to the boundary. 

\paragraph{Compositions of crossing paths.}

We say that we reach $q$ from $p$, written $\trav{p}{q}$, if there is a traversal $T = (t_1, \dots, t_\ell)$ with $t_1 = p$, $t_\ell = q$, and $t_i$ is free for all $1 < i < \ell$ (note that we do not require $t_1$ and $t_\ell$ to be free). We call such a slightly adapted notion of traversal a \emph{reach traversal}. By connecting the points of $T$ by straight lines, we may view $T$ also as a polygonal curve in $\R^2$. The following property is a standard observation for problems related to the Fr\'echet distance.

\begin{obs}\label{obs:crossingpaths}
Let $T_1, T_2$ be reach traversals from $p_1$ to $q_1$ and from $p_2$ to $q_2$, respectively. Then if $T_1$ and $T_2$ intersect, we have $\trav{p_1}{q_2}$ (and, symmetrically, $\trav{p_2}{q_1}$).
\end{obs}
\begin{proof}
Let $t \in [n]\times [n]$ be a free position in which $T_1, T_2$ intersect (observe that such a point with integral coordinates must exist unless $p_1 = p_2$ or $q_1 = q_2$; in the latter case, the claim is trivial). Note that $t$ splits $T_1, T_2$ into $T_1 = T^a_1 \concat T^b_1$ and $T_2 = T_2^a \concat T^b_2$ such that $T_1^a, T_2^a$ are reach traversals ending in $t$ and $T^b_1, T_2^b$ are reach traversal starting in $t$. By concatenating $T_1^a$ and  $T_2^b$, we obtain a reach traversal from $p_1$ to $q_2$. Symmetrically, $T^a_2 \concat T^b_1$ proves $\trav{p_2}{q_1}$. 
\end{proof}

Let $B \in \blocks$ and recall that $\ind(\cdot)$ orders $\Bout$ counter-clockwise. For any $p \in B$, we define $\sA(p) \coloneqq \min \{ \ind(q) \mid q\in \Bout, \trav{p}{q} \}$, and analogously $\sZ(p) \coloneqq \max \{ \ind(q) \mid q\in \Bout, \trav{p}{q}\}$ (note that $\sA(p)$ and $\sZ(p)$ correspond to the lowest/rightmost and highest/leftmost pointer, respectively, in \cite[Section 3.2]{AltG95}). These two values define a corresponding \emph{reachability interval} $\intvl(p) \coloneqq [\sA(p),\sZ(p)]$ that contains all $q\in \Bout$ with $\trav{p}{q}$. In the following analysis, we slightly abuse notation by also using $\ind(p)$ to denote the corresponding (unique) position $p\in [n]\times[n]$.

\begin{defn}
Let $p \in B$ with $\infty > \sA(p), \sZ(p) > - \infty$ and fix any reach traversals $T_A, T_Z$ from $p$ to $\sA(p)$ and $\sZ(p)$ such that we can write
\begin{align*}
T_A & = \Pcom \concat P_A',\\
T_Z & = \Pcom \concat P_Z',
\end{align*}
for some polygonal curves $\Pcom, P_A', P_Z'$ with $P_A', P_Z'$ non-intersecting. Let $\face$ be the face enclosed by $P_A', P_Z'$ and the path from $\sA(p)$ to $\sZ(p)$ on $\Bout$ (if $\sA(p) = \sZ(p)$, we let \face be the empty set). We define the \emph{reach region} of $p$ as 
\[
\reachr(p) \coloneqq \face \cup \Pcom. 
\]

\end{defn}

We refer to Figure~\ref{fig:reachregion} for an illustration. Observe that the desired traversals $T_A, T_Z$ for defining $\reachr(p)$ always exist: For any reach traversals $T'_A, T_Z'$ from $p$ to $\sA(p)$ and $\sZ(p)$, respectively, consider the latest point in which $T'_A, T'_B$ intersect, say $t$. We can define reach traversals $T_A$ and $T_Z$ by following $T'_A$ until $t$ and then following the remainder of $T_A'$ or $T_Z'$ to reach $\sA(p)$ or $\sZ(p)$, respectively. These traversals satisfy the conditions by construction. (Strictly speaking, any feasible choice for $T_A, T_Z$ gives a potentially different reach region $\reachr(p)$. However, any fixed choice will be sufficient for our proofs, e.g., choosing lexicographically smallest/largest traversals.)  

The following property generalizes an insightful property of reachability from the inputs to the outputs (cf. \cite[Lemma 10]{AltG95} and \cite[Corollary 4.2]{avraham2015faster}) to a similar property for reachability from arbitrary block positions to the outputs, using the same argument of crossing traversals.  

\begin{figure}
\centering
\includegraphics[width=0.5\textwidth]{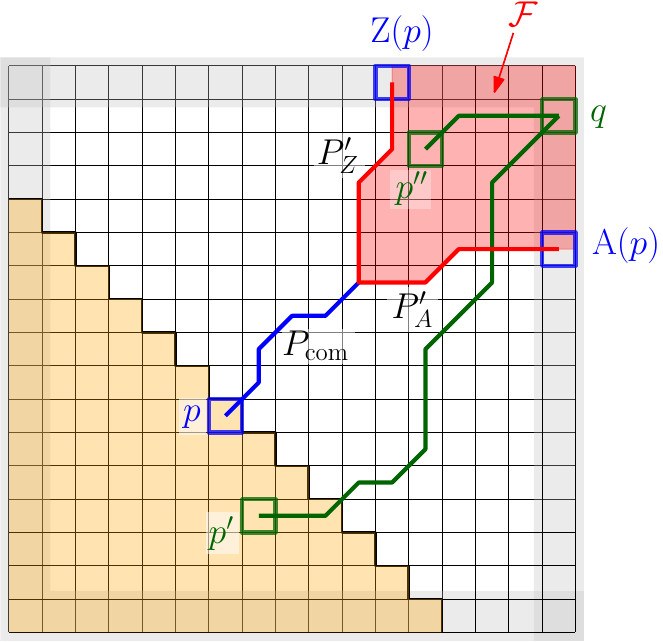}
\caption{Illustration of $\reachr(p)$, Proposition~\ref{prop:reachregion} and Lemma~\ref{lem:totalordermain}: Any reach traversal from $p' \notin \reachr(p)$ must cross $P'_A$ or $P'_Z$ to reach $q$. However, if $\trav{p''}{q}$ but $p'' \in \reachr(p)$, then $q$ might not be reachable from $p$. A sufficient condition for $p' \notin \reachr(p)$ is that $p'\ne p$ and $L(p') \le L(p)$ (indicated by the orange triangular area).}
\label{fig:reachregion}
\end{figure}

\begin{prop}\label{prop:reachregion}
Let $p, p' \in B$, $q \in \Bout$ with $\ind(q)\in \intvl(p)$ and $p' \notin \reachr(p)$. Then $\trav{p'}{q}$ implies $\trav{p}{q}$.
\end{prop}
\begin{proof}
The claim holds trivially if $\ind(q) = \sA(p)$ or $\ind(q) = \sZ(p)$. Thus, we may assume that $\sA(p) < \sZ(p)$, which implies that the face \face in $\reachr(p)$ is nonempty with $q \in \face$ and $p' \notin \face$. Hence any reach traversal $T$ from $p'$ to $q$ must cross the boundary of $\face$, in particular, the path $\mypath(T_A)$ or $\mypath(T_Z)$, where $T_A, T_Z$ both originate in $p$. By Observation~\ref{obs:crossingpaths}, this yields $\trav{p}{q}$.
\end{proof}

\paragraph{Reachability Labelling.}

We define a total order on nodes in $B$ that allows us to succinctly represent reachability towards $\Bout$ for any subset $S\subseteq B$ in space $\tOh(|S|+ |\Bout|)$. The key is a labelling $L: [n]\times[n] \to \mathbb{N}$, defined by $L((x,y)) = x+y$, that we call the \emph{reachability labelling}. For an illustration of the following lemma, we refer to Figure~\ref{fig:reachregion}.

\begin{lem}\label{lem:totalordermain}
Let $p=(x,y), p'=(x',y') \in B$ with $L(p') \le L(p)$ and $q\in \Bout$ with $\ind(q)\in \intvl(p)$. Then $\trav{p'}{q}$ implies $\trav{p}{q}$.
\end{lem}
\begin{proof}
The proof idea is to show that $L(p') \le L(p)$ implies that $p' \notin \reachr(p)$, and hence \propref{reachregion} shows the claim.
Note that by monotonicity of reach traversals, any point $r=(r_x,r_y) \in \reachr(p)$  satisfies $r_x \ge x$ and $r_y \ge y$. Thus, $p' \in \reachr(p)$ only if $x' \ge x$, $y' \ge y$, but this together with $x' + y' = L(p') \le L(p) = x+y$ implies $(x',y') = (x,y)$. Summarizing, we either have $p = p'$, which trivially satisfies the claim, or $p'\notin \reachr(p)$, which yields the claim by \propref{reachregion}.
\end{proof}

For any $S\subseteq B$, this labelling enables a surprisingly succinct characterization of which terminals in $S$ have reach traversals to which outputs in $\Bout$ by the following lemma (greatly generalizing a simpler characterization\footnote{In our language, this characterization is as follows: For any $p\in \Bin, q\in \Bout$, we have $\trav{p}{q}$ if and only if $\ind(q)\in \intvl(p)$ and there is some $p'\in \Bin$ with $\trav{p'}{q}$. It is easy to see that this characterization no longer holds if we replace $\Bin$ by an arbitrary subset $S\supseteq \Bin$; our approach instead relies on the reachability labelling to obtain a succinct and algorithmically tractable characterization.} for the special case of $S= \Bin$, cf. \cite[Lemma 10]{AltG95} and~\cite[implicit in Lemma 4.4]{avraham2015faster}). This is one of our key insights.

\begin{corbox}\label{cor:reachlabels}
Let $q\in \Bout$ and define $\ell(q) \coloneqq \min \{ L(p) \mid p \in B, \trav{p}{q}\}$. Then for any $p\in B$, we have 
\[\trav{p}{q} \qquad \text{if and only if} \qquad \ind(q) \in \intvl(p) \text{ and } \ell(q) \le L(p).\] 
\end{corbox}

\begin{proof}
Clearly, $\trav{p}{q}$ implies, by definition of $\sA(p), \sZ(p)$, and $\ell(q)$, that $\sA(p) \le \ind(q) \le \sZ(p)$ and $\ell(q) \le L(p)$.

Conversely, assume that $\ind(q) \in \intvl(p)$ and $\ell(q) \le L(p)$. Take any $p' \in B$ with $\trav{p'}{q}$ and $\ell(q) = L(p')$.  Thus we have $L(p') = \ell(q) \le L(p)$, $\ind(q) \in \intvl(p)$ and $\trav{p'}{q}$, which satisfies the requirements of \lemref{totalordermain}, yielding $\trav{p}{q}$.
\end{proof}

Given this characterization, we obtain a highly succinct representation of reachability information. Specifically, to represent the information which terminals in $S$ have reach traversals to which outputs in~$\Bout$, we simply need to store $\ell(q)$ for all $q\in \Bout$ as well as the interval $\intvl(p)$ for all $p \in S$. Thus, the space required to store this information amounts to only $\Oh(\sBbound + |S|)$, which greatly improves over a naive $\Oh(\sBbound \cdot |S|)$-sized tabulation.

\paragraph{Reverse Information.}

By defining $\Lrev((x,y)) = - L((x,y)) = - x - y$, we obtain a labelling with symmetric properties. In particular, define $\sArev(q) \coloneqq \min \{ \ind(p) \mid p\in \Bin, \trav{p}{q} \}$, $\sZrev(q) \coloneqq \max \{ \ind(p) \mid p\in \Bin, \trav{p}{q}\}$ and the corresponding \emph{reverse reachability interval} $\intvlrev(q)\coloneqq [\sArev(q),\sZrev(q)]$. It is straightforward to prove the following symmetric variant of \corref{reachlabels}.
\begin{cor}\label{cor:reachlabelsrev}
Let $p\in \Bin$ and define $\ellrev(p) \coloneqq \min \{ \Lrev(q) \mid q \in B, \trav{p}{q}\}$. Then for any $q\in B$, we have 
\[\trav{p}{q} \qquad \text{if and only if} \qquad \ind(p) \in \intvlrev(q) \text{ and } \ellrev(p) \le \Lrev(q).\] 
\end{cor}

\paragraph{Summary of Reachability Characterization}

As a convenient reference, we collect here the main notation and results introduced in this section.

\begin{sumbox}
For any $p\in B$, the \emph{reachability interval} $\intvl(p)$ is defined as $[\sA(p),\sZ(p)]$ with 
\begin{align*}
\sA(p) & = \min \{ \ind(q) \mid q\in \Bout, \trav{p}{q}\},\\
\sZ(p) & = \max \{ \ind(q) \mid q\in \Bout, \trav{p}{q}\}.
\end{align*}
(Note that $\intvl(p)$ might be empty if $\sA(p) = \infty, \sZ(p) = -\infty$.) For any $q\in \Bout$, its \emph{reachability level}~$\ell(q)$ is defined as 
\[\ell(q) = \min \{ L(p) \mid p \in B, \trav{p}{q}\},\]
where $L((x,y)) = x+y$. For any $p\in B, q\in \Bout$, we have the reachability characterization that  
\[\trav{p}{q} \qquad \text{if and only if} \qquad \ind(q) \in \intvl(p) \text{ and } \ell(q) \le L(p).\] 
For any $q\in B$, we have the \emph{reverse reachability interval} $\intvlrev(q) = [\sArev(q),\sZrev(q)]$ with
\begin{align*}
\sArev(q) & = \min \{ \ind(p) \mid p\in \Bin, \trav{p}{q} \}, \\
\sZrev(q) & = \max \{ \ind(p) \mid p\in \Bin, \trav{p}{q}\}.
\end{align*}
(Again, $\intvlrev(q)$ might be empty if $\sArev(q) = \infty, \sZrev(q) = -\infty$.) For any $p\in \Bin$, its \emph{reverse reachability level} $\ellrev(p)$ is defined as 
\[\ellrev(p) = \min \{ \Lrev(q) \mid q \in B, \trav{p}{q}\},\]
where $\Lrev((x,y)) = -x-y$. For any $p\in \Bin, q\in B$, we have the reachability characterization that  
\[\trav{p}{q} \qquad \text{if and only if} \qquad \ind(p) \in \intvlrev(q) \text{ and } \ellrev(p) \le \Lrev(q).\] 
\end{sumbox}

\subsection{Information stored at canonical block \boldmath $B$}
\label{sec:ds-information}

Using the characterization given in \correfs{reachlabels}{reachlabelsrev}, we can now describe which information we need to store for any canonical block $B\in \blocks$.

\begin{defnbox}\label{def:blockinfo}
Let $B \in \blocks$. The \emph{information stored at $B$} (which we denote as $\diagram_{M, \term}(B)$) consists of the following information: First, we store \emph{forward reachability information} consisting of,
\begin{itemize}
\item for every $p\in \Bin \cup \term_B$, the interval $\intvl(p)$, and
\item for every $q \in \Bout$, the reachability level $\ell(q)$.
\end{itemize}
Symmetrically, we store \emph{reverse reachability information} consisting of,
\begin{itemize}
\item for every $q\in \Bout \cup \term_B$, the interval $\intvlrev(q)$, and
\item for every $p \in \Bin$, the reverse reachability level $\ellrev(p)$.
\end{itemize}
Finally, if $B$ has children $B_1, B_2 \in \blocks$, where $B_1$ is the lower or left sibling of $B_2$, we additionally store
\begin{itemize}
\item an orthogonal range minimization data structure $\ORB$ storing, for each \emph{free} $q\in \Bmid = \Bout_1 \cap \Bin_2$, the value $\ellrev_2(q)$ under the key $(\ind(q), \ell_1(q))$. Here $\ell_1(q)$ denotes the forward reachability level in $B_1$, and $\ellrev_2(q)$ denotes the reverse reachability level in $B_2$. 
\end{itemize}
\end{defnbox}

\subsection{Computing Information at Parent From Information at Children}
\label{sec:ds-computeParentInformation}

We show how to construct the information stored at the blocks quickly in a recursive fashion.

\begin{lem}\label{lem:computeParentInformation}
Let $B \in \blocks$ with children $B_1, B_2$. Given the information stored at $B_1$ and $B_2$, we can compute the information stored at $B$ in time $\Oh( (\sBbound + |\term_B|)\log \sBbound )$.
\end{lem}
\begin{proof}
Without loss of generality, we assume that $B_1, B_2$ are obtained from $B$ by a vertical split (the other case is analogous) -- let $B_\lt, B_\rt$ denote the left and right child, respectively.
As a convention, we equip the information stored at $B_\lt, B_\rt$ with the subscript $\lt, \rt$, respectively, and write the information stored at $B$ without subscript. Furthermore, we let $\Bmidfree$ denote the set of \emph{free} positions of the splitting boundary $\Bmid = \Bout_\lt \cap \Bin_\rt$.

\paragraph{Computation of $\intvl(p)$.}
Let $p \in \Bin \cup \term_B$ be arbitrary. We first explain how to compute $\sA(p)$ (see Figure~\ref{fig:intervalcomp} for an illustration).
If $p \in B_\rt$, then $\sA(p) = \sA_\rt(p)$, since by monotonicity any $q\in \Bout$ with $\trav{p}{q}$ satisfies $q\in \Bout_\rt$. Thus, it remains to consider $p\notin B_\rt$. 

\begin{figure}
\centering
\includegraphics[width=0.8\textwidth]{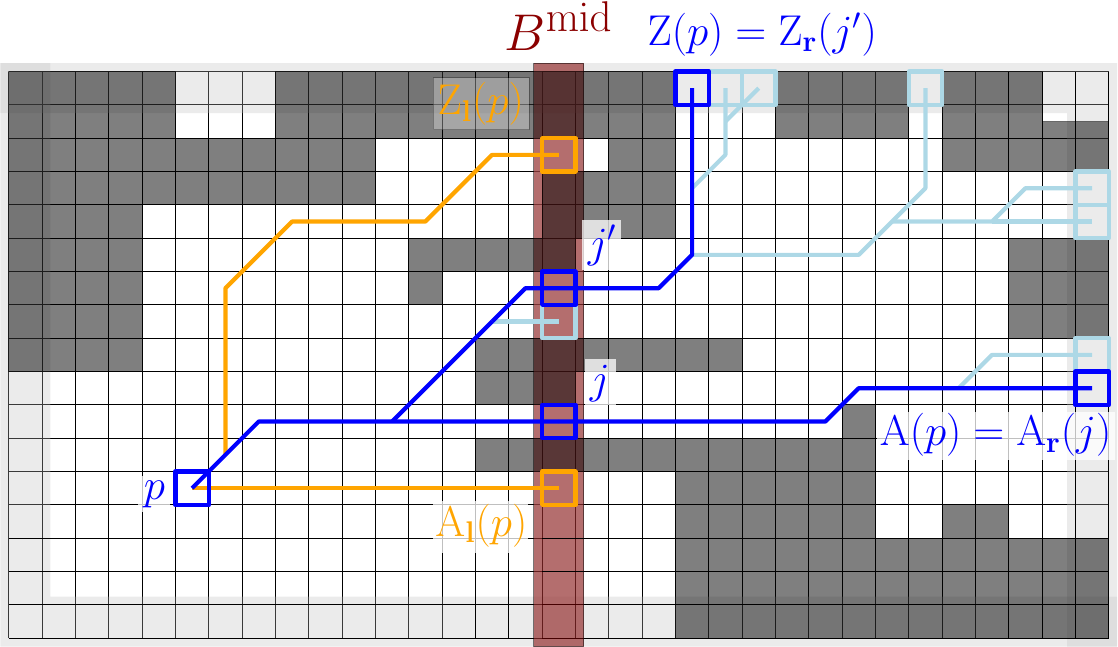}
\caption{Computation of $\intvl(p)$. To determine the smallest (largest) reachable index on $\Bout \cap B_\rt$,  we optimize, over all $j \in \Bmid$ with $\trav{p}{j}$, the smallest (largest) reachable index $\sA_\rt(j)$ ($\sZ_\rt(j)$) on $\Bout_\rt$.
	In this diagram, bright (dark) cells show free (non-free) positions. $\Bmid$ is the boundary shared by $B_l$ (left of $\Bmid$) and $B_r$ (right of $\Bmid$); see Figure \ref{fig:blockdefs} as a reminder of how we split boxes.
}
\label{fig:intervalcomp}
\end{figure}

We claim that for $p\notin B_\rt$, we have $A(p) = \min\{A_1(p), A_2(p)\}$,  where
\begin{align*}
A_1(p) & \coloneqq \min_{\substack{q\in \Bout \cap B_\lt,\\ \trav{p}{q}}} \ind(q) \\
A_2(p) & \coloneqq \min_{\substack{j \in \Bmidfree,\\\trav{p}{j}}} \; \min_{\substack{q\in \Bout \cap B_\rt,\\ \trav{j}{q}}} \ind(q)
\end{align*}
Indeed, this follows since each path starting in $p \in B_\lt$ and ending in $\Bout$ must end in $B_\lt$, or cross $\Bmid$ at some free $j \in \Bmid$ and end in $B_\rt$.

To compute $A_1(p)$ note that \corref{reachlabels} yields $A_1(p) = \min \{ \ind(q) \mid q\in \Bout \cap B_\lt, \ind(q) \in [\sA_\lt(p),\sZ_\lt(p)], \ell_\lt(q) \le L(p) \}$, which can be expressed as an orthogonal range minimization query. 

Likewise, to compute $A_2(p)$, note that $\Bout \cap B_\rt = \Bout_\rt$. Thus,
\[ \sA_2(p) = \min_{\substack{j \in \Bmidfree,\\ \trav{p}{j}}} \; \min_{\substack{q\in \Bout_\rt,\\\trav{j}{q}}} \; \ind(q) = \min_{\substack{j \in \Bmidfree,\\ \trav{p}{j}}} \; \sA_\rt(j) = \min_{\substack{j \in \Bmidfree,\\ \ind(j) \in \intvl_\lt(p), \ell_\lt(j) \le L(p)}} \sA_\rt(j), \]
where the second and last equalities follow from the definition of $\sA_\rt$ and \corref{reachlabels}, respectively. It follows that we can compute $\sA_2(p)$ using a simple orthogonal range minimization query.

Switching the roles of minimization and maximization, we obtain the analogous statements for computing $\sZ(p)$. We summarize the resulting algorithm for computing the reachability intervals $\intvl(p)$ for all $p\in \Bin \cup \term_B$ formally in Algorithm~\ref{alg:computeNewIntervals}. Its correctness follows from the arguments above.

Let us analyze the running time of Algorithm~\ref{alg:computeNewIntervals}: Observe that $|\Bmidfree| \le \sBbound$. Thus, we can construct the orthogonal range data structures $\OR_\sA$, $\OR_\sZ$, and $\OR_\tp$ in time $\Oh(\sBbound \log \sBbound)$ (see \secref{orthrange}). For each $p \in \Bin \cup \term_B$, we perform at most a constant number of two-dimensional orthogonal range minimization/maximization queries, which takes time $\Oh(\log \sBbound)$, followed by constant-time computation. The total running time amounts to $\Oh((\sBbound + |\term_B|) \log \sBbound)$.

\begin{algorithm} 
\begin{algorithmic}[1]
\State Build $\OR_{\sA}$ storing $\sA_\rt(j)$ under the key $(\ind(j), \ell_\lt(j))$ for all $j \in \Bmidfree$ (for minimization queries) 
\State Build $\OR_{\sZ}$ storing $\sZ_\rt(j)$ under the key $(\ind(j), \ell_\lt(j))$ for all $j \in \Bmidfree$ (for maximization queries) 
\State Build $\OR_{\tp}$ storing $\ind(q)$ under the key $(\ind(q), \ell_\lt(q))$ for all $q \in \Bout \cap B_\lt$ (for both queries) 
\For{$p \in (\Bin \cup \term_B)$}
\If{$p \in B_\rt$}
\State $\intvl(p) \gets \intvl_\rt(p)$
\Else
\State $\sA_1(p) \gets \OR_\tp.\min([\sA_\lt(p),\sZ_\lt(p)]\times (-\infty, L(p)])$
\State $\sA_2(p) \gets \OR_\sA.\min([\sA_\lt(p), \sZ_\lt(p)]\times (-\infty, L(p)])$
\State $\sA(p) \gets \min\{\sA_1(p), \sA_2(p)\}$
\Algblankline
\State $\sZ_1(p) \gets \OR_\tp.\max([\sA_\lt(p),\sZ_\lt(p)]\times (-\infty, L(p)])$
\State $\sZ_2(p) \gets \OR_\sZ.\max([\sA_\lt(p), \sZ_\lt(p)]\times (-\infty, L(p)])$ 
\State $\sZ(p) \gets \max\{\sZ_1(p), \sZ_2(p)\}$
\EndIf
\EndFor
\end{algorithmic}
\caption{Computing $\intvl(p) = [\sA(p), \sZ(p)]$ for all $p\in \Bin \cup \term_B$.}
\label{alg:computeNewIntervals}
\end{algorithm}

\paragraph{Computation of  $\ell(q)$.}
Let $q \in \Bout$ be arbitrary. If $q \in B_\lt$, then $\ell(q) = \ell_\lt(p)$, since by monotonicity every $p\in B$ with $\trav{p}{q}$ is contained in $B_\lt$. Thus, we may assume that $q \notin B_\lt$.

We claim that for $q\notin B_\lt$, we have $\ell(q) = \min\{\ell_1(q), \ell_2(q)\}$,  where
\begin{align*}
\ell_1(q) & \coloneqq \min_{\substack{p\in B_\rt,\\ \trav{p}{q}}} L(p) \\
\ell_2(q) & \coloneqq \min_{\substack{j \in \Bmidfree,\\\trav{j}{q}}} \; \min_{\substack{p\in B_\lt,\\ \trav{p}{j}}} L(p)
\end{align*}
Indeed, this follows since each path starting in $B$ and ending in $q \in B_\rt$ must start in $B_\rt$, or start in $B_\lt$ and cross $\Bmid$ at some free $j \in \Bmid$.

Observe that the definition of $\ell_1(q)$ coincides with the definition of $\ell_\rt(q)$. Thus it only remains to compute $\ell_2(q)$. We write
\[ \ell_2(q) = \min_{\substack{j \in \Bmidfree,\\ \trav{j}{q}}} \; \min_{\substack{p\in B_\lt,\\\trav{p}{j}}} \; L(p) = \min_{\substack{j \in \Bmidfree,\\ \trav{j}{q}}} \; \ell_\lt(j) = \min_{\substack{j \in \Bmidfree,\\ \ind(j) \in \intvlrev_\rt(q), \ellrev_\rt(j) \le \Lrev(q)}} \ell_\lt(j), \]
where the second and last equalities follow from the definition of $\ell_\lt(j)$ and \corref{reachlabelsrev}, respectively. It follows that we can compute $\ell_2(p)$ using a simple orthogonal range minimization query.
For an illustration of $\ell_2(q)$, we refer to Figure~\ref{fig:levelcomp}.

We summarize the resulting algorithm for computing the reachability levels $\ell(q)$ for all $q\in \Bout$ formally in Algorithm~\ref{alg:computeNewReachLevels}. Its correctness follows from the arguments above.

\begin{figure}
\centering
\includegraphics[width=0.8\textwidth]{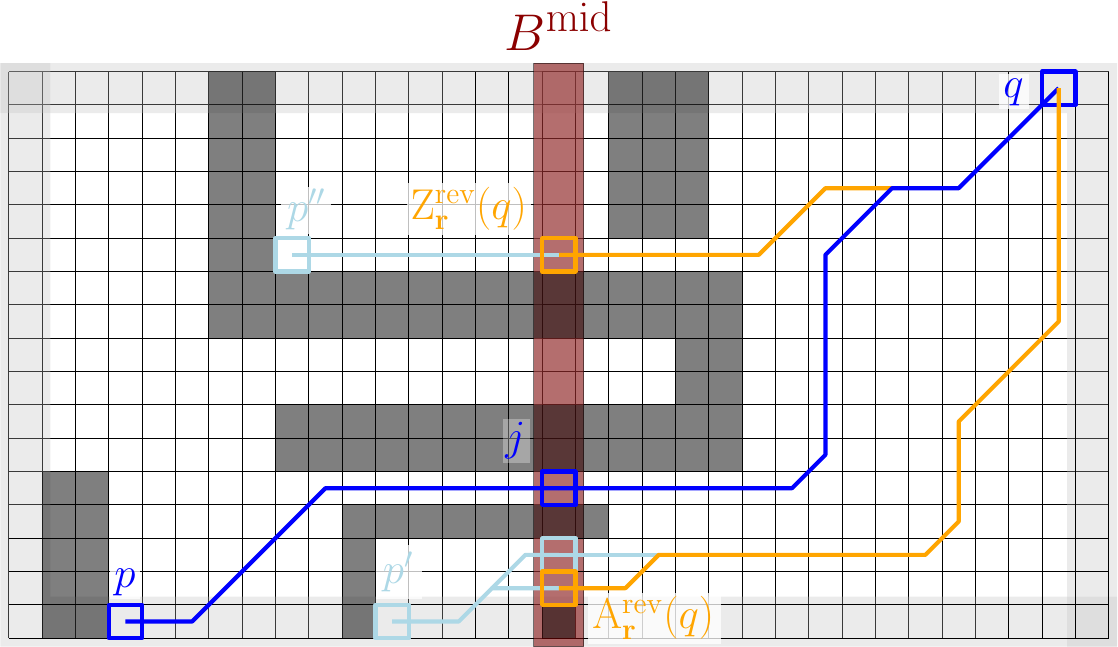}
\caption{Computation of $\ell(q)$. To determine the smallest label of a position in $B_\lt$ reaching $q$, we optimize, over all $j \in \Bmid$ with $\trav{j}{q}$, the smallest label $\ell_\lt(j)$ of a position $p\in B_\lt$ reaching~$j$.}
\label{fig:levelcomp}
\end{figure}

To analyze the running time of Algorithm~\ref{alg:computeNewReachLevels}, observe that $|\Bmidfree| \le \sBbound$. Thus, we can construct $\OR_\ell$ in time $\Oh(\sBbound \log\sBbound)$ (see \secref{orthrange}). For each $q\in \Bout$, we then perform at most one minimization query to $\OR_\ell$ in time $\Oh(\log \sBbound)$, followed by a constant-time computation. Thus, the total running time amounts to $\Oh(\sBbound \log \sBbound)$.  

\begin{algorithm} 
\begin{algorithmic}[1]
\State Build $\OR_\ell$ storing $\ell_\lt(j)$ under the key $(\ind(j), \ellrev_\rt(j))$ for all $j \in \Bmidfree$ (for minimization queries)
\For{$q \in \Bout$}
\If{$q \in B_\lt$}
\State $\ell(q) \gets \ell_\lt(q)$
\Else
\State $\ell_2(q) \gets \OR_{\ell}.\min([\sArev_\rt(q),\sZrev_\rt(q)]\times (-\infty, \Lrev(q)])$
\State $\ell(q) \gets \min\{\ell_\rt(q), \ell_2(q)\}$
\EndIf
\EndFor
\end{algorithmic}
\caption{Computing $\ell(q)$ for all $q\in \Bout$.}
\label{alg:computeNewReachLevels}
\end{algorithm}

\paragraph{Computation of reverse information.}
Switching the direction of reach traversals (which switches roles of inputs and outputs, $B_\lt$ and $B_\rt$, etc.) as well as $L$ and $\Lrev$, we can use the same algorithms to compute the reverse reachability information in the same running time of $\Oh((\sBbound + |\term_B|)\log \sBbound)$. 

\paragraph{Computation of $\ORB$.}

Finally, we need to construct the two-dimensional orthogonal range minimization data structure $\ORB$: Recall that $\ORB$ stores, for each $q \in \Bmidfree$, the value $\ellrev_\rt(q)$ under the key $(\ind(q), \ell_\lt(q))$ for minimization queries. Since $|\Bmidfree| \le \sBbound$, this can be done in time $\Oh(\sBbound \log \sBbound)$ (cf.~\secref{orthrange}).

\paragraph{Summary}
In summary, we can compute the information stored at $B$ (according to Definition~\ref{def:blockinfo}) from the information stored at $B_1$ and $B_2$ in time $\Oh((\sBbound + |\term_B|) \log \sBbound)$, as desired.
\end{proof}

\subsection{Initialization and Updates}
\label{sec:ds-updates}

We show how to construct our reachability data structure (using \lemref{computeParentInformation} that shows how to compute the information stored at some canonical block $B$ given the information stored at both children). Specifically, the following lemma proves \ref{enum:ds-init} of \lemref{ds}.

\begin{lem}\label{lem:ds-construct}
We can construct $\diagram_{M, \term}$ in time $\Oh(n^2+ |\term| \log^2 n)$.
\end{lem}

\begin{proof}
We use the obvious recursive algorithm to build $\diagram_{M, \term}$ in a bottom-up fashion using \lemref{computeParentInformation}. Recall that $n = 2^\kappa + 1$ for some $\kappa \in \mathbb{N}$. Note that for the blocks $B\in \blocks_{2\kappa}$ in the lowest level, we can compute the information stored in $B$ in constant time, which takes time $\Oh(|\blocks_{2\kappa}|) = \Oh(n^2)$ in total.

It remains to bound the running time to compute $\diagram_{M, \term}(B)$ for $B\in \blocks_{\ell}$ for $0 \le \ell < 2\kappa$. 
Observe that this running time is bounded by $\Oh(\sum_{\ell=0}^{2\kappa-1} \sum_{B \in \blocks_{\ell}} c_B)$ by \lemref{computeParentInformation}, where $c_B \coloneqq \sBbound \log \sBbound + |\term_B| \log \sBbound$.

Let $0 \le \ell < 2\kappa$. By construction, we have  $|\blocks_{\ell}| = 2^{\ell}$. Furthermore, for any $B\in \blocks_{\ell}$, observe that its side lengths are bounded by $2^{\kappa - \lfloor \ell/2 \rfloor} + 1$, and thus $\sBbound \le 4 \cdot 2^{\kappa - \lfloor \ell/2 \rfloor} \le 2^{\kappa- \ell/2 + 3}$. Hence, we may compute
\begin{align}
\sum_{\ell = 0}^{2\kappa-1} \sum_{B \in \blocks_{\ell}} \sBbound \log \sBbound & \le \sum_{\ell = 0}^{2\kappa-1} |\blocks_{\ell}| 2^{\kappa-\ell/2+3} (\kappa-\ell/2+3) \nonumber\\
 & = \sum_{\ell = 0}^{2\kappa-1} 2^{\kappa+\ell/2+3} (\kappa-\ell/2+3) \label{eq:shortcut}\\
 & \le 2\left(\sum_{i = 0}^{\kappa} 2^{\kappa+i+3} (\kappa-i+3)\right) \nonumber\\
 & = 2(2^{2(\kappa + 3)} - 2^{\kappa+3}(\kappa + 5)) = \Oh(2^{2\kappa}) = \Oh(n^2). \nonumber
\end{align}
Furthermore, we have
\[ \sum_{\ell = 0}^{2\kappa-1} \sum_{B \in \blocks_{\ell}} |\term_B| \log \sBbound  \le \sum_{\ell = 0}^{2\kappa-1} 4|\term| (\kappa - \ell/2 + 3) = \Oh(|\term|\kappa^2) = \Oh(|\term|\log^2 n), \]
where we used that $\sum_{B\in \blocks_\ell} |\term_B| \le 4|\term|$ (as any position in $[n]\times [n]$ is shared by at most 4 blocks at the same level).
In total, we obtain a running time bound of $\Oh(n^2 + |\term|\log^2 n)$, as desired.
\end{proof}

With very similar arguments, we can prove \ref{enum:ds-updates} of \lemref{ds}.

\begin{lem} 
Let $M,M'$ be any $n\times n$ 0-1-matrices differing in at most $k$ positions and $\term, \term' \subseteq [n]\times [n]$ be any sets of terminals of size $k$. Given the data structure $\diagram_{M,\term}$, the set $\term'$, as well as the set $\Delta$ of positions in which $M$ and $M'$ differ,  we can update $\diagram_{M,\term}$ to $\diagram_{M',\term'}$ in time $\Oh(n\sqrt{k} \log n + k \log^2 n)$.
\end{lem}
\begin{proof}
Set $X\coloneqq \Delta \cup \term \cup \term'$ and note that $|X| = \Oh(k)$. Observe that for any $B \in \blocks$ with $B\cap X= \emptyset$, we have $\diagram_{M, \term}(B) = \diagram_{M', \term'}(B)$, since the information stored at this block does not depend on any changed entry in $M$ and does not contain any of the old or new terminals. Thus, we only need to update $\diagram_{M,\term}(B)$ to $\diagram_{M',\term'}(B)$ for all $B\in \blocks$ with $B \cap X \ne \emptyset$. We do this by computing the information for these blocks in a bottom-up fashion analogously to \lemref{ds-construct}. Specifically, for any lowest-level block $B\in \blocks_{2\kappa}$ with $B\cap X \ne \emptyset$, we can compute the information stored in $B$ in constant time. Since there are at most $4|X|$ such blocks, this step takes time $\Oh(|X|) = \Oh(k)$ in total.

It remains to bound the running time to compute $\diagram_{M, \term}(B)$ for $B\in \blocks_{\ell}$ with $B\cap X \ne \emptyset$, where $0 \le \ell < 2\kappa$. For any such $B$, we let again $c_B \coloneqq \sBbound \log \sBbound + |\term_B| \log \sBbound$.
Observe that the running time for the remaining task is thus bounded by $\Oh(\sum_{\ell=0}^{2\kappa-1} \sum_{B \in \blocks_{\ell}, B \cap X \ne \emptyset} c_B)$ by \lemref{computeParentInformation}.

We do a case distinction into $0 \le \ell < \bell$ and $\bell \le \ell < 2\kappa$ where $\bell \coloneqq \lfloor \log k \rfloor$. For the first case, we bound
\begin{align*}
\sum_{\ell = 0}^{\bell-1} \sum_{\substack{B \in \blocks_{\ell},\\ B \cap X \ne \emptyset}} \sBbound \log \sBbound & \le \sum_{\ell = 0}^{\bell-1} \sum_{B \in \blocks_{\ell}} \sBbound \log \sBbound \\
 & \le \sum_{i = 0}^{\bell - 1} 2^{\kappa+i/2+3} (\kappa-i/2+3) \\
 & \le \left(\sum_{i = 0}^{\bell - 1} 2^{i/2}\right) 2^{\kappa+3}\kappa  = (1+\sqrt{2}) (2^{\bell/2}-1)  2^{\kappa+3} \kappa  = \Oh(\sqrt{k} n \log n),
\end{align*}
where the second inequality is derived as in~\eqref{eq:shortcut}. Recall that for any $0 \le \ell < 2\kappa$, there are at most $4|X|$ blocks $B\in \blocks_{\ell}$ with $B\cap X \ne \emptyset$ and for any $B\in \blocks_\ell$, we have $\sBbound \le 2^{\kappa-\ell/2+3}$. We compute
\begin{align*}
\sum_{\ell = \bell}^{2\kappa-1} \sum_{\substack{B \in \blocks_{\ell},\\ B \cap X \ne \emptyset}} \sBbound \log \sBbound & \le \sum_{\ell = \bell}^{2\kappa-1} 4|X| 2^{\kappa-\ell/2+3} (\kappa-\ell/2+3)\\
 & \le  4|X|2^{\kappa-\bell/2+3}(\kappa+3) \cdot \sum_{\ell = 0}^{2\kappa-\bell- 1} 2^{-\ell/2}  \\
 & = \Oh\left(|X|\frac{n}{\sqrt{k}}\log n\right) = \Oh(\sqrt{k}n\log n). 
\end{align*}
Furthermore, as in the proof of \lemref{ds-construct}, we again compute
\[ \sum_{\ell = 0}^{2\kappa-1} \sum_{\substack{B \in \blocks_{\ell},\\B \cap X \ne \emptyset}} |\term_B| \log \sBbound  \le \sum_{\ell = 0}^{2\kappa-1} 4|\term| (\kappa - \ell/2 + 3) = \Oh(|\term|\kappa^2) = \Oh(|\term|\log^2 n). \]
Thus, in total we obtain a running time of $\Oh(k + n\sqrt{k}\log n + |\term|\log^2 n) = \Oh(n \sqrt{k}\log n + k \log^2 n)$.
\end{proof}

\subsection{Reachability queries}
\label{sec:ds-reachqueries}

It remains to show how to use the information stored at all canonical blocks to answer reachability queries quickly. Specifically, the following lemma proves \ref{enum:ds-reachqueries} of \lemref{ds}.

\begin{lem}\label{lem:reachqueries}
Given $\diagram_{M,\term} = (\diagram_{M,\term}(B))_{B \in \blocks}$, we can answer reachability queries for $F\subseteq \term$ in time $\Oh( |\term| \log^3 n)$.
\end{lem}

Recall that we aim to determine whether there is a monotone path in $M$ using only positions $(i,j)$ with $M_{i,j} = 1$ or $(i,j) \in F$, i.e., we view $F$ as a set of \emph{free terminals} (typically,  $(i,j) \in F$ is a \emph{non-free} position). In this section we assume, without loss of generality, that $(1,1), (n,n) \in \term_B$ (whenever we construct/update to the data structure $\diagram_{M, \term}$, we may construct/update to $\diagram_{M, \term \cup \{(1,1),(n,n)\}}$ in the same asymptotic running time).

For any block $B\in \blocks$, $S\subseteq F \subseteq \term_B$, we define the function $\reach(B,S,F)$ that returns the set
\[ R \coloneqq \{ t\in F \mid \exists f_1, \dots, f_\ell \in F: f_1 \in S, f_\ell = t, \travseq{f_1}{f_2}{f_\ell}\}, \]
i.e., we interpret $S$ as a set of admissible \emph{starting positions} for a reach traversal and ask for the set of positions reachable from $S$ using only free positions or free terminals. We call any such position \emph{$F$-reachable from $S$}. (Recall that in the definition of $\trav{p}{q}$, only the intermediate points on a reach traversal from $p$ and $q$ are required to be free, while the endpoints $p$ and $q$ are allowed to be non-free.)

We show that $\reach(B,S,F)$ can be computed in time $\Oh( |\term_B| \log^3 n)$. Given this, we can answer any reachability query in the same asymptotic running time: the reachability query asks whether there is a sequence $f_1, \dots, f_\ell \in F\cup \{(1,1),(n,n)\}$ such that (i)  $f_1 = (1,1)$ and $f_\ell = (n,n)$, (ii) both $(1,1)$ and $(n,n)$ are free positions or contained in $F$ and (iii) $\travseq{f_1}{f_2}{f_\ell}$. Since (ii) can be checked in constant time, it remains to determine whether 
\[(n,n) \in \reach([n]\times[n],\; \{(1,1)\},\; F \cup \{(1,1), (n,n)\}).\]

\subsubsection{Computation of $\reach(B,S,F)$}

To compute $\reach(B,S,F)$, we work on the recursive block structure of $\diagram_{M, \term}$. Specifically, consider any canonical block $B \in \blocks$ (containing some free terminal) with children $B_1, B_2$. The (somewhat simplified) approach is the following: We first (recursively) determine all free terminals that are $F$-reachable from $S$ in $B_1$ and call this set $R_1$. Then, we determine all free terminals in $B_2$ that are (directly) reachable from $R_1$ and call this set $T_2$. Finally, we (recursively) determine all free terminals in $B_2$ that are $F$-reachable from $T_2 \cup (S\cap B_2)$ and call this set $R_2$. The desired set of free terminals that are $F$-reachable from $S$ is then $R_1\cup R_2$.  The main challenge in this process is the computation of the set $T_2$; this task is solved by the following lemma.

\begin{lem}\label{lem:singlestepreach}
Let $B \in \blocks$ be a block with children $B_1,B_2$. Given $S \subseteq B_1 \setminus \Bmid$ and $F \subseteq B_2 \setminus \Bmid$ with $S,F \subseteq \term_B$, we can compute the set 
\[T = \{ t\in F \mid \exists s \in S: \trav{s}{t} \} \]
in time $\Oh( |\term_B| \log^2 n)$. We call this procedure $\singleStepReach(B,S,F)$.
\end{lem}

We postpone the proof of this lemma to \secref{singlestepreach} and first show how this yields an algorithm for $\reach$, and thus, for reachability queries. 

\begin{proof}[Proof of \lemref{reachqueries}] 
We claim that Algorithm~\ref{alg:reach} computes $R$ in time $\Oh( |\term| \log^3 n)$.

\begin{algorithm} 
\begin{algorithmic}[1]
\Function{$\reach$}{$B, S, F$}
\If{$F = \emptyset$}
\State \Return $\emptyset$
\ElsIf{$B$ is a $2\times 2$ block}
\State Compute $R$ by checking all possibilities
\State \Return $R$
\Else \Comment{$B$ splits into child blocks $B_1, B_2$}
\State $S_1 \gets S \cap B_1$, $S_2 \gets S \cap B_2$
\State $R_1 \gets \reach(B_1, S_1, F \cap B_1)$
\State $T_2 \gets \singleStepReach(B, R_1 \setminus \Bmid, F \setminus B_1)$
\State $R_2 \gets \reach(B_2, S_2 \cup T_2 \cup (R_1\cap \Bmid), F \cap B_2)$
\State \Return $R_1 \cup R_2$
\EndIf
\EndFunction
\end{algorithmic}
\caption{Computing $\reach(B,S,F)$ for $B\in \blocks$, $S\subseteq F \subseteq \term_B$.}
\label{alg:reach}
\end{algorithm}

To ease the analysis, we introduce the shorthand that $\travsub{s}{t}{F}$ if and only if there are $f_1, \dots, f_\ell \in F$ with $f_1 = s, f_\ell = t$ and $\travseq{f_1}{f_2}{f_\ell}$, i.e., $t$ is $F$-reachable from $s$. 

We show that Algorithm~\ref{alg:reach} computes $R = \{ t \in F \mid \exists s\in S, \travsub{s}{t}{F} \}$ inductively: The base case for $2\times 2$ blocks $B$ holds trivially. Otherwise, by inductive assumption, we have  
\[ R_1 = \{ t\in F \cap B_1 \mid \exists s \in S \cap B_1, \travsub{s}{t}{F\cap B_1}\}.\]
Note that by definition of $\singleStepReach$, we furthermore have
\[ T_2 = \{ t\in F\setminus B_1 \mid \exists s \in R_1 \setminus \Bmid, \trav{s}{t}  \}. \]
Finally, by inductive assumption,
\begin{align*}
R_2  = \; &  \{ t \in F\cap B_2 \mid \exists s \in S \cap B_2, \travsub{s}{t}{F\cap B_2} \} \; \cup \\
& \{ t \in F\cap B_2 \mid \exists s \in T_2, \travsub{s}{t}{F\cap B_2} \} \; \cup \\
& \{ t \in F\cap B_2 \mid \exists s \in R_1\cap \Bmid, \travsub{s}{t}{F\cap B_2} \}.
\end{align*}
First, we show that any $t\in R_1 \cup R_2$ is contained in $R$: If $t \in R_1$, then there is some $s\in S \cap B_1 \subseteq S$ with $\travsub{s}{t}{F\cap B_1}$ (trivially implying $\travsub{s}{t}{F}$), and thus $t\in R$. Likewise, if $t\in T_2$, then there is some $t' \in R_1\setminus \Bmid$ with $\trav{t'}{t}$. Since $t'\in R_1$, there must exist some $s\in S$ with $\travsub{s}{t'}{F}$. Thus $\travsub{s}{t'}{F}$ and $\trav{t'}{t}$ yields $\travsub{s}{t}{F}$ and $t\in R$. Finally, if $t\in R_2$, there exists some $t'$ with $\travsub{t'}{t}{F\cap B_2}$ and either $t' \in S\cap B_2$, $t' \in T_2$, or $t' \in R_1\cap \Bmid$. In all these cases, there is some $s\in S$ with $\travsub{s}{t'}{F}$. Hence $\travsub{s}{t'}{F}$ and $\travsub{t'}{t}{F\cap B_2}$ imply $\travsub{s}{t}{F}$, placing $t$ in $R$. 

We proceed to show the converse direction that any $t\in R$ is contained in $R_1 \cup R_2$: Let $s\in S$ with $\travsub{s}{t}{F}$. If $t\in F \cap B_1$, then $\travsub{s}{t}{F}$ is equivalent to $\travsub{s}{t}{F \cap B_1}$ and $s\in S\cap B_1$ (by monotonicity). Thus, $t \in R_1$. It only remains to consider the case that $t\in F\setminus B_1$. If $s\in S\cap B_2$, then again my monotonicity $\travsub{s}{t}{F\cap B_2}$ must hold, which implies $t\in R_2$. Otherwise, we have $s \in S\setminus B_2$. Since additionally $t\in F\setminus B_1$, there must exist either (1) some $r \in F\cap \Bmid$ with $\travsub{s}{r}{F\cap B_1}$ and $\travsub{r}{t}{F\cap B_2}$ or (2) some $t'\in F \setminus B_2$, $t'' \in F \setminus B_1$ with $\trav{\travsub{s}{t'}{F\cap B_1}}{\travsub{t''}{t}{F\cap B_2}}$ (by monotonicity). For (1), note that $r\in R_1$ (as shown above), and thus $t\in R_2$. For (2), note that $t' \in R_1 \setminus B_2 = R_1 \setminus \Bmid$ (as $s\in S\cap B_1,t'\in F\setminus B_2$ and $\travsub{s}{t'}{F\cap B_1}$), $t'' \in T_2$ (as $t' \in R_1 \setminus \Bmid$, $t''\in F\setminus B_1$ and $\trav{t'}{t''}$) and finally $t\in R_2$ (as $t'' \in T_2$ and $\travsub{t''}{t}{F\cap B_2}$), as desired.

We analyze the running time of a call $\reach(B_0, S_0, F_0)$. Let $T(B) = \Oh( |\term_B| \log^2 n) $ denote the running time of $\singleStepReach(B,S,F)$ for arbitrary $S, F$. Observe that the running time of $\reach(B_0, S_0, F_0)$ is bounded by 
\begin{equation} \label{eq:reachbound}
\sum_{\substack{B\in \blocks,\\ \term_B \ne \emptyset}} \Oh(T(B)), 
\end{equation}
as for $\term_B = \emptyset$, we have $F \subseteq \term_B = \emptyset$, which is a base case of $\reach(\cdot)$. To bound the above term, fix any $\ell$, and note that for any $t \in \term$, there are at most 4 level-$\ell$ blocks $B\in \blocks_\ell$ with $t\in \term_B$ (if $t$ is on the boundary of some block $B \in \blocks_\ell$, it is shared between different blocks; however, any position is shared by at most 4 blocks). Thus $\sum_{B \in \blocks_\ell} |\term_B| \le 4|\term|$. Thus, \eqref{eq:reachbound} is bounded by
\[ \sum_{\ell=0}^{2\log(n-1)}  \sum_{\substack{B\in \blocks_\ell,\\ \term_B \ne \emptyset}} \Oh(|\term_B| \log^2 n ) =  \sum_{\ell=0}^{2\log(n-1)} \Oh(|\term|\log^2 n) = \Oh(|\term|\log^3 n).\] 

\end{proof}

\subsubsection{Computation of $\singleStepReach(B,S,F)$} \label{sec:singlestepreach}
It remains to prove \lemref{singlestepreach} to conclude the proof of \lemref{reachqueries}.
\begin{proof}[Proof of \lemref{singlestepreach}]
 
Consider $B\in \blocks$. We only consider the case that $B$ is split vertically (the other case is symmetric); let $B_\lt, B_\rt$ denote its left and right sibling, respectively.  
Let $S \subseteq B_\lt \setminus \Bmid$, $F\subseteq B_\rt \setminus \Bmid$ with $S,F \subseteq \term_B$ be arbitrary. We use notation (subscripts $\lt, \rt$, etc.) as in the proof of \lemref{computeParentInformation}. 

Observe that for any $s\in S, f\in F$, we have that $\trav{s}{f}$ if and only if there exists some $j\in \Bmidfree$ with $\trav{s}{j}$ and $\trav{j}{f}$. To introduce some convenient conventions, let $\Jmid = \{j_1, \dots, j_N\}$, where $j_1, \dots, j_N$ is the sorted sequence of $\ind(q)$ with $q\in \Bmidfree$. We call $J \subseteq \Jmid$ an interval of $\Jmid$ if $J = \{j_a, j_{a+1}, \dots, j_b\}$ for some $1 \le a \le b \le N$ and write it as $J = [j_a, j_b]_{\Jmid}$ (i.e., $[j_a, j_b]_{\Jmid}$ simply disregards any indices in $[j_a, j_b]$ representing positions not in $\Bmidfree$).

Consider any interval $J$ of $\Jmid$ with the property that for all $s\in S$ we either have $J \cap \intvl_\lt(s) = J$ or $J\cap \intvl_\lt(s) = \emptyset$ and for all $f\in F$ we either have $J\cap \intvl^\rev_\rt(f) = J$ or $J\cap \intvl^\rev_\rt(f) = \emptyset$. We call such a $J$ an \emph{$(S,F)$-reach-equivalent} interval. Note that by splitting $\Jmid$ right before and right after all points $\sA(s), \sZ(s)$ with $s\in S$ and $\sArev(f), \sZrev(f)$ with $f\in F$, we obtain a partition of $\Jmid$ into $(S,F)$-reach-equivalent intervals $J_1, \dots, J_\ell$ with $\ell = \Oh(|S \cup F|) = \Oh(|\term_B|)$.\footnote{To be more precise, we start with the partition $\partJ$ consisting of the singleton $\Jmid$. We then iterate over any point $j$ among $\sA(s), \sZ(s), s\in S$ and $\sArev(f), \sZrev(f), f\in F$, and replace the interval $J=[j_a,j_b]_\Jmid \in \partJ$ containing $j$ by the three intervals $[j_a, j)_{\Jmid}, \{j\}, (j, j_b]_{\Jmid}$, where the first and the last interval may be empty.}

\begin{claim}
Let $J$ be an $(S,F)$-reach-equivalent interval $J$. Let $R^J$ be the set of $t\in F$ reachable from $S$ via $J$, i.e., $R^J \coloneqq \{ t\in F \mid \exists s \in S, j \in J: \trav{\trav{s}{j}}{t} \}$. Define 
\[ \ell^J \coloneqq \min_{\substack{j \in J,\\\exists s \in S: \trav{s}{j}}} \ellrev_\rt(j).\]
We have 
\begin{equation}\label{eq:RJ}
R^J = \{ t \in F \mid J\subseteq \intvlrev_\rt(t), \ell^J \le \Lrev(t)\}. 
\end{equation}
\end{claim}
\begin{proof}
See Figure~\ref{fig:singlestepreach} for an illustration.
Indeed, for any $t \in F$ with $J\subseteq \intvlrev_\rt(t)$ and $\ell^J \le \Lrev(t)$, consider any $j \in J$ with $\ellrev_\rt(j) = \ell^J$ and $\trav{s}{j}$ for some $s \in S$. Then we have $j \in J \subseteq \intvlrev_\rt(t)$ and $\ellrev_\rt(j) = \ell^J \le \Lrev(t)$. Thus by \corref{reachlabelsrev}, $\trav{j}{t}$, which together with $\trav{s}{j}$ implies $\trav{\trav{s}{j}}{t}$, as desired. For the converse, let $t \in F$ with $\trav{\trav{s}{j}}{t}$ for some $s\in S, j \in J$. Then by definition of $\ell^J$, we obtain $\ell^J \le \ellrev_\rt(j)$. Furthermore, by \corref{reachlabelsrev}, $\trav{j}{t}$ implies that $j \in \intvlrev_\rt(t)$ with $\Lrev(t) \ge \ellrev_\rt(j) \ge \ell^J$. Note that $j \in \intvlrev_\rt(t)$ implies $J \subseteq \intvlrev_\rt(t)$ (as $J$ is $(S,F)$-reach-equivalent), thus we obtain that $J \subseteq \intvlrev_\rt(t)$ and $\ell^J \le \ellrev_\rt(j)$, as desired.
\end{proof}
\begin{figure}
\centering
\includegraphics[width=0.8\textwidth]{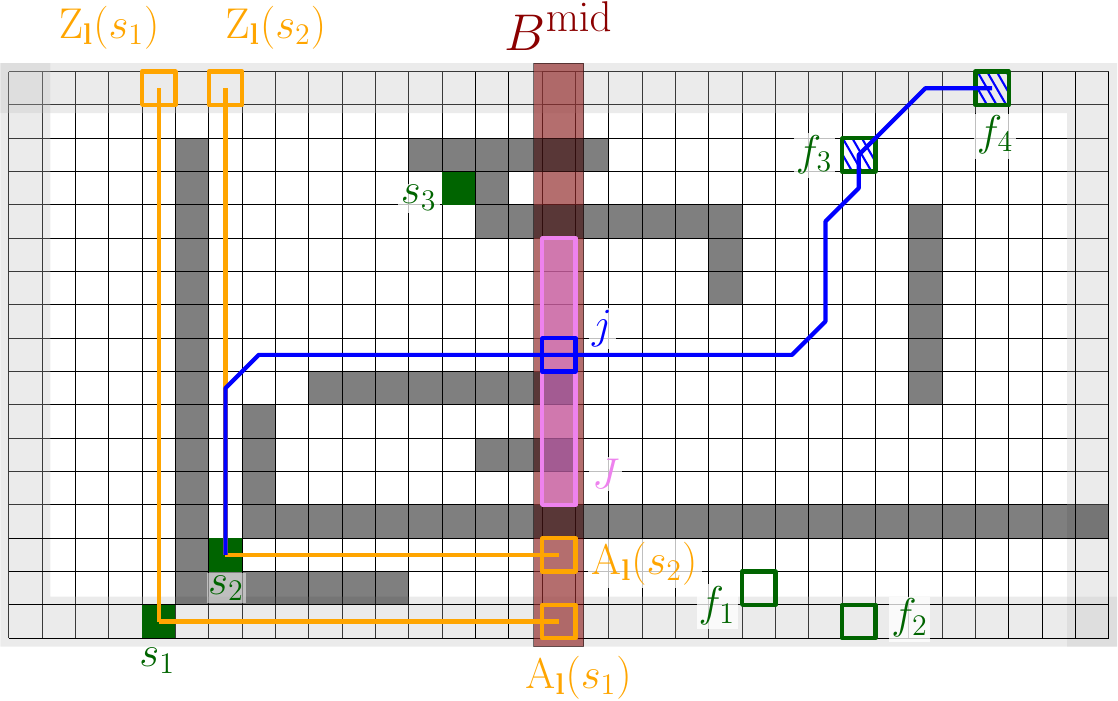}
\caption{Computation of $R^J$ for an $(S,F)$-reach equivalent interval $J$. Intuitively, we first determine, among indices in $J$ reachable from some $s\in S$, the index $j\in J$ with the best reachability towards $F$. We then determine all $f\in F$ reachable from $j$ (indicated by hatched boxes).}\label{fig:singlestepreach}
\end{figure}

Thus, after computing $\ell^J$, an orthogonal range reporting query can be used to report all $t\in F$ reachable from $S$ via $J$. To compute $\ell^J$, we observe that for any $j \in J$, we have
\[
\exists s \in S: \trav{s}{j}  \quad \stackrel{\text{Cor.\ \ref{cor:reachlabels}}}{\iff} \quad \exists s \in S: j \in \intvl_\lt(s), \ell_\lt(j) \le L(s) \quad \iff \quad \ell_\lt(j) \le \max_{\substack{s \in S,\\ j \in \intvl_\lt(s)}} L(s) =: L^j. \]
Noting (by $(S,F)$-reach-equivalence of $J$) that $j \in \intvl_\lt(s)$ if and only if $J \subseteq \intvl_\lt(s)$, we have that $L^j$ is independent of $j\in J$, and, in particular, equal to 
\begin{equation}
L^J \coloneqq \max_{\substack{s\in S,\\ J \subseteq \intvl_\lt(s)}} L(s), \label{eq:LJ}
\end{equation}
which can be computed by a single orthogonal range minimization query. Equipped with this value, we may determine $\ell^J$ as
\begin{equation}\label{eq:ellJ}
\ell^J = \min_{\substack{j \in J,\\\ell_\lt(j) \le L^J}} \ellrev_\rt(j).
\end{equation}
Note that given $\ell^J$, we may determine $R^J$ by a single orthogonal range reporting query; by ~\eqref{eq:RJ}.

We obtain the algorithm specified in Algorithm~\ref{alg:ssreach}, whose correctness we summarize as follows: in the $i$-th loop, we consider the  $i$-th $(S,F)$-reach-equivalent interval $J_i \eqqcolon [a_i, b_i]_{\Jmid}$ in the above partition of $\Jmid$. Observe that we determine $L^{J_i}$ according to its definition in~\eqref{eq:LJ}, and $\ell^{J_i}$ according to~\eqref{eq:ellJ}. Finally, we include in the set $R_i$ all elements of $R^{J_i}$ that have not yet been reported in a previous iteration, exploiting~\eqref{eq:RJ}. 

\begin{algorithm} 
\begin{algorithmic}[1]
\Function{$\singleStepReach$}{$B, S, F$}
\State Compute a partition of $\Jmid$ into $(S,F)$-reach-equivalent intervals $J_1, \dots, J_\ell$
\State Build $\OR_{S}$ storing $L(s)$ under the key $(\sA_\lt(s), \sZ_\lt(s))$ for all $s \in S$ (for maximization queries) 
\State Build $\OR_{F}$ storing $\ind(f)$ under the key $(\sArev_\rt(f), \sZrev_\rt(f), \Lrev(f))$ for all $f \in F$ \Statex \hfill (for decremental range reporting queries) 
\State \emph{Recall:} (precomputed) $\OR_B$ stores $\ellrev_2(q)$ under the key $(\ind(q), \ell_1(q))$ for all $q\in \Bmidfree$ \Statex \hfill (for minimization queries)
\For{$i= 1, \dots, \ell$} \Comment consider $J = [a_i,b_i]_{\Jmid}$
\State $L^J \gets \OR_S.\max((-\infty, a_i]\times [b_i, \infty))$
\State $\ell^J \gets \ORB.\min([a_i, b_i]\times (-\infty, L^J])$
\State $R_i \gets \OR_F.\mathrm{report}((-\infty, a_i] \times [b_i, \infty) \times [\ell^J, \infty))$
\State $\OR_F.\mathrm{delete}(R_i)$
\EndFor 
\State \Return $\bigcup_{i=1}^{\ell} R_i$
\EndFunction
\end{algorithmic}
\caption{Computing $\singleStepReach(B,S,F)$ for $B\in \blocks$, $S\subseteq B_\lt \setminus \Bmid,  F \subseteq B_\rt\setminus \Bmid$.}
\label{alg:ssreach}
\end{algorithm}

Let us analyze the running time.
Recall that $\ell = \Oh(|S \cup F|) = \Oh(|\term_B|)$. Thus, we can compute $J_1, \dots, J_\ell$ in time $\Oh(\ell \log \ell) = \Oh( |\term_B| \log n)$. Furthermore, as discussed in \secref{orthrange}, we can build $\OR_{S}$ in time $\Oh(|S| \log |S|) = \Oh(|\term_B| \log n)$ to support maximization queries in time $\Oh(\log |S|) = \Oh(\log n)$. Also, we can build $\OR_F$ in time $\Oh(|F|\log^2 |F|) = \Oh(|\term_B| \log^2 |\term_B|)$ to support deletions in time $\Oh(\log^2 |F|) = \Oh(\log^2 |\term_B|)$ and queries in time $\Oh(\log^2 |F| + k) = \Oh(\log^2 |\term_B| + k)$ where $k$ denotes the number of reported elements. Observe that $\ORB$ is already precomputed, as it belongs to the information stored at block $B$ (see Definition~\ref{def:blockinfo}).

We perform $\ell = \Oh(|\term_B|)$ iterations of the following form: First, we make a query to $\OR_S$ running in time $\Oh(\log n)$, followed by a query to $\OR_B$ running in time $\Oh(\log \sBbound) = \Oh(\log n)$. Then we obtain a set $R_i$ by a reporting query to $\OR_F$ running in time $\Oh(\log^2 |\term_B| + |R_i|)$. Afterwards, we delete all reported elements, which takes time $\Oh(|R_i| \log^2 |\term_B|)$. Thus, the total running time is bounded by $\Oh(|\term_B| \log n + \sum_{i=1}^\ell |R_i| \log^2 |\term_B|)$. Observe that we report each element in $\term_B$ at most once, which results in  $\sum_{i=1}^\ell |R_i| \le |\term_B|$. Hence, the total running time is bounded by $\Oh(|\term_B| (\log n + \log^2 |\term_B|)) = \Oh(|\term_B| \log^2 n)$, as desired.

\end{proof}

\newcommand{\abs}[1]{\left\vert #1 \right\vert}
\newcommand{\card}[1]{\left\vert #1 \right\vert}
\newcommand{\pione}{\ensuremath{\pi_F(v_1)}}
\newcommand{\sigmaone}{\ensuremath{\sigma_F(v_2)}}
\newcommand{\pitwo}{\ensuremath{\pi_{F'}(v_1)}}
\newcommand{\sigmatwo}{\ensuremath{\sigma_{F'}(v_2)}}
\newcommand{\pithree}{\ensuremath{\pi_G(v_3)}}
\newcommand{\sigmathree}{\ensuremath{\sigma_G(v_4)}}
\newcommand{\pifour}{\ensuremath{\pi_{G'}(v_3)}}
\newcommand{\sigmafour}{\ensuremath{\sigma_{G'}(v_4)}}
\newcommand{\piOR}{\ensuremath{\pi_\mathrm{OR}}}
\newcommand{\sigmaOR}{\ensuremath{\sigma_\mathrm{OR}}}
\newcommand{\deltar}{\ensuremath{2 + \frac{1}{4}\epsilon}}
\newcommand{\taurange}{\ensuremath{[-\frac{1}{4} \cdot \epsilon, (N^2 - \frac{3}{4}) \cdot \epsilon]}}
\newcommand{\taurangesqr}{\ensuremath{\taurange \times \taurange}}

\section{Conditional Lower Bound} \label{sec:lowerbound}

In this section we prove a lower bound of $n^{4-o(1)}$ for the discrete Fréchet distance under translation for two curves of length $n \in \mathbb{R}^2$ conditional on the Strong Exponential Time Hypothesis, or more precisely the $4$-OV Hypothesis. To this end, we reduce $4$-OV to the discrete Fréchet distance under translation.

Let us first have a closer look at $4$-OV. Given four sets of $N$ vectors $V_1, \dots, V_4 \subseteq \{0, 1\}^D$, the $4$-OV problem can be expressed as
\begin{equation} \label{eq:ov}
	\exists v_1 \in V_1, \dots, v_4 \in V_4 \;\forall j \in [D] \;\exists i \in \{1, \dots, 4\}: v_i[j] = 0.
\end{equation}
Recall from the introduction that we encode choosing the vectors $v_1, \dots, v_4$ by the canonical translation $\tau = (\tau_1, \tau_2) = (h_1 \cdot \epsilon, h_2 \cdot \epsilon)$ with $h_1, h_2 \in \{0, \dots, N^2 - 1\}$ for some constant $\epsilon > 0$ which is sufficiently small. To be concrete, let
\[
	\epsilon \coloneqq 0.001 / N^4
\]
for the remaining section. Choosing $v_1 \in V_1$ and $v_2 \in V_2$, we define
\[
h_1 \coloneqq h(v_1, v_2) \coloneqq \ind(v_1) + \ind(v_2) \cdot N,
\]
where $\ind(v_i)$ is the index of vector $v_i$ in the set $V_i$; similarly for $v_3 \in V_3, v_4 \in V_4$ we define $h_2 \coloneqq h(v_3,v_4)$. To perform the reduction, we want to construct two curves $\pi$ and $\sigma$ whose Fréchet distance decision for some $\delta$ is equivalent to the following expression, which is equivalent to (\ref{eq:ov}):
\begin{equation} \label{eq:reduction1}
	\exists \tau \in \mathbb{R}^2 \;\forall j \in [D] \;\exists i \in \{1, 2\}, v \in V_{2i-1}, v' \in V_{2i}: (v[j] = 0 \lor v'[j] = 0) \land (h(v,v') \cdot \epsilon = \tau_i).
\end{equation}
The expressions (\ref{eq:ov}) and (\ref{eq:reduction1}) are equivalent as the three quantifiers encode the same choices and we evaluate if there exists a zero in one of the chosen vectors. In (\ref{eq:reduction1}) we additionally need to make sure that the translation chosen by the outermost quantifier indeed is consistent with the vectors that are chosen by the innermost quantifier, which is done by requiring $h(v,v') \cdot \epsilon = \tau_i$.

We can further transform this expression to make it easier to create gadgets for the reduction:
\[
	\exists \tau \in [0,(N^2-1) \cdot \epsilon] \times [0,(N^2-1) \cdot \epsilon]\colon \;\bigwedge_{j \in [D]} \quad\quad \bigvee_{\mathclap{\substack{i \in \{1,2\} \\ v \in V_{2i-1}, v' \in V_{2i}: \\ v[j] = 0 \text{ or } v'[j] = 0}}} [h(v,v') \cdot \epsilon = \tau_i].
\]
According to this formula, we will construct gadgets. However, we cannot exactly ensure the equality $h(v,v') \cdot \epsilon = \tau_i$. Therefore, we resort to an approximate equality which still fulfills the intended usage of mapping translations to vector choices. The approximate values just snap to the closest canonical values. The gadgets we construct are the following:
\begin{itemize}
	\item \emph{Translation gadget:} It ensures that $\tau \in \taurangesqr$, \ie, we are always close to the points in the $\epsilon$-grid of translations that choose our vectors $v_1, \dots, v_4$.
	\item \emph{OV-dimension gadget:} AND over all $j \in [D]$.
	\item \emph{OR gadget:} The big OR in the formula.
	\item \emph{Equality gadget:} This gadget is only traversable if the two vectors it was created for correspond to $\tau$, \ie, it ensures that $h(v,v') \cdot \epsilon \approx \tau_i$.
\end{itemize}

We use the above mentioned gadgets as follows. The constructed curves $\pi$ and $\sigma$ start with the translation gadget consisting of the curves $\pi^{(0)}, \sigma^{(0)}$. They are followed by $D$ different parts that form the OV-dimension gadget. Each of the $D$ parts is an OR gadget and we call the respective curves $\pi^{(j)}$ and $\sigma^{(j)}$ for $j \in [D]$. Each of the OR gadgets $(\pi^{(j)}, \sigma^{(j)})$ contains several equality gadgets. We will use different variations of the equality gadget (one for each set of vectors $V_1, \dots, V_4$) but they are all of very similar structure. We need four different types of equality gadgets because for a certain $v_i \in V_i$ a part of the gadget is only inserted if $v_i[d] = 0$. Thus, if we traverse an equality gadget later, we know that it corresponds to one zero entry and also to the current translation. See Figure \ref{fig:whole_reduction} for an overview of the whole construction.

\begin{figure}
	\centering
	\includegraphics[width=\textwidth]{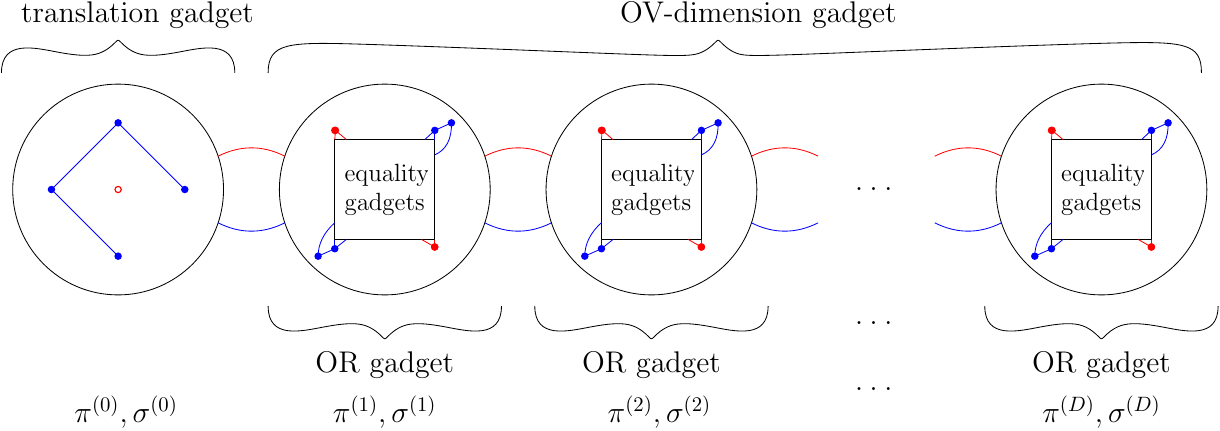}
	\caption{Overview of how the different gadgets are used in the curves that result from the reduction. We use one translation gadget, one OV-dimension gadget, $D$ OR gadgets, and $\Oh(ND)$ equality gadgets.}
	\label{fig:whole_reduction}
\end{figure}

Without loss of generality, assume that for all dimensions $j \in [D]$ at least one vector in $V_1 \cup \dots \cup V_4$ contains a $0$ in dimension $j$.
Now we give the detailed construction of the gadgets and the proofs of correctness. The instance of the discrete Fréchet distance under translation that we construct in the reduction uses a threshold distance of $\delta = \deltar$, \ie, we want to know for the constructed curves $\pi$ and $\sigma$ if their discrete Fréchet distance under translation is not more than $\delta$.

\begin{wrapfigure}{R}{0.25\textwidth}
	\vspace{-10pt}
	\centering
	\includegraphics[width=0.18\textwidth]{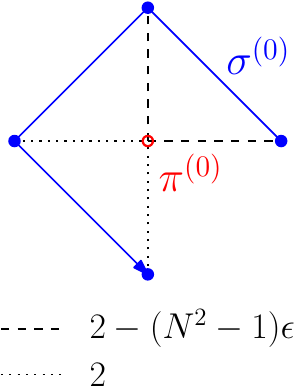}
	\caption{Translation gadget}
	\label{fig:translation_gadget}
	\vspace{-30pt}
\end{wrapfigure}

\paragraph{Translation Gadget.}
First we have to restrict the possible translations. To this end, we build a gadget to ensure $\tau \in \mathcal{T}$ for
\[
	\mathcal{T} \coloneqq \taurangesqr.
\]
This is realized by a gadget where curve $\pi^{(0)}$ consists of only one vertex and curve $\sigma^{(0)}$ consists of four vertices:
\begin{equation*}
	\hspace{-3cm} \begin{split}
	\pi^{(0)} &\coloneqq \langle (0,0) \rangle,\\
	\sigma^{(0)} &\coloneqq \langle (2-(N^2 - 1) \epsilon, 0), (0,2 - (N^2 - 1) \epsilon), (-2, 0), (0, -2) \rangle.
\end{split}
\end{equation*}
This gadget is sketched in Figure \ref{fig:translation_gadget}. 

\begin{lem} \label{lem:translation_gadget}
	Given two curves $\pi = \pi^{(0)} \circ \pi'$ and $\sigma = \sigma^{(0)} \circ \sigma'$ (with $\pi^{(0)}, \sigma^{(0)}$ as defined above), such that each $p \in \pi^{(0)}$ has distance greater than $10$ to each $p' \in \pi'$, the following holds:
\begin{enumerate}[label=(\roman*)]
	\item if $\tau \in [0, (N^2 - 1) \epsilon] \times [0, (N^2 - 1) \epsilon]$, then $\delta_F(\pi^{(0)}, \sigma^{(0)} + \tau) \leq \delta$
	\item if $\delta_F(\pi, \sigma + \tau) \leq \delta$, then $\tau \in \taurangesqr$
\end{enumerate}
\end{lem}
\begin{proof}
	We start with showing (i), so assume $\tau \in [0, (N^2 - 1) \epsilon] \times [0, (N^2 - 1) \epsilon]$. Note that the maximal distance $\max_{q \in \sigma^{(0)}} \max_\tau \norm{\pi^{(0)} - (q+\tau)}$ is an upper bound on $\delta_F(\pi^{(0)}, \sigma^{(0)}+\tau)$. By a simple calculation we obtain the desired result:
\[
	\max_{q \in \sigma^{(0)}} \max_\tau \norm{\pi^{(0)} - (q+\tau)} < \sqrt{2^2 + \epsilon^2 N^4} < \sqrt{2^2 + \epsilon + \frac{1}{16}\epsilon^2} = 2 + \frac{1}{4} \epsilon,
\]
where we used $\epsilon \leq N^{-4}$.
	
Now we prove (ii). Note that the start points of $\pi$ and $\sigma + \tau$ have to be in distance $\leq \delta$, thus $\tau \in [-5, 5]^2$ (using a very rough estimate). As all points of $\pi'$ are further than $10$ from any point in $\pi^{(0)}$ and thus all points on the postfix $\pi'$ are further than $\delta$ from $\sigma^{(0)}+\tau$, we have to stay in $\pi^{(0)}$ while traversing $\sigma^{(0)}$. Thus, the following inequalities hold for $\tau_i > (N^2 - \tfrac{3}{4})\epsilon$ or $\tau_i < -\tfrac{1}{4}\epsilon$ and $i \in \{1,2\}$ (where $\norm{v}_\infty$ denotes the infinity norm of $v$):
\[
	\delta_F(\pi, \sigma + \tau) \geq \delta_F(\pi^{(0)}, \sigma^{(0)} + \tau) \geq \max_{i \in [4]}\left\{\norm{\pi_1^{(0)} - (\sigma_i^{(0)} + \tau)}_\infty\right\} > \delta,
\]
which is the contrapositive of (ii).
\end{proof}

\paragraph{OV-dimension Gadget.}
For every $4$-OV dimension $j \in [D]$, we construct separate gadgets $\pi^{(1)}, \dots, \pi^{(D)}$ for $\pi$ and $\sigma^{(1)}, \dots, \sigma^{(D)}$ for $\sigma$. We want to connect these gadgets in a way that the two curves are in distance not more than $\delta$ if and only if all gadgets have distance not more than $\delta$ for a given translation $\tau$. This is done by simply placing the gadgets in distance greater than $\delta + N^2 \cdot \epsilon$ from each other and concatenating them.

\begin{lem} \label{lem:ov_dimension_gadget}
	Given a translation $\tau \in \mathcal{T}$ and curves $\pi = \pi^{(1)}, \dots, \pi^{(D)}$ and $\sigma = \sigma^{(1)}, \dots, \sigma^{(D)}$ where for all $j \in [D]$ all points of $\pi^{(j)}$ are further than $\delta + 2N^2 \cdot \epsilon$ from each point of $\sigma^{(j')}$ with $j \neq j'$, then $\delta(\pi, \sigma + \tau) \leq \delta$ if and only if $\delta_F(\pi^{(j)}, \sigma^{(j)} + \tau) \leq \delta$ for all $j \in [D]$.
\end{lem}
\begin{proof}
First, note that whatever $\tau$ we choose in the given range, $\sigma^{(j)} + \tau$ is still in distance greater than $\delta$ from every $\pi^{(j')}$ with $j' \neq j$.

Now, assume that for all $j \in [D]$ the curves $\pi^{(j)}, \sigma^{(j)} + \tau$ have distance at most $\delta$. Then we can traverse the gadgets in order and do simultaneous jumps between them. Thus, also the distance of the whole curves $\pi$ and $\sigma + \tau$ is at most $\delta$. For the other direction, assume that for at least one $j \in [D]$ the distance between $\pi^{(j)}$ and $\sigma^{(j)} + \tau$ is greater than $\delta$. On the one hand, if we do not traverse simultaneously (\ie, at one point the traversal is in $\pi^{(j)}$ and $\sigma^{(j')}$ for $j \neq j'$), then due to large distances of $\pi^{(j)}, \sigma^{(j')} + \tau$ for $j \neq j'$, we have distance greater than $\delta$ for this traversal. On the other hand, if we traverse $\pi^{(j)}$ and $\sigma^{(j)}$ together for all $j$, we also have distance greater than $\delta$ due to the gadget with distance greater than $\delta$.
\end{proof}

\noindent
For the remaining gadgets we define for convenience:
\[
	\eta \coloneqq 3 \cdot N^2 \epsilon.
\]

\paragraph{Equality Gadget.}
An equality gadget $F(v_1,v_2)$ for the vectors $v_1 \in V_1, v_2 \in V_2$ is a pair of two line segments, $\pione$ and $\sigmaone$, see Figure \ref{fig:vector_gadget}:
\begin{equation*}
	\begin{split}
		\pione &\coloneqq \langle (1 + \epsilon \cdot \ind(v_1), -1-\eta), (-1 + \epsilon \cdot \ind(v_1), 1+\eta) \rangle,
	\\
	\sigmaone &\coloneqq \langle (- 1 - \epsilon \cdot \ind(v_2) \cdot N, -1-\eta), (1 - \epsilon \cdot \ind(v_2) \cdot N, 1+\eta) \rangle.
	\end{split}
\end{equation*}
Note that this gives us $N^2$ different gadgets consisting of $2N$ different line segments. We later use the line segments $\pione$ in $\pi$ and the line segments $\sigmaone$ in $\sigma$ where they can be combined to form an equality gadget.

\begin{figure}
	\begin{subfigure}[b]{0.48\textwidth}
		\includegraphics[width=\textwidth]{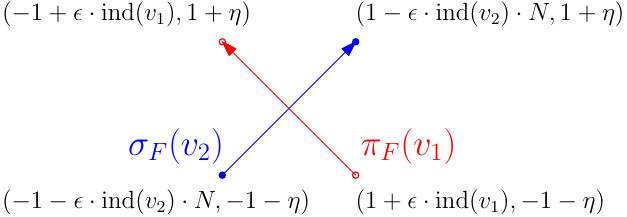}
		\caption{Equality Gadget for $V_1$ and $V_2$}
		\label{fig:vector_gadget}
	\end{subfigure}
	\hfill
	\begin{subfigure}[b]{0.48\textwidth}
		\includegraphics[width=\textwidth]{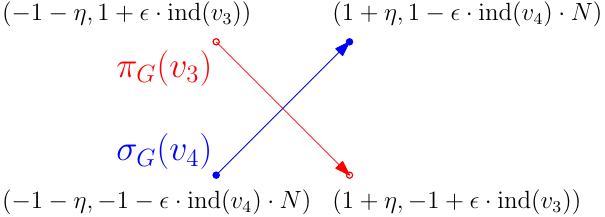}
		\caption{Equality Gadget for $V_3$ and $V_4$}
		\label{fig:vector_gadget_2}
	\end{subfigure}
	\caption{The equality gadgets for $F$ and $G$. The equality gadgets $F'$ and $G'$ are simply shifted.}
	\label{fig:vector_gadgets}
\end{figure}

\begin{lem} \label{lem:equality_gadget}
	Given curves $\pione, \sigmaone$ for some $v_1 \in V_1$ and $v_2 \in V_2$, and given a translation $\tau \in \mathcal{T}$, the following properties hold:
\begin{enumerate}[label=(\roman*)]
	\item if $\tau_1 = \epsilon \cdot (\ind(v_1) + \ind(v_2) \cdot N)$, then $\delta_F(\pione, \sigmaone + \tau) \leq \delta$
	\item if $\delta_F(\pione, \sigmaone + \tau) \leq \delta$, then $\abs{\epsilon \cdot (\ind(v_1) + \ind(v_2) \cdot N) - \tau_1} \leq \frac{1}{3}\epsilon$
\end{enumerate}
	
\end{lem}
\begin{proof}
	To prove (i), it suffices to give a valid traversal. We traverse $\pione = (p_1, p_2)$ and $\sigmaone = (q_1, q_2)$ simultaneously. Thus, we just want an upper bound on the distance between the (translated) first points $p_1,q_1+\tau$ and the distance between the (translated) second points $p_2,q_2+\tau$ to get an upper bound on $\delta_F(\pione, \sigmaone + \tau)$. These distances are
\begin{align*}
	\begin{alignedat}{3}
		\norm{p_1 - (q_1 + \tau)}^2 &= (2 + \epsilon \cdot \ind(v_1) + \epsilon \cdot \ind(v_2) \cdot N - \tau_1)^2 + \tau_2^2 &&= 4 + \tau_2^2 &&\leq 4 + \epsilon + \frac{1}{16}\epsilon^2 = \delta^2,\\
		\norm{p_2 - (q_2 + \tau)}^2 &= (-2 + \epsilon \cdot \ind(v_1) + \epsilon \cdot \ind(v_2) \cdot N - \tau_1)^2 + \tau_2^2 &&= 4 + \tau_2^2 &&\leq \delta^2,
	\end{alignedat}
\end{align*}
where we used $\abs{\tau_2} \leq N^2 \epsilon$ and thus $\tau_2^2 \leq N^4 \epsilon^2 \leq \epsilon$ since $\epsilon \leq N^{-4}$.
Both distances are at most $\delta$ and thus the discrete Fréchet distance is at most $\delta$ as well.

For proving (ii), first note that the first (respectively second) point of $\pione$ is far from the second (respectively first) point of $\sigmaone$, due to $\eta \geq N^2 \epsilon$. Thus, we have to traverse the gadget simultaneously. Let $\Delta \coloneqq \epsilon \cdot \ind(v_1) + \epsilon \cdot \ind(v_2) \cdot N - \tau_1$, it remains to show that $\Delta \leq \frac{1}{3} \epsilon$. For $p_1, q_1$ we then get
\[
\def\arraystretch{1.3}
\begin{array}{rrcl}
	\norm{p_1 - (q_1 + \tau)}^2  = & (2 + \epsilon \cdot \ind(v_1) + \epsilon \cdot \ind(v_2) \cdot N - \tau_1)^2 + \tau_2^2 & \leq & (2 + \frac{1}{4} \epsilon)^2 \\
		\Leftrightarrow & (2 + \Delta)^2 + \tau_2^2 & \leq & 4 + \epsilon + \frac{1}{16} \epsilon^2 \\
		\Leftrightarrow & 4 + 4\Delta + \Delta^2 + \tau_2^2 & \leq & 4 + \epsilon + \frac{1}{16} \epsilon^2 \\
		\Rightarrow & 4\Delta & \leq & \epsilon + \frac{1}{16} \epsilon^2\\
		\Rightarrow & \Delta & \leq & \frac{1}{4}\epsilon + \frac{1}{64}\epsilon^2 \leq \frac{1}{3} \epsilon.
\end{array}
\]
The last inequality follows from plugging in $\epsilon = 0.001/N^4$ and using the fact that $N \geq 1$.
With a similar calculation for $p_2, q_2$ we obtain that $\Delta \geq -\frac{1}{3}\epsilon$, and thus $\abs{\Delta} \leq \frac{1}{3}\epsilon$.
\end{proof}

Now we introduce three gadgets which have the same properties as the equality gadget but are slightly different. The aim is to have four types of gadgets which are pairwise further than a discrete Fréchet distance of $\delta$ apart such that we can use them together in one big OR expression.

\paragraph{Shifted Equality Gadget.}
As described in the introduction of this section, we want to use the curves $\pione, \sigmaone$ in case $v_1[j] = 0$ and we need an additional gadget for $v_2[j] = 0$. However, those two gadgets should not be too close such that the curves cannot be matched but also not too far such that the OR gadget (which we introduce later) still works. Thus, we introduce another gadget $F'(v_1,v_2)$ which consists of a pair of curves $\pitwo, \sigmatwo$ that are just shifted versions of $\pione, \sigmaone$; shifted by $N^2 \epsilon$ in the first dimension. More formally,
\begin{equation*}
\begin{split}
	\pitwo &\coloneqq \pione + (N^2 \epsilon, 0), \\
	\sigmatwo &\coloneqq \sigmaone + (N^2 \epsilon, 0).
\end{split}
\end{equation*}
Before proving the desired properties, we introduce the remaining two variants of the equality gadget.

\paragraph{Equality Gadget for $V_3$ and $V_4$.}
The above introduced equality gadgets only work for vectors in $V_1$ and $V_2$ but we also need a gadget for vectors in $V_3$ and $V_4$. Therefore, we introduce the gadget $G(v_3, v_4)$, which is a mirrored equality gadget consisting of the curves $\pithree$ and $\sigmathree$, see Figure \ref{fig:vector_gadget_2}:
\begin{equation*}
	\begin{split}
		\pithree &\coloneqq \langle (-1-\eta, 1 + \epsilon \cdot \ind(v_3)), (1+\eta, -1 + \epsilon \cdot \ind(v_3)) \rangle,
	\\
	\sigmathree &\coloneqq \langle (-1-\eta, - 1 - \epsilon \cdot \ind(v_4) \cdot N), (1+\eta, 1 - \epsilon \cdot \ind(v_4) \cdot N) \rangle.
	\end{split}
\end{equation*}

\paragraph{Shifted Equality Gadget for $V_3$ and $V_4$.} We define $G'(v_3,v_4)$ similarly to $F'(v_1,v_2)$, \ie, we shift the curves of $G$ by $N^2 \epsilon$, but in contrast to $F'$ we shift it in the second dimension. More formally:
\begin{equation*}
\begin{split}
	\pifour &\coloneqq \pithree + (0, N^2 \epsilon), \\
	\sigmafour &\coloneqq \sigmathree + (0, N^2 \epsilon).
\end{split}
\end{equation*}

\noindent
Due to the similar structure of the curve pairs of $F(v_1,v_2)$ and $F'(v_1,v_2), G(v_3,v_4), G'(v_3,v_4)$, analogous statements to Lemma \ref{lem:equality_gadget} also hold for the curve pairs from $F'(v_1,v_2)$, $G(v_3,v_4)$, and $G'(v_3,v_4)$. Specifically, we have:

\begin{lem} \label{lem:equality_gadget2}
	Given curves $\pitwo, \sigmatwo$ for some $v_1 \in V_1$ and $v_2 \in V_2$, and given a translation $\tau \in \mathcal{T}$, the following properties hold:
\begin{enumerate}[label=(\roman*)]
	\item if $\tau_1 = \epsilon \cdot (\ind(v_1) + \ind(v_2) \cdot N)$, then $\delta_F(\pitwo, \sigmatwo + \tau) \leq \delta$
	\item if $\delta_F(\pitwo, \sigmatwo + \tau) \leq \delta$, then $\abs{\epsilon \cdot (\ind(v_1) + \ind(v_2) \cdot N) - \tau_1} \leq \frac{1}{3}\epsilon$
\end{enumerate}
\end{lem}

\begin{lem} \label{lem:equality_gadget3}
	Given curves $\pithree, \sigmathree$ for some $v_3 \in V_3$ and $v_4 \in V_4$, and given a translation $\tau \in \mathcal{T}$, the following properties hold:
\begin{enumerate}[label=(\roman*)]
	\item if $\tau_2 = \epsilon \cdot (\ind(v_3) + \ind(v_4) \cdot N)$, then $\delta_F(\pithree, \sigmathree + \tau) \leq \delta$
	\item if $\delta_F(\pithree, \sigmathree + \tau) \leq \delta$, then $\abs{\epsilon \cdot (\ind(v_3) + \ind(v_4) \cdot N) - \tau_2} \leq \frac{1}{3}\epsilon$
\end{enumerate}
\end{lem}

\begin{lem} \label{lem:equality_gadget4}
	Given curves $\pifour, \sigmafour$ for some $v_3 \in V_3$ and $v_4 \in V_4$, and given a translation $\tau \in \mathcal{T}$, the following properties hold:
\begin{enumerate}[label=(\roman*)]
	\item if $\tau_2 = \epsilon \cdot (\ind(v_3) + \ind(v_4) \cdot N)$, then $\delta_F(\pifour, \sigmafour + \tau) \leq \delta$
	\item if $\delta_F(\pifour, \sigmafour + \tau) \leq \delta$, then $\abs{\epsilon \cdot (\ind(v_3) + \ind(v_4) \cdot N) - \tau_2} \leq \frac{1}{3}\epsilon$
\end{enumerate}
\end{lem}

\noindent
We now show that all subcurves of different equality gadgets are pairwise further apart than $\delta$. Here we say that the curves $\pione$ and $\sigmaone$ have \emph{type $F$}. Similarly, the other curves constructed above have type $F'$, $G$, or $G'$.

\begin{lem} \label{lem:equality_gadgets_far}
  For any vectors $v_1 \in V_1, \dots, v_4 \in V_4$ and any translation $\tau \in \mathcal{T}$, for any curves $\pi \in \{ \pione,\pitwo, \pithree, \pifour \}$ and $\sigma \in \{ \sigmaone, \sigmatwo, \sigmathree, \sigmafour \}$ of different type, we have $\delta_F(\pi, \sigma+\tau) > \delta$.
\end{lem}
\begin{proof}
	We first consider $\pitwo$ and $\sigmaone$. Consider the first point of $\sigmaone$ which we call~$q$. This point is further than $2 + N^2 \epsilon$ from both points of $\pitwo$. When translating $\sigma$ with $\tau \in \mathcal{T}$, the distance is still greater than $2 + \tfrac{3}{4} \epsilon > \delta$. Thus, $\sigmaone$ and $\pitwo$ are in discrete Fréchet distance greater than $\delta$ for any valid $\tau$.
	
	Similarly, consider $\pione$ and $\sigmatwo$, and let $p$ be the second point of $\pione$. The point $p$ has distance greater than $2 + \epsilon$ from $\sigmatwo$. With translation $\tau \in \mathcal{T}$ this distance is still greater than $2 + \tfrac{3}{4}\epsilon > \delta$ and thus $\pione$ and $\sigmatwo$ are in discrete Fréchet distance greater than $\delta$ for any valid~$\tau$.
The proof for types $G$ and $G'$ is symmetric.

Now we prove the lemma for types $F$ and $G$. First note that every point of $\pione$ is in distance $1+\eta$ of the first coordinate axis and every point of $\sigmathree$ is in distance $1+\eta$ of the second coordinate axis. Additionally, no point of $\pione$ is closer than $1 - 2 N^2 \epsilon$ to the second coordinate axis while no point of $\sigmathree$ is closer than $1 - 2 N^2 \epsilon$ to the first coordinate axis. This means that every point of $\pione$ is in distance at least $2 + \eta - 2 N^2 \epsilon = 2 + N^2 \epsilon$ of any point of $\sigmathree$. Even with translation this distance is at least $2 + \tfrac{3}{4} \epsilon > \delta$. Thus, also the discrete Fréchet distance is greater than $\delta$. The proofs for the remaining cases are symmetric.
\end{proof}

\noindent
We moreover observe that our equality gadgets lie in very restricted regions.
Specifically, call a curve \emph{diagonal} if all of its vertices are in $R_1 \cup R_2$ with
\[
	R_1 \coloneqq [-1 - 2\eta, -1 + 2\eta]^2, \; R_2 \coloneqq [1 - 2\eta, 1 + 2\eta]^2,
\]
and we call it \emph{anti-diagonal} if all of its vertices are contained in $R_3 \cup R_4$ with
\[
	R_3 \coloneqq [-1 - 2\eta, -1 + 2\eta] \times [1 - 2\eta, 1 + 2\eta], \; R_4 \coloneqq [1 - 2\eta, 1 + 2\eta] \times [-1 - 2\eta, -1 + 2\eta].
\]
See Figure \ref{fig:abstract_or_gadget}. Also, note that the order in which the curves visit the regions is not specified in the definition of (anti-)diagonal.
\begin{obs} \label{obs:diagonal}
    The constructed curves $\pione, \pitwo, \pithree, \pifour$ are anti-diagonal, and the curves $\sigmaone, \sigmatwo, \sigmathree, \sigmafour$ are diagonal.
\end{obs}
\begin{proof}
We observe that each coordinate of a vertex of any of these curves differs from $1$ or $-1$ by at most $\eps N^2 + \max\{\eta, \eps N^2\}$, by bounding $0 \le \ind(v_i) \le N$. Recalling $\eta = 3 \cdot N^2 \epsilon$, we have $\eps N^2 \le \eta$. Therefore, any coordinate differs from $1$ or $-1$ by at most $2\eta$, that is, each coordinate lies in $R_1 \cup R_2 \cup R_3 \cup R_4$. The general shape of being (anti-)diagonal can be inferred from Figure~\ref{fig:vector_gadgets}.
\end{proof}
We are now ready to describe the last gadget. For proving its correctness, we will essentially only use the diagonal and anti-diagonal property of the curves.

\paragraph{OR Gadget.} We construct an OR gadget over diagonal and anti-diagonal curves which we will later apply to equality gadgets. Before introducing the gadget itself, we define various auxiliary points whose meaning will become clear later. Here we keep notation close to \cite{Bringmann14}, although the details of our construction are quite different.
\begin{equation*}
	\begin{split}
		&s_1 \coloneqq \left(-\tfrac{1}{4}, -\tfrac{1}{4}\right),\; t_1 \coloneqq \left(\tfrac{1}{4}, \tfrac{1}{4}\right),\; r_1 \coloneqq \left(\tfrac{99}{100}, -\tfrac{5}{4}\right),\; r_1' \coloneqq \left(-\tfrac{99}{100}, \tfrac{5}{4}\right),\\
		&s_2 \coloneqq (0,0),\; s_2^* \coloneqq \left(-\tfrac{3}{2}, -\tfrac{3}{2}\right),\; t_2^* \coloneqq \left(\tfrac{3}{2}, \tfrac{3}{2}\right),\; t_2 \coloneqq (0,0),\; r_2 \coloneqq \left(-\tfrac{99}{100}, -\tfrac{5}{4}\right),\; r_2' \coloneqq \left(\tfrac{99}{100}, \tfrac{5}{4}\right).
	\end{split}
\end{equation*}

\begin{figure}
	\centering
	\includegraphics[width=.8\textwidth]{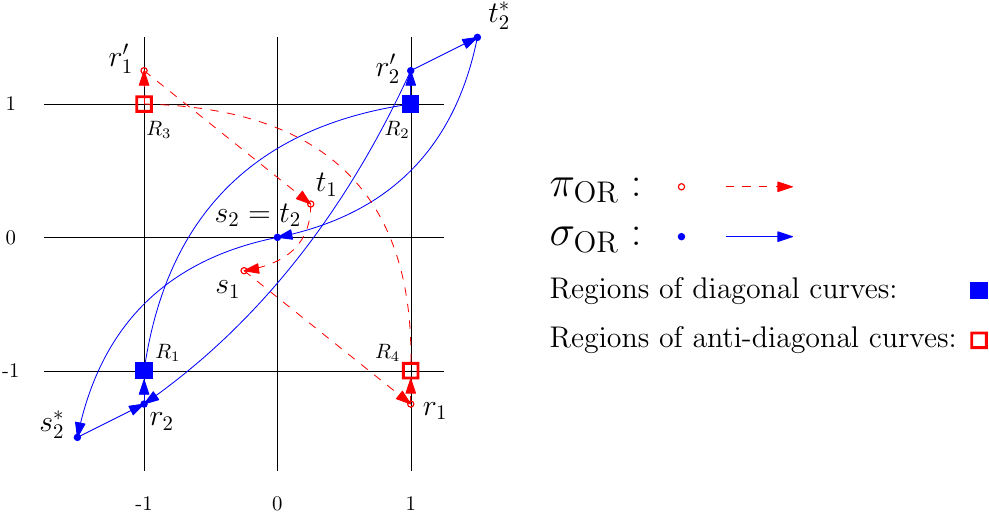}
	\caption{The OR gadget for general diagonal and anti-diagonal curves.}
	\label{fig:abstract_or_gadget}
\end{figure}

\noindent
Now, given diagonal curves $\hat{\sigma}^1, \dots, \hat{\sigma}^\ell$ and anti-diagonal curves $\hat{\pi}^1, \dots, \hat{\pi}^k$, we define the two curves of the OR gadget as
\begin{equation*}
\begin{split}
	\piOR &\coloneqq \mathop{\bigcirc}\limits_{i \in [k]} s_1 \circ r_1 \circ \hat{\pi}^i \circ r_1' \circ t_1,\\
	\sigmaOR &\coloneqq s_2 \circ s_2^* \circ (\mathop{\bigcirc}_{j \in [\ell]} r_2 \circ \hat{\sigma}^j \circ r_2') \circ t_2^* \circ t_2.
\end{split}
\end{equation*}
See Figure \ref{fig:abstract_or_gadget} for a visualization. Now let us prove correctness of the gadget.
\begin{lem} \label{lem:abstract_or_gadget}
	Given an OR gadget over diagonal curves $\hat{\sigma}^1, \dots, \hat{\sigma}^\ell$ and anti-diagonal curves $\hat{\pi}^1, \dots, \hat{\pi}^k$, for any translation $\tau \in \mathcal{T}$ we have $\delta_F(\piOR, \sigmaOR + \tau) \leq \delta$ if and only if $\delta_F(\hat{\pi}^i, \hat{\sigma}^j + \tau) \leq \delta$ for some $i,j$.
\end{lem}
\begin{proof}
	We first observe that for none of the auxiliary points $p \in \piOR$ and $q \in \sigmaOR$ we have that $\norm{p - q} \in [1.99, 2.01]$. This can be verified by calculating all distances, but we omit this due to readability of the proof. Also observe that $\mathcal{T} \subset [-0.001, 0.001]$ and $\delta \in [2, 2.001]$.
	It follows from the above observations that the translation $\tau$ does not change whether auxiliary points are closer than $\delta$ or not. Thus, we can ignore the translation for distances between auxiliary points in this proof. For reference we state which auxiliary points are closer than $\delta$ for all $\tau \in \mathcal{T}$. For each auxiliary point in \piOR{} we list its close auxiliary points in \sigmaOR:
\begin{equation*}
	\begin{split}
		s_1:& \quad s_2, s_2^*, t_2, r_2, r_2', \\
		t_1:& \quad s_2, t_2, r_2, r_2', \\
		r_1:& \quad s_2, t_2, r_2, \\
		r_1':& \quad s_2, t_2, r_2'. \\
	\end{split}
\end{equation*}
All other pairs are in distance greater than $\delta$. Note that for the remainder of the proof, we do not have to consider the specific value for $\tau$ anymore.

We first show that if $\delta_F(\hat{\pi}^i, \hat{\sigma}^j + \tau) \leq \delta$ for some $i,j$, then $\delta_F(\piOR, \sigmaOR + \tau) \leq \delta$ by giving a valid traversal. We start in $s_1, s_2+\tau$. Then we traverse $\piOR$ until the copy of $s_1$ which comes before the subcurve $\hat{\pi}^i$. While staying in $s_1$, we traverse $\sigmaOR+\tau$ until we reach the copy of $r_2+\tau$ right before the subcurve $\hat{\sigma}^j + \tau$. We then do one step on $\piOR$ to $r_1$. Now we step to the first nodes of $\hat{\pi}^i$ and $\hat{\sigma}^j + \tau$ simultaneously, and then traverse these two subcurves in distance $\delta$, which is possible due to $\delta_F(\hat{\pi}^i, \hat{\sigma}^j + \tau) \leq \delta$. We then step to the copies of $r_1'$ and $r_2' + \tau$ simultaneously.
We then step to $t_1$ on $\piOR$, while staying at $r_2' + \tau$ in $\sigmaOR + \tau$.
Subsequently, while staying in $t_1$, we traverse $\sigmaOR + \tau$ until we reach its last point, namely $t_2 + \tau$. Now we can traverse the remainder of $\piOR$. One can check that this traversal stays within distance $\delta$.

We now show that if $\delta_F(\piOR, \sigmaOR + \tau) \leq \delta$, then there exist $i, j$ such that $\delta_F(\hat{\pi}^i, \hat{\sigma}^j + \tau) \leq \delta$. Pick any valid traversal for which $\delta_F(\piOR, \sigmaOR + \tau) \leq \delta$. We reconstruct in the following how it passed through $\piOR$ and $\sigmaOR + \tau$. Consider the point when $s_2^* + \tau$ is reached. At that point, we have to be in some copy of $s_1$ as this is the only type of node of $\piOR$ which is in distance at most $\delta$ from $s_2^* + \tau$. Let $\hat{\pi}^i$ be the subcurve right after this copy of $s_1$. When we step to the copy of $r_1$ right after this $s_1$, there are only three types of nodes from $\sigmaOR + \tau$ in distance $\delta$: $s_2 + \tau, t_2 + \tau, r_2 + \tau$. Note that we already passed $s_2 + \tau$, and we cannot have reached $t_2 + \tau$ yet, as $t_2^* + \tau$ is neither in reach of $s_1$ nor $r_1$. Thus, we are in $r_2 + \tau$. Let the curve right after $r_2 + \tau$ be $\hat{\sigma}^j + \tau$. The only option now is to do a simultaneous step to the first nodes of $\hat{\pi}^i$ and $\hat{\sigma}^j + \tau$. Now, consider the point when either $r_1'$ or $r_2' + \tau$ is first reached. All points of $\hat{\pi}^i$ are far from $r_2' + \tau$ and all points of $\hat{\sigma}^j + \tau$ are far from $r_1'$ and thus we have to be in $r_1'$ and $r_2' + \tau$ at the same time. This implies that we traversed $\hat{\pi}^i$ and $\hat{\sigma}^j + \tau$ from the start to the end nodes in distance $\delta$ and therefore $\delta_F(\hat{\pi}^i, \hat{\sigma}^j + \tau) \leq \delta$.
\end{proof}

\paragraph{Assembling $\pi^{(j)}$ and $\sigma^{(j)}$.} Now we can apply the OR gadget to the equality gadgets in the following way. For each of the $D$ dimensions we construct an OR gadget. The OR gadget for dimension $j \in [D]$ contains as anti-diagonal curves all $\pione$ with $v_1[j] = 0$, all $\pitwo$, all $\pithree$ with $v_3[j] = 0$, and all $\pifour$; and as diagonal curves it contains all $\sigmaone$, all $\sigmatwo$ with $v_2[j] = 0$, all $\sigmathree$, and all $\sigmafour$ with $v_4[j] = 0$. Note that these curves fulfill the requirements stated in Observation \ref{obs:diagonal} for usage in the OR gadget as (anti-)diagonal curves. We denote the resulting curves by $\pi^{(j)}$ and $\sigma^{(j)}$, and we write $H(j) = (\pi^{(j)}, \sigma^{(j)})$. This yields the following lemma.

\begin{lem} \label{lem:or_gadget}
	Given a $4$-OV instance $V_1, \dots, V_4$, and consider the corresponding OR gadget $H(j) = (\pi^{(j)}, \sigma^{(j)})$ for some $j \in [D]$. It holds that:
\begin{enumerate}[label=(\roman*)]
	\item For any vectors $v_1 \in V_1, \dots, v_4 \in V_4$ with $v_1[j] \cdot v_2[j] \cdot v_3[j] \cdot v_4[j] = 0$ we have $\delta_F(\pi^{(j)}, \sigma^{(j)} + \tau) \leq \delta$ for $\tau = ((\ind(v_1) + \ind(v_2) \cdot N) \cdot \epsilon, (\ind(v_3) + \ind(v_4) \cdot N) \cdot \epsilon)$.

	\item If $\delta_F(\pi^{(j)}, \sigma^{(j)} + \tau) \leq \delta$ for some $\tau \in \mathcal{T}$, then
	\begin{itemize}
		\item $\exists v_1 \in V_1, v_2 \in V_2: v_1[j] \cdot v_2[j] = 0 \enspace \text{and} \enspace \abs{\epsilon \cdot (\ind(v_1) + \ind(v_2) \cdot N) - \tau_1} \leq \frac{1}{3}\epsilon$
\item[] or
		\item $\exists v_3 \in V_3, v_4 \in V_4: v_3[j] \cdot v_4[j] = 0 \enspace \text{and} \enspace \abs{\epsilon \cdot (\ind(v_3) + \ind(v_4) \cdot N) - \tau_2} \leq \frac{1}{3}\epsilon$
	\end{itemize}
\end{enumerate}
\end{lem}
\begin{proof}
	For (i), from $v_1[j] \cdot v_2[j] \cdot v_3[j] \cdot v_4[j] = 0$ it follows that at least one gadget of $F(v_1, v_2)$, $F'(v_1, v_2)$, $G(v_3, v_4)$, $G'(v_3, v_4)$ is contained in $H(j)$. By Lemmas \ref{lem:equality_gadget} to \ref{lem:equality_gadget4}, we know that the discrete Fréchet distance of this gadget is small. By Lemma \ref{lem:abstract_or_gadget} it then follows that $\delta_F(\pi^{(j)}, \sigma^{(j)} + \tau) \leq \delta$.

	For (ii), from $\delta_F(\pi^{(j)}, \sigma^{(j)} + \tau) \leq \delta$ it follows by Lemmas \ref{lem:abstract_or_gadget} and \ref{lem:equality_gadgets_far} that there exists a gadget $\Gamma$ for which the discrete Fréchet distance is at most $\delta$. From Lemmas \ref{lem:equality_gadget} to \ref{lem:equality_gadget4} it then follows that
\[
	\abs{\epsilon \cdot (\ind(v_1) + \ind(v_2) \cdot N) - \tau_1} \leq \frac{1}{3}\epsilon \quad \text{or} \quad \abs{\epsilon \cdot (\ind(v_3) + \ind(v_4) \cdot N) - \tau_2} \leq \frac{1}{3}\epsilon.
\]
for some vectors $v_1 \in V_1, \dots, v_4 \in V_4$. As $\Gamma$ is contained in the OR gadget, we additionally have that $v_1[j] \cdot v_2[j] = 0$ or $v_3[j] \cdot v_4[j] = 0$, respectively.
\end{proof}

\paragraph{Final Curves.} The final curves $\pi$ and $\sigma$ are now defined as follows. We start with the translation gadget $\pi^{(0)}$ ($\sigma^{(0)}$). Then the curves $\pi^{(j)}$ ($\sigma^{(j)}$) follow for $j \in [D]$. Note that we have to translate these curves to fulfill the requirements of Lemmas \ref{lem:translation_gadget} and \ref{lem:ov_dimension_gadget}, thus, we translate $\pi^{(j)}$ ($\sigma^{(j)}$) by $(100 \cdot j, 0)$. More explicitly, the final curves are
\begin{equation*}
	\begin{split}
		\pi &\coloneqq \pi^{(0)} \bigcirc_{j \in [D]} \pi^{(j)} + (100 \cdot j, 0), \\
		\sigma &\coloneqq \sigma^{(0)} \bigcirc_{j \in [D]} \sigma^{(j)} + (100 \cdot j, 0).
	\end{split}
\end{equation*}

We are now ready to prove Theorem \ref{thm:mainlower}. Recall its statement:
\mainlower*

\noindent
Also recall that it suffices to prove a lower bound under the $4$-OV hypothesis. For clarity of structure, we split the proof into Lemma \ref{lem:main1} and Lemma \ref{lem:main2} which together imply Theorem~\ref{thm:mainlower}.

\begin{lem} \label{lem:main1}
	Given a YES-instance of $4$-OV, the curves $\pi$ and $\sigma$ constructed in our reduction have discrete Fréchet distance under translation at most $\delta$, \ie, $\min_\tau \delta_F(\pi,\sigma + \tau) \leq \delta$.
\end{lem}
\begin{proof}
	Let $v_1 \in V_1, \dots, v_4 \in V_4$ be orthogonal vectors and let $\tau = ((\ind(v_1) + \ind(v_2) \cdot N) \cdot \epsilon, (\ind(v_3) + \ind(v_4) \cdot N) \cdot \epsilon)$ be the translation corresponding to those vectors. From Lemma \ref{lem:translation_gadget} we know that $\delta_F(\pi^{(0)}, \sigma^{(0)}+\tau) \leq \delta$, and thus there is a valid traversal to the endpoints of the translation gadget. Then we simultaneously step to the start of $\pi^{(1)}$ and $\sigma^{(1)}$. From Lemma \ref{lem:or_gadget} we know that there also exist traversals of $\pi^{(1)}, \dots, \pi^{(D)}$ and $\sigma^{(1)}+\tau, \dots, \sigma^{(D)}+\tau$ of distance at most $\delta$. It follows from Lemma \ref{lem:ov_dimension_gadget} that we can also traverse those gadgets sequentially in distance $\delta$ and thus $\delta_F(\pi,\sigma + \tau) \leq \delta$.
\end{proof}

\begin{lem} \label{lem:main2}
If the curves $\pi$ and $\sigma$ constructed in our reduction have discrete Fréchet distance under translation at most $\delta$, then the given $4$-OV instance is a YES-instance.
\end{lem}
\begin{proof}
	Let $\tau$ be a translation such that $\delta_F(\pi,\sigma+\tau) \leq \delta$. We know from Lemma \ref{lem:translation_gadget} that $\tau \in \mathcal{T}$. Furthermore, from Lemma \ref{lem:ov_dimension_gadget} we know that for all $j \in [D]$ it holds that $\delta_F(\pi^{(j)}, \sigma^{(j)}+\tau) \leq \delta$. It follows from Lemma \ref{lem:or_gadget} that for every $j \in [D]$ there exist $v_1 \in V_1, v_2 \in V_2$ such that $v_1[j] \cdot v_2[j] = 0$ and $\abs{\epsilon \cdot (\ind(v_1) + \ind(v_2) \cdot N) - \tau_1} \leq \frac{1}{3}\epsilon$ or there exist $v_3 \in V_3, v_4 \in V_4$ such that $v_3[j] \cdot v_4[j] = 0$ and $\abs{\epsilon \cdot (\ind(v_3) + \ind(v_4) \cdot N) - \tau_2} \leq \frac{1}{3}\epsilon$.
	Therefore, every dimension $j \in [D]$ gives us constraints on either $v_1, v_2$ or $v_3, v_4$. Due to Lemma \ref{lem:or_gadget} these constraints have to be consistent. If in total this gives us constraints for $v_1, \dots, v_4$, then we are done. Otherwise, if this only gives us constraints for $v_1, v_2$, then we already found $v_1, v_2$ which are orthogonal and thus we can pick arbitrary $v_3 \in V_3, v_4 \in V_4$ to obtain an orthogonal set of vectors. The case of only $v_3, v_4$ being constrained is symmetric.
\end{proof}

\begin{proof}[Proof of Theorem \ref{thm:mainlower}.]
	The Strong Exponential Time Hypothesis implies the $k$-OV hypothesis.
	The reduction above from a $4$-OV instance of size $N$ over $\{0, 1\}^D$ to an instance of the discrete Fréchet distance under translation in $\mathbb{R}^2$ results in two curves of length $\mathcal{O}(D \cdot N)$. Lemmas \ref{lem:main1} and \ref{lem:main2} show correctness of this reduction. Hence, any $\mathcal{O}(n^{4-\epsilon})$-time algorithm for the discrete Fréchet distance under translation would imply an algorithm for $4$-OV in time $\mathcal{O}((D \cdot N)^{4-\epsilon}) = \mathcal{O}(\text{poly}(D) \cdot N^{4-\epsilon})$, refuting the $k$-OV hypothesis.
\end{proof}
 \section{Conclusion}

In this work, we designed an improved algorithm for the discrete Fréchet distance under translation running in time $\tOh(n^{14/3}) = \tOh(n^{4.66...})$. As a crucial subroutine, we developed an improved algorithm for offline dynamic grid reachability. Additionally, we presented a conditional lower bound of $n^{4-o(1)}$ based on the Strong Exponential Time Hypothesis, which, despite not yet matching our upper bound, strongly separates the discrete Fréchet distance under translation from the standard discrete \Fr distance.

Our use of offline dynamic grid reachability yields further motivation for studying the offline setting of dynamic algorithms, for potential use as subroutines in static algorithms. 
Problems left open by this paper include: (1) Closing the gap between our upper and conditional lower bound. This might require a solution to offline dynamic grid reachability with polylogarithmic amortized update time. (2) Generalizing our bounds to $d=1$ or higher dimensions $d\ge 3$, as in this paper we only considered curves in the plane. While generalizing our algorithm to $d=1$ or $d \ge 3$ seems rather straight-forward but technical, obtaining strong conditional lower bounds for these cases is more interesting. (3) Considering different transformations such as scaling, rotation, or affine transformations in general; here we only treated translations. Significantly new ideas seem necessary to obtain meaningful lower bounds for other transformations. 
(4) Determine whether the time complexity of variants of the discrete Fréchet distance, such as the continuous or weak Fréchet distance, have similar or different relationships to their translation-invariant analogues. 
 
\bibliographystyle{alpha}
\bibliography{biblio}

\end{document}